\documentclass[fontsize=12,english,british,reprint,nocopyrightspace,mathic=true,xypdf,abstract=true,obeyfinal,final]{scrartcl}
\PassOptionsToPackage{natbib=true}{biblatex}
\PassOptionsToPackage{obeyFinal}{todonotes}
\usepackage[T1]{fontenc}
\usepackage[utf8]{inputenc}
\usepackage{color}
\usepackage{babel}
\usepackage{prettyref}
\usepackage{xr}
\usepackage{mathtools}
\usepackage{enumitem}
\usepackage{todonotes}
\usepackage{varwidth}
\usepackage{amsmath}
\usepackage{amsthm}
\usepackage{geometry}
\usepackage{xargs}[2008/03/08]
\usepackage[all]{xy}
\usepackage[pdfusetitle,
 bookmarks=true,bookmarksnumbered=false,bookmarksopen=false,
 breaklinks=true,pdfborder={0 0 0},pdfborderstyle={},backref=false,colorlinks=true]
 {hyperref}
\hypersetup{
 ocgcolorlinks,linkcolor={dark-red},citecolor={dark-blue},urlcolor={blue}}

\makeatletter

\g@addto@macro\bfseries{\boldmath}
\usepackage[fixxits,nomathstyle,realVert]{perfectcut}
\usepackage{ifxetex}
\newif\iffontspec
\ifxetex\fontspectrue\else\fontspecfalse\fi

\usepackage{trimclip}
\usepackage{mleftright}
\mleftright

\makeatother
\def\Xymatrixone#1{\xymatrix@1}
\makeatletter

\providecommand{\tabularnewline}{\\}

      \@ifundefined{mdframed}{\usepackage{mdframed}}{}
\newenvironment{boiteombree}
{\begin{mdframed}[innerleftmargin=0pt,innerrightmargin=0pt,shadow=true,shadowsize=2pt,shadowcolor=black!8]\vspace*{-1.5ex}}
{\vspace*{1ex}\end{mdframed}\vspace*{-2.5ex}}
\@ifundefined{varwidth}{\usepackage{varwidth}}{}

\newrefformat{lem}{Lemma \ref{#1}}
\newrefformat{thm}{Theorem \ref{#1}}
\newrefformat{prop}{Proposition \ref{#1}}
\newrefformat{cor}{Corollary \ref{#1}}
\newrefformat{defn}{Definition \ref{#1}}

\usepackage{tikz}
\usepackage{tikz-cd}
\usetikzlibrary{matrix}
\usetikzlibrary{arrows}
\usetikzlibrary{arrows.meta}
\usetikzlibrary{decorations.pathmorphing,shapes}
\usetikzlibrary{decorations.markings}

\tikzset{spanmap/.style={
            decoration={markings,
            mark= at position 0.5 with {
                  \node[transform shape] (tempnode) {$|$};
}
              },
              postaction={decorate}
}
}

\tikzcdset{every arrow/.append style = -{Computer Modern Rightarrow[scale=.7]}}

\usepackage{caption}
\usepackage{cprotect}
\usepackage{etoolbox}
\PassOptionsToPackage{caption=false}{subfig}

\usepackage{adjustbox}
\usepackage{mdframed}
\usepackage{xcolor}
  \definecolor{dark-red}{rgb}{0.6,0,0}
  \definecolor{dark-blue}{rgb}{0,0.0,0.6}
  \definecolor{blue}{rgb}{0,0,0.8}

\usepackage[expansion=false]{microtype}

\usepackage{mlmodern}
\usepackage{fdsymbol}\usepackage[oldbold,scr=boondox,scrscaled=1.07]{mathalfa}

\usepackage{comment}

\usepackage{etex}
\usepackage{bussproofs}
\usepackage{mleftright}
\mleftright
\usepackage[realVert]{perfectcut}
\usepackage[all]{xy}
\SelectTips{eu}{10}\UseTips

\hypersetup{unicode,psdextra}

\proofSkipAmount{\vskip1ex plus.8ex minus.4ex}

\usepackage{multirow}

\def\forceAfour{
  \setlength{\hoffset}{0in}
  \setlength{\voffset}{0in}
\KOMAoptions{paper=A4,DIV=18}
\recalctypearea
}
\forceAfour{}

\g@addto@macro\bfseries{\boldmath}

\usepackage{enumitem}

\usepackage{flushend}

\usepackage[expansion=false]{microtype}

\AtBeginDocument{
\setcounter{biburllcpenalty}{1000}
\setcounter{biburlucpenalty}{2000}
\setcounter{biburlnumpenalty}{200}

\DeclareFieldFormat[article]{number}{\mkbibparens{#1}}
\AtBeginBibliography{
\sloppy
}
}

\ifdefined\showcaptionsetup
\PassOptionsToPackage{caption=false}{subfig}
\fi
\usepackage{subfig}
\makeatother

\providecommand\theoremname{Theorem}
\theoremstyle{plain}
\newtheorem{thm}{\protect\theoremname}
\providecommand\propositionname{Proposition}
\newtheorem{prop}[thm]{\protect\propositionname}
\providecommand\definitionname{Definition}
\theoremstyle{definition}
\newtheorem{defn}[thm]{\protect\definitionname}
\providecommand\examplename{Example}
\newtheorem{example}[thm]{\protect\examplename}
\providecommand\corollaryname{Corollary}
\theoremstyle{plain}
\newtheorem{cor}[thm]{\protect\corollaryname}
\providecommand\lemmaname{Lemma}
\newtheorem{lem}[thm]{\protect\lemmaname}
\providecommand\remarkname{Remark}
\theoremstyle{remark}
\newtheorem{rem}[thm]{\protect\remarkname}
\usepackage[bibstyle=ext-authoryear,citestyle=ext-authoryear-comp,sortcites=false,date=year,maxcitenames=3, mincitenames=2,maxbibnames=10, minbibnames=10]{biblatex}
\addto\captionsbritish{\newrefformat{cor}{Corollary \ref{#1}}}
\addto\captionsbritish{\newrefformat{def}{Definition \ref{#1}}}
\addto\captionsbritish{\newrefformat{enu}{Item~\ref{#1}}}
\addto\captionsbritish{\newrefformat{exa}{Example \ref{#1}}}
\addto\captionsbritish{\newrefformat{lem}{Lemma \ref{#1}}}
\addto\captionsbritish{\newrefformat{par}{Paragraph~\ref{#1}}}
\addto\captionsbritish{\newrefformat{prop}{Proposition \ref{#1}}}
\addto\captionsbritish{\newrefformat{rem}{Remark \ref{#1}}}
\addto\captionsbritish{\newrefformat{sec}{Section~\ref{#1}}}
\addto\captionsbritish{\newrefformat{subsec}{Section~\ref{#1}}}
\addto\captionsbritish{\newrefformat{thm}{Theorem \ref{#1}}}
\addto\captionsbritish{\renewcommand{\corollaryname}{Corollary}}
\addto\captionsbritish{\renewcommand{\definitionname}{Definition}}
\addto\captionsbritish{\renewcommand{\examplename}{Example}}
\addto\captionsbritish{\renewcommand{\lemmaname}{Lemma}}
\addto\captionsbritish{\renewcommand{\propositionname}{Proposition}}
\addto\captionsbritish{\renewcommand{\remarkname}{Remark}}
\addto\captionsbritish{\renewcommand{\theoremname}{Theorem}}
\addto\captionsenglish{\newrefformat{cor}{Corollary \ref{#1}}}
\addto\captionsenglish{\newrefformat{def}{Definition \ref{#1}}}
\addto\captionsenglish{\newrefformat{enu}{Item~\ref{#1}}}
\addto\captionsenglish{\newrefformat{exa}{Example \ref{#1}}}
\addto\captionsenglish{\newrefformat{fn}{Footnote~\ref{#1}}}
\addto\captionsenglish{\newrefformat{lem}{Lemma \ref{#1}}}
\addto\captionsenglish{\newrefformat{par}{Paragraph~\ref{#1}}}
\addto\captionsenglish{\newrefformat{prop}{Proposition \ref{#1}}}
\addto\captionsenglish{\newrefformat{rem}{Remark \ref{#1}}}
\addto\captionsenglish{\newrefformat{sec}{Section~\ref{#1}}}
\addto\captionsenglish{\newrefformat{subsec}{Section~\ref{#1}}}
\addto\captionsenglish{\newrefformat{thm}{Theorem \ref{#1}}}
\addto\captionsenglish{\renewcommand{\corollaryname}{Corollary}}
\addto\captionsenglish{\renewcommand{\definitionname}{Definition}}
\addto\captionsenglish{\renewcommand{\examplename}{Example}}
\addto\captionsenglish{\renewcommand{\lemmaname}{Lemma}}
\addto\captionsenglish{\renewcommand{\propositionname}{Proposition}}
\addto\captionsenglish{\renewcommand{\remarkname}{Remark}}
\addto\captionsenglish{\renewcommand{\theoremname}{Theorem}}
\newrefformat{cor}{Corollary \ref{#1}}
\newrefformat{def}{Definition \ref{#1}}
\newrefformat{enu}{Item~\ref{#1}}
\newrefformat{exa}{Example \ref{#1}}
\newrefformat{fig}{\figurename~\ref{#1}}
\newrefformat{fn}{Footnote~\ref{#1}}
\newrefformat{lem}{Lemma \ref{#1}}
\newrefformat{par}{Paragraph~\ref{#1}}
\newrefformat{prop}{Proposition \ref{#1}}
\newrefformat{rem}{Remark \ref{#1}}
\newrefformat{sec}{Section~\ref{#1}}
\newrefformat{subsec}{Section~\ref{#1}}
\newrefformat{thm}{Theorem \ref{#1}}

\begin{document}
\selectlanguage{english}\global\long\def\proofsep{\mathbin{\vdash}}

\global\long\def\fCenter{\proofsep}

\global\long\def\prooflabelformat#1{{\scriptstyle \textrm{#1}}}

\global\long\def\binaryprimitive#1#2{\BinaryInf$#1\fCenter#2$}

\global\long\def\axiomprimitive#1#2{\Axiom$#1\fCenter#2$}

\global\long\def\unaryprimitive#1#2{\UnaryInf$#1\fCenter#2$}

\global\long\def\trinaryprimitive#1#2{\TrinaryInf$#1\fCenter#2$}

\global\long\def\bussproof#1{#1\DisplayProof}

\global\long\def\axiomc#1{\AxiomC{\ensuremath{#1}}}

\global\long\def\axiom#1#2{\axiomprimitive{#1}{#2}}

\global\long\def\binaryinfcdecorated#1#2#3#4#5{#1#2#5\RightLabel{\ensuremath{\prooflabelformat{#4}}}\BinaryInfC{\ensuremath{#3}}}

\global\long\def\trinaryinfcdecorated#1#2#3#4#5#6{#1#2#3#6\RightLabel{\ensuremath{\prooflabelformat{#5}}}\TrinaryInfC{\ensuremath{#4}}}

\global\long\def\unaryinfcdecorated#1#2#3#4{#1#4\RightLabel{\ensuremath{\prooflabelformat{#3}}}\UnaryInfC{\ensuremath{#2}}}

\global\long\def\binaryinfdecorated#1#2#3#4#5#6{#1#2#6\RightLabel{\ensuremath{\prooflabelformat{#5}}}\binaryprimitive{#3}{#4}}

\global\long\def\trinaryinfdecorated#1#2#3#4#5#6#7{#1#2#3#7\RightLabel{\ensuremath{\prooflabelformat{#6}}}\trinaryprimitive{#4}{#5}}

\global\long\def\unaryinfdecorated#1#2#3#4#5{#1#5\RightLabel{\ensuremath{\prooflabelformat{#4}}}\unaryprimitive{#2}{#3}}

\global\long\def\binaryinfc#1#2#3#4{\binaryinfcdecorated{#1}{#2}{#3}{#4}{}}

\global\long\def\trinaryinfc#1#2#3#4#5{\trinaryinfcdecorated{#1}{#2}{#3}{#4}{#5}{}}

\global\long\def\unaryinfc#1#2#3{\unaryinfcdecorated{#1}{#2}{#3}{}}

\global\long\def\binaryinf#1#2#3#4#5{\binaryinfdecorated{#1}{#2}{#3}{#4}{#5}{}}

\global\long\def\trinaryinf#1#2#3#4#5#6{\trinaryinfdecorated{#1}{#2}{#3}{#4}{#5}{#6}{}}

\global\long\def\unaryinf#1#2#3#4{\unaryinfdecorated{#1}{#2}{#3}{#4}{}}

\global\long\def\binaryinfcdouble#1#2#3#4{\binaryinfcdecorated{#1}{#2}{#3}{#4}{\doubleLine}}

\global\long\def\trinaryinfcdouble#1#2#3#4#5{\trinaryinfcdecorated{#1}{#2}{#3}{#4}{#5}{\doubleLine}}

\global\long\def\unaryinfcdouble#1#2#3{\unaryinfcdecorated{#1}{#2}{#3}{\doubleLine}}

\global\long\def\binaryinfdouble#1#2#3#4#5{\binaryinfdecorated{#1}{#2}{#3}{#4}{#5}{\doubleLine}}

\global\long\def\trinaryinfdouble#1#2#3#4#5#6{\trinaryinfdecorated{#1}{#2}{#3}{#4}{#5}{#6}{\doubleLine}}

\global\long\def\unaryinfdouble#1#2#3#4{\unaryinfdecorated{#1}{#2}{#3}{#4}{\doubleLine}}

\global\long\def\binaryinfcdotted#1#2#3#4{\binaryinfcdecorated{#1}{#2}{#3}{#4}{\dottedLine}}

\global\long\def\trinaryinfcdotted#1#2#3#4#5{\trinaryinfcdecorated{#1}{#2}{#3}{#4}{#5}{\dottedLine}}

\global\long\def\unaryinfcdotted#1#2#3{\unaryinfcdecorated{#1}{#2}{#3}{\dottedLine}}

\global\long\def\binaryinfdotted#1#2#3#4#5{\binaryinfdecorated{#1}{#2}{#3}{#4}{#5}{\dottedLine}}

\global\long\def\trinaryinfdotted#1#2#3#4#5#6{\trinaryinfdecorated{#1}{#2}{#3}{#4}{#5}{#6}{\dottedLine}}

\global\long\def\unaryinfdotted#1#2#3#4{\unaryinfdecorated{#1}{#2}{#3}{#4}{\dottedLine}}

\global\long\def\axrulesp#1#2#3#4{\bussproof{\unaryinf{\axiomc{\vphantom{#4}}}{#1}{#2}{#3}}}

\global\long\def\axrulespdouble#1#2#3#4{\bussproof{\unaryinfdouble{\axiomc{\vphantom{#4}}}{#1}{#2}{#3}}}

\global\long\def\axrulespdotted#1#2#3#4{\bussproof{\unaryinfdotted{\axiomc{\vphantom{#4}}}{#1}{#2}{#3}}}

\global\long\def\axrule#1#2#3{\axrulesp{#1}{#2}{#3}.}

\global\long\def\axruledouble#1#2#3{\axrulespdouble{#1}{#2}{#3}.}

\global\long\def\axruledotted#1#2#3{\axrulespdotted{#1}{#2}{#3}.}

\global\long\def\unrule#1#2#3#4#5{\bussproof{\unaryinf{\axiom{#1}{#2}}{#3}{#4}{#5}}}

\global\long\def\unruledouble#1#2#3#4#5{\bussproof{\unaryinfdouble{\axiom{#1}{#2}}{#3}{#4}{#5}}}

\global\long\def\unruledotted#1#2#3#4#5{\bussproof{\unaryinfdotted{\axiom{#1}{#2}}{#3}{#4}{#5}}}

\global\long\def\unrulec#1#2#3{\bussproof{\unaryinfc{\axiomc{#1}}{#2}{#3}}}

\global\long\def\unrulecdouble#1#2#3{\bussproof{\unaryinfcdouble{\axiomc{#1}}{#2}{#3}}}

\global\long\def\unrulecdotted#1#2#3{\bussproof{\unaryinfcdotted{\axiomc{#1}}{#2}{#3}}}

\global\long\def\unrulecsep#1#2#3#4#5{\bussproof{\unaryinfc{\axiomc{#1\proofsep#2}}{#3\proofsep#4}{#5}}}

\global\long\def\unrulecsepdouble#1#2#3#4#5{\bussproof{\unaryinfcdouble{\axiomc{#1\proofsep#2}}{#3\proofsep#4}{#5}}}

\global\long\def\unrulecsepdotted#1#2#3#4#5{\bussproof{\unaryinfcdotted{\axiomc{#1\proofsep#2}}{#3\proofsep#4}{#5}}}

\global\long\def\binrule#1#2#3#4#5#6#7{\bussproof{\binaryinf{\axiom{#1}{#2}}{\axiom{#3}{#4}}{#5}{#6}{#7}}}

\global\long\def\binruledouble#1#2#3#4#5#6#7{\bussproof{\binaryinfdouble{\axiom{#1}{#2}}{\axiom{#3}{#4}}{#5}{#6}{#7}}}

\global\long\def\binruledotted#1#2#3#4#5#6#7{\bussproof{\binaryinfdotted{\axiom{#1}{#2}}{\axiom{#3}{#4}}{#5}{#6}{#7}}}

\global\long\def\binrulec#1#2#3#4{\bussproof{\binaryinfc{\axiomc{#1}}{\axiomc{#2}}{#3}{#4}}}

\global\long\def\binrulecdouble#1#2#3#4{\bussproof{\binaryinfcdouble{\axiomc{#1}}{\axiomc{#2}}{#3}{#4}}}

\global\long\def\binrulecdotted#1#2#3#4{\bussproof{\binaryinfcdotted{\axiomc{#1}}{\axiomc{#2}}{#3}{#4}}}

\global\long\def\axrulec#1#2{\unrulec{}{#1}{#2}}

\global\long\def\axrulecdouble#1#2{\unrulecdouble{}{#1}{#2}}

\global\long\def\axrulecdotted#1#2{\unrulecdotted{}{#1}{#2}}

\global\long\def\xyhelp{{\scriptscriptstyle \begin{array}{lc}
\mathrm{\backslash xyarrow} & \text{arrow }\texttt{udlr}\\
\mathrm{\backslash xyarrowdir} & \mathrm{\backslash hole}\\
\mathrm{\backslash xyembed},\mathrm{\backslash xymapsto} & \text{style (tail shaft head) }\texttt{=.:\textasciitilde-->>||(//ox}\\
 & \text{style variant }\texttt{\ensuremath{\mathrm{\backslash sp}}\ensuremath{\mathrm{\backslash sb}}23}\\
\mathrm{\backslash xy\{port,starboard,break\}} & \text{curve,offset }\texttt{\ensuremath{\pm}length}\\
\mathrm{\backslash xynoarrow} & \text{curve direction}\texttt{ul,ur}\\
 & \text{placement }\texttt{-<><<(0.4)!{[r];[d]}*<modifiers>}
\end{array}}}

\global\long\def\xymathcomment#1{\{\}}

\newcommandx\xyarrow[5][usedefault, addprefix=\global, 1=, 2=->, 3=, 4=0ex, 5=0ex]{\ar@#3{#2}@/^{#4}/@<#5>[#1]}

\newcommandx\xyembed[5][usedefault, addprefix=\global, 1=, 2=, 3=, 4=0ex, 5=0ex]{\ar@{^{(}->}@/^{#4}/@<#5>[#1]}

\newcommandx\xymapsto[5][usedefault, addprefix=\global, 1=, 2=, 3=, 4=0ex, 5=0ex]{\ar@{|->}@/^{#4}/@<#5>[#1]}

\newcommandx\xyarrowdir[5][usedefault, addprefix=\global, 1=, 2=->, 3=, 4={ur,dr}, 5=0ex]{\ar@#3{#2}@(#4)@<#5>[#1]}

\newcommandx\xynoarrow[4][usedefault, addprefix=\global, 1=, 2=, 3=-, 4=0ex]{\ar@{}@<#4>[#1]|#3{#2}}

\global\long\def\xymid#1#2#3{\csname@firstofone\endcsname#1|#3{#2}}

\global\long\def\xyport#1#2#3{\csname@firstofone\endcsname#1\sp#3{#2}}

\global\long\def\xystarboard#1#2#3{\csname@firstofone\endcsname#1\sb#3{#2}}

\global\long\def\xyR#1{\csname xydef@\expandafter\endcsname\csname xymatrixrowsep@\endcsname{#1}}

\global\long\def\xyC#1{\csname xydef@\expandafter\endcsname\csname xymatrixcolsep@\endcsname{#1}}

\newcommandx\xyresize[3][usedefault, addprefix=\global, 1=2pc, 2=2pc]{\xyR{#1}\xyC{#2}#3\xyR{2pc}\xyC{2pc}}

\global\long\def\xyadjarrows#1#2#3{\xyport{\xyarrow[#1][][][1pc]}{#2}{}\xynoarrow[#1][{\textstyle \bot}]\xystarboard{\xyarrow[#1][<-][][-1pc]}{#3}{} }

\global\long\def\xyadjcustom#1#2#3#4#5#6{\xyport{\xyarrow[#1][#4][][1pc]}{#2}{}\xynoarrow[#1][{\textstyle #6}]\xystarboard{\xyarrow[#1][#5][][-1pc]}{#3}{} }

\global\long\def\Operatorname#1{\operatorname{#1}}

\global\long\def\rotxc#1{\begin{sideways}#1\end{sideways}}

\global\long\def\invert#1{\hbox{\rotxc{\rotxc{\ensuremath{#1}}}}}

\global\long\def\lyxmathchoice#1#2#3#4{\mathchoice{#1}{#2}{#3}{#4}}

\global\long\def\mathraise#1#2#3#4{\lyxmathchoice{\raisebox{#2}{${\displaystyle {#1}}$}}{\raisebox{#2}{$#1$}}{\raisebox{#3}{${\scriptstyle {#1}}$}}{\raisebox{#4}{${\scriptscriptstyle {#1}}$}}}

\global\long\def\mathraiseord#1#2#3#4{\mathord{\mathraise{#1}{#2}{#3}{#4}}}

\global\long\def\ltxifnextchar{\csname@ifnextchar\endcsname}

\global\long\def\df{\iffontspec\eqdef\else\overset{\raisebox{-0.4ex}{\ensuremath{{\scriptscriptstyle \textup{def}}}}}{=}\fi}

\global\long\def\set#1#2{\perfectbinary{IncreaseHeight}\{|\}{#1\mathrel{}}{\mathrel{}#2}}

\global\long\def\shpos{\mathraiseord{\downarrow}{0.22ex}{0.14ex}{0.11ex}\ltxifnextchar\shneg{\mkern-2.8mu  }{}}

\global\long\def\shneg{\mathraiseord{\uparrow}{0.29ex}{0.21ex}{0.14ex}\ltxifnextchar\shpos{\mkern-2.8mu  }{}}

\global\long\def\Shpos{\mathraiseord{\Downarrow}{0.2ex}{0.15ex}{0.09ex}\ltxifnextchar\Shneg{\iffontspec\mkern-5.3mu  \else\mkern-2.5mu  \fi}{}}

\global\long\def\Shneg{\mathraiseord{\Uparrow}{0.28ex}{0.21ex}{0.14ex}\ltxifnextchar\Shpos{\iffontspec\mkern-5mu  \else\mkern-2.5mu  \fi}{}}

\global\long\def\wpar{\mathbin{\iffontspec\upand\else\invert{\&}\fi}}

\global\long\def\oc{\mathord{!}}

\global\long\def\ocv{\mathord{\textrm{\textup{!`}}}}

\global\long\def\wn{\mathord{?}}

\global\long\def\with{\mathbin{\&}}

\global\long\def\falso{\mathbf{\boldsymbol{0}}}

\global\long\def\unite{\mathbf{\boldsymbol{1}}}

\global\long\def\mixarrow{\rightarrowtriangle}

\global\long\def\system#1{\mathbf{#1}}

\global\long\def\LL{\system{LL}}

\global\long\def\ILL{\system{ILL}}

\global\long\def\LK{\system{LK}}

\global\long\def\LLP{\system{LLP}}

\global\long\def\LC{\system{LC}}

\global\long\def\LJ{\system{LJ}}

\global\long\def\NJ{\system{NJ}}

\global\long\def\MALL{\system{MALL}}

\global\long\def\LU{\system{LU}}

\global\long\def\polsystem#1{\texorpdfstring{{\system{#1}^{\mathrlap{\eta}}}_{\mkern-1.5mu  p}}{#1\eta p}}

\global\long\def\MLJp{\polsystem{MLJ}}

\global\long\def\IMLLp{\polsystem{IMLL}}

\global\long\def\LJp{\polsystem{LJ}}

\global\long\def\IMALLp{\polsystem{IMALL}}

\global\long\def\IMELLp{\polsystem{IMELL}}

\global\long\def\ILLp{\polsystem{ILL}}

\global\long\def\LKp{\polsystem{LK}}

\global\long\def\LLp{\polsystem{LL}}

\global\long\def\subtype{\mathrel{\mathord{<}\mkern-1.5mu  {\vcentcolon}}}

\global\long\def\otimescriptstyle{{\raisebox{0.07ex}{\hbox{\mbox{\ensuremath{{\scriptstyle \otimes}}}}}}}

\global\long\def\otimescriptscriptstyle{{\raisebox{0.05ex}{\hbox{\mbox{\ensuremath{{\scriptscriptstyle \otimes}}}}}}}

\global\long\def\smallotimes{\mkern1mu  {\lyxmathchoice{\otimescriptstyle}{\otimescriptstyle}{\otimescriptscriptstyle}{\otimescriptscriptstyle}}\mkern1mu  }

\global\long\def\keyword#1{\Operatorname{\textsf{#1}}}

\global\long\def\bnf{\iffontspec\Coloneq\else\Coloneqq\fi}

\global\long\def\bmid{\mathrel{\boldsymbol{\nthleft{1}|}}}

\global\long\def\moins{\circleddash}

\global\long\def\plus{+}

\global\long\def\push{\mathord{\cdot}}

\global\long\def\pleft{\pi_{1}}

\global\long\def\pright{\pi_{2}}

\global\long\def\pith{\pi_{i}}

\global\long\def\ppair#1#2{{<}\mspace{1mu}#1\mspace{1mu}\mathord{;}\mspace{3mu}#2\mspace{1mu}{>}}

\global\long\def\ileft{\iota_{1}}

\global\long\def\iright{\iota_{2}}

\global\long\def\iith{\iota_{i}}

\global\long\def\mt{\smash{\tilde{\mu}}\vphantom{\mu}}

\global\long\def\lmmt{\bar{\lambda}\mu\mt}

\global\long\def\mustar{\mu\mathord{\star}}

\global\long\def\cut#1#2{\cutprimitive{#1}{#2}}

\global\long\def\Bra#1{\cutbraprimitive{#1}}

\global\long\def\Ket#1{\cutketprimitive{#1}}

\global\long\def\parens#1{\perfectparens{#1}}

\global\long\def\tphat{{\raisebox{-0.5ex}{\hbox{\ensuremath{\hat{\hphantom{e}}}}}}}

\global\long\def\muhat{\hat{\mu}}

\global\long\def\contor#1#2#3#4{\perfectcase{{#1.}#2\mkern2mu  }{\mkern2mu  {#3.}#4}}

\global\long\def\mui#1#2#3#4{\mt\contor{#1}{#2}{#3}{#4}}

\global\long\def\mupair#1#2#3#4{\mu\ppair{{#1}.#2}{{#3}.#4}}

\global\long\def\term#1{\mu{{<}#1{>}}}

\global\long\def\init#1{\mt[#1]}

\global\long\def\curlybraces#1{\perfectunary{CurrentHeight}\lbrace\rbrace{#1}}

\global\long\def\brackets#1{\perfectbrackets{#1}}

\global\long\def\substslash#1#2{\csname cut@computeBinary@CurrentHeightPlusOne\endcsname{}{#1}{\nthmiddle{\csname cut@n\endcsname}/}{#2}{}}

\global\long\def\bracketsmid#1#2{\csname cut@computeBinary@CurrentHeightPlusOne\endcsname{}{#1}{\mathbin{\nthmiddle{\csname cut@n\endcsname}|}}{#2}{}}

\global\long\def\subst#1#2{\perfectunary{CurrentHeightPlusOne}[]{\substslash{#1}{#2}}}

\global\long\def\expsubst#1#2{\perfectunary{CurrentHeightPlusOne}\{\}{\substslash{#1}{#2}}}

\global\long\def\mkernsp#1#2{\ltxifnextchar#2{\csname mkern\endcsname#1}{}}

\global\long\def\varp#1{#1^{{\mkern-1mu  \mathord{\plus}\mkern-2mu  }}\mkernsp{-3mu}.}

\global\long\def\varn#1{#1^{{\moins}\mkern-1mu  }\mkernsp{-3mu}.}

\global\long\def\termp#1{#1_{\smash{\plus}}}

\global\long\def\termn#1{#1_{\smash{\moins}}}

\global\long\def\eps#1{\varepsilon_{#1}}\global\long\def\epsprime#1{\varepsilon'_{#1}}

\global\long\def\vareps#1#2{#1^{\mkern0.5mu  {\eps{#2}}\mkern-1mu  }\mkernsp{-2mu}.}

\global\long\def\termeps#1#2{#1_{\smash{\eps{#2}}}}

\global\long\def\varepsprime#1#2{#1^{\mkern0.5mu  {\epsprime{#2}}\mkern-1mu  }\mkernsp{-2mu}.}

\global\long\def\termepsprime#1#2{#1_{\smash{\epsprime{#2}}}}

\global\long\def\tbrack#1{\curlybraces{#1}}

\global\long\def\ebrack#1{\curlybraces{#1}}

\global\long\def\mtbrack#1{\mt\{#1\}}

\global\long\def\mubrack#1{\mu\{#1\}}

\global\long\def\letin#1{\keyword{let}#1\keyword{in}}

\global\long\def\matchwith#1{\keyword{match}#1\keyword{with}}

\global\long\def\be{\keyword{be}}

\global\long\def\addsubst#1#2{,\substslash{#1}{#2}}

\global\long\def\variableset#1#2{\Operatorname{\textbf{#1}}{#2}}

\global\long\def\v#1{\variableset v{#1}}

\global\long\def\fv#1{\variableset{fv}{#1}}

\global\long\def\Bracks#1{[\mkern-3mu  [#1]\mkern-3mu  ]}

\global\long\def\dispatch#1#2#3#4#5{\delta(#1,\allowbreak#2.#3,\allowbreak#4.#5)}

\global\long\def\L{\system L}

\global\long\def\argrel#1#2#3{\mathrel{#3_{\cramped{\textup{#1}\ifx\empty#1\empty#2\else_{#2}\fi}}}}

\global\long\def\argred#1#2{\argrel{#1}{#2}{\vartriangleright}}

\global\long\def\argnred#1#2{\iffontspec\argrel{#1}{#2}{\nvartriangleright}\else\argrel{#1}{#2}{\not\vartriangleright}\fi}

\global\long\def\argredinv#1#2{\argrel{#1}{#2}{\vartriangleleft}}

\global\long\def\argRed#1#2{\argrel{#1}{#2}{\rightarrow}}

\global\long\def\argRedinv#1#2{\argrel{#1}{#2}{\leftarrow}}

\global\long\def\argPar#1#2{\argrel{#1}{#2}{\Rightarrow}}

\global\long\def\argParinv#1#2{\argrel{#1}{#2}{\Leftarrow}}

\global\long\def\argreds#1#2{\argrel{#1}{#2}{\vartriangleright^{\ast}}}

\global\long\def\argredinvs#1#2{\argrel{#1}{#2}{\vartriangleleft^{\ast}}}

\global\long\def\argReds#1#2{\argrel{#1}{#2}{\rightarrow^{\ast}}}

\global\long\def\argRedinvs#1#2{\argrel{#1}{#2}{\leftarrow^{\ast}}}

\global\long\def\argPars#1#2{\argrel{#1}{#2}{\Rightarrow^{\ast}}}

\global\long\def\argParinvs#1#2{\argrel{#1}{#2}{\Leftarrow^{\ast}}}

\global\long\def\argeq#1#2{\argrel{#1}{#2}{\simeq}}

\global\long\def\argneq#1#2{\iffontspec\argrel{#1}{#2}{\nsimeq}\else\argrel{#1}{#2}{\not\simeq}\fi}

\global\long\def\machred#1{\succ_{#1}}

\global\long\def\machredinv#1{\prec_{#1}}

\global\long\def\group#1{{#1}}

\global\long\def\visiblelefteqn#1{\lefteqn{#1}}

\global\long\def\mathcomment#1{\{\}}

\global\long\def\proofsep{\mathrel{\vphantom{(\mid c'}{\vdash}}}

\global\long\def\cond#1{\tag{\ensuremath{#1}}}

\global\long\def\relrel#1{\mathord{}\mathrel{#1}\mathord{}}

\DeclarePairedDelimiter\abs{\lvert}{\rvert}
\DeclarePairedDelimiter\norm{\lVert}{\rVert}

\global\long\def\abs#1{\left|#1\right|}

\global\long\def\norm#1{\left\Vert #1\right\Vert }

\global\long\def\dom#1{\Operatorname{dom}#1}

\global\long\def\id{\mathrm{id}}

\global\long\def\Id{\mathrm{Id}}

\global\long\def\idfunctor{1}

\global\long\def\obj#1{\abs{#1}}

\global\long\def\op#1{{#1}^{\textup{op}}}

\global\long\def\inv#1{#1^{-1}}

\global\long\def\restr#1#2{{#1_{\upharpoonright#2}}}

\global\long\def\dfequiv{\overset{{\scriptscriptstyle \textrm{def}}}{\Longleftrightarrow}}

\global\long\def\reflect{\mathbin{\triangleleft}}

\global\long\def\sha#1{#1^{\sharp}}

\global\long\def\fla#1{#1^{\flat}}

\global\long\def\adj#1{#1^{\ast}}

\global\long\def\Set#1{\big\{#1\big\}}

\global\long\def\setmid{\,\big|\,}

\global\long\def\C{\mathscr{C}}

\global\long\def\Ccal{\mathcal{C}}

\global\long\def\scL{\mathscr{L}}

\global\long\def\scC{\mathscr{C}}

\global\long\def\scV{\mathscr{V}}

\global\long\def\bbN{\iffontspec\BbbN\else\mathbb{N}\fi}

\global\long\def\to{\mathbin{\rightarrow}}

\global\long\def\dotarrow{\mathrel{\mkern-2mu  \dot{\mkern2mu  \rightarrow}}}

\global\long\def\nattransf{\dotarrow}

\global\long\def\ddotarrow{\mathrel{\!\ddot{\,\rightarrow}}}

\global\long\def\embed{\hookrightarrow}

\global\long\def\simrightarrow{\iffontspec\similarrightarrow\else\stackrel{\sim}{\rightarrow}\fi}

\global\long\def\farrow#1#2#3{#2\xrightarrow{#1}#3}

\global\long\def\farr#1{\mathrel{{\relbar}\raisebox{0.2ex}{\ensuremath{{\scriptstyle #1}}}\mkern-1mu  {\rightarrow}}}

\global\long\def\stackrightleftarrows#1#2{\mathrel{\underset{#2}{\overset{#1}{\rightleftarrows}}}}

\global\long\def\stackleftrightarrows#1#2{\mathrel{\underset{#2}{\overset{#1}{\leftrightarrows}}}}

\global\long\def\labelrightarrowcompact#1{\mathrel{\xrightarrow{#1}}}

\def\labelrightarrowcompact#1#2{\mbox{$#1\ifx#1\textstyle \clipbox{0em 0em {.4\width} 0em}{$\vphantom{A}{\relbar}$}\mkern0mu\raisebox{0.3ex}{$\scriptstyle#2$}\mkern0mu\clipbox{{.3\width} 0em 0em 0em}{$\rightarrow$}\else \xrightarrow{#2}\fi $}}

\global\long\def\labelrightarrow#1{\mathbin{\mathpalette\labelrightarrowcompact{#1}}}

\global\long\def\morphism#1{\labelrightarrow{#1}}

\global\long\def\natiso{\mathrel{{\textstyle \cong}}}

\global\long\def\spancols#1#2{\csname multicolumn\endcsname{#2}{c}{#1}}

\global\long\def\spanrows#1#2{\multirow{#2}{*}{\ensuremath{#1}}}

\global\long\def\spancolsleft#1#2{\csname multicolumn\endcsname{#2}{l}{#1}}

\global\long\def\spantwo#1{\csname multicolumn\endcsname{2}{c}{#1}}

\global\long\def\linespacing#1{\rule[-#1]{0pt}{0pt}}

\global\long\def\mathvspace{\vphantom{\rule{0pt}{1.75em}}}

\global\long\def\qedmultiline{\tag*{\qedhere}}

\global\long\def\where{\;\quad\text{where}\;\quad}

\global\long\def\when{\;\quad\text{when}\;\quad}

\global\long\def\textor{\;\quad\text{or}\;\quad}

\global\long\def\addstretchamount{2.5}

\global\long\def\addsmallstretchamount{1.3}

\newcommandx\addcustomstretch[3][usedefault, addprefix=\global, 1=\addstretchamount, 2=0.0em]{\begingroup\edef\arraystretch{#1}\addtolength\arraycolsep{#2}#3\endgroup}

\global\long\def\addstretch#1{\addcustomstretch[\addstretchamount][0.0em]{#1}}

\global\long\def\addsmallstretch#1{\addcustomstretch[\addsmallstretchamount][-0.4em]{#1}}

\global\long\def\removestretch#1{\addcustomstretch[1.0][-4pt]{#1}}

\global\long\def\prom#1{\oc{#1}}

\global\long\def\vdashsubscript#1{\mathrel{{\vdash}\!{\ensuremath{_{#1}}}}}

\global\long\def\focvdash{\mathrel{\vdashsubscript f}}

\global\long\def\diam#1{\group{#1}^{\diamond}}

\global\long\def\colon{\mskip1muplus4muminus2mu\mathord{:}\mskip1muplus4muminus2mu}

\global\long\def\ech#1{\perfectbrackets{#1}}

\global\long\def\fresh{\mathrel{\#}}

\global\long\def\varpm#1{#1^{{\pm}}\mkernsp{-3mu}.}

\global\long\def\cs#1#2{#1^{#2\mkern-1mu  }\mkernsp{-3.5mu}.}

\global\long\def\push{\mathord{\cdot}}

\global\long\def\Li{\system L_{\text{int}}}

\global\long\def\yon{\mathbf{y}}

\global\long\def\pcomp{\mathbin{\iffontspec\vysmblkcircle\else\bullet\fi}}

\global\long\def\ncomp{\mathbin{\iffontspec\vysmwhtcircle\else\circ\fi}}

\global\long\def\pair#1#2{\langle#1,#2\rangle}

\global\long\def\dotcirc{\ifx\fisheye\undefined\mathbin{\visiblelefteqn{\mspace{2.4mu}\mathord{\cdotp}}\mathord{\ncomp}}\else\mathbin{\fisheye}\fi}

\global\long\def\polarity{\varpi}

\global\long\def\swash#1{\iffontspec{\text{\fontspec[Style=Swash,Scale=MatchLowercase]{Arno Pro Bold Italic}#1}}\else\boldsymbol{\textrm{\textbf{\textit{#1}}}}\fi}

\def\swasharg#1{\ensuremath{\mathbcal{#1}\mkern-0.5mu}}\def\swash#1{\textbf{\textit{\swasharg #1}}}

\global\long\def\catstyle#1{\swash{#1}}

\global\long\def\Setcat{\catstyle{Set}}

\global\long\def\Cat{\catstyle{Cat}}

\global\long\def\Dupl{\catstyle{Dupl}}

\global\long\def\Adj{\catstyle{Adj}}

\global\long\def\Tf{\catstyle{Tf}}

\global\long\def\profunctor#1#2{\op{#1}\times#2\rightarrow\Setcat}

\global\long\def\extlambda#1#2{\widehat{#2}.#1}

\global\long\def\D{\mathscr{D}}

\global\long\def\N{\mathscr{N}}

\global\long\def\P{\mathscr{P}}

\global\long\def\psh#1{\widehat{#1}}

\global\long\def\enr#1{\underline{#1\mspace{-1mu}}\mspace{1mu}}

\global\long\def\orth#1{{{#1}^{\bot}}}

\global\long\def\biorth#1{{{#1}^{\bot\bot}}}

\global\long\def\interp#1{\big|#1\big|}

\global\long\def\interpv#1{\big\|#1\big\|}

\providecommand{\Vbar}{\mathord{\bot\mkern-10.5mu\bot}}

\global\long\def\Vbar{\mathord{\bot\mkern-10.5mu  \bot}}

\global\long\def\Bot{\mathord{\Vbar}}

\global\long\def\polar{\mathrel{\Bot}}

\global\long\def\nBot{\mathrel{\not\mkern-1.6mu  \polar}}

\global\long\def\padifnempty#1{\ifx\empty#1\empty\else\mkern2mu  #1^{\mkern0.5mu  }\mkern2mu  \fi}

\global\long\def\xyadjnoresize#1#2#3#4#5{\Xymatrixone\xymatrix{#1\ \ar@<0.8ex>[r]|-(.275){#2}\ar@{}@<0.1ex>[r]|-{\scriptscriptstyle #5}  &  \ar@<0.6ex>[l]|-(.275){#3}\ #4}
 }

\global\long\def\xyadjinline#1#2#3#4{\xyresize[][\ifx\empty#2#3\empty1pc\else25pt\fi]{\xyadjnoresize{#1}{\padifnempty{#2}}{\padifnempty{#3}}{#4}{\bot}}}

\global\long\def\xyadjdisplay#1#2#3#4{\xymatrix{#1\ \ar@<1ex>[r]\sp-{#2}\ar@{}[r]|-\bot &  \ar@<1ex>[l]\sp-{#3}\ #4}
 }

\global\long\def\xyadj#1#2#3#4{\mathchoice{\xyadjdisplay{#1}{#2}{#3}{#4}}{\xyadjinline{#1}{#2}{#3}{#4}}{\xyadjinline{#1}{#2}{#3}{#4}}{\xyadjinline{#1}{#2}{#3}{#4}}}

\makeatletter
\def\underline@size#1#2{\setbox0=\hbox{$\m@th#1#2$}\rlap{\resizebox*{\wd0}{\totalheight}{\raisebox{-\dp0}{\hbox{\normalfont\_}}}}\box0}
\def\underline{\mathpalette\underline@size}
\makeatother
\global\def\dashedScoreFiller{\hbox{\raisebox{-0.4ex}[0.1em][0pt]{\normalfont\hspace{0.1em}\textendash\hspace{0.1em}}}}
\global\def\ruleScoreFiller{\hbox{}\hskip0pt\raisebox{-0.45ex}[0.1ex][0pt]{\resizebox{\displace}{\height}{\hbox{\normalfont\textemdash}}}
}
\newcommand{\hlinesmooth}[1]{\raisebox{-0.45ex}[0.1ex][0pt]{\resizebox{#1}{\height}{\hbox{\normalfont\textemdash}}}}
\newcommand{\hlinemed}{\vskip\baselineskip \noindent \hspace*{\fill}\hlinesmooth{0.62\linewidth}
  \hspace*{\fill}}
\newcommand{\hlinefill}{\vskip0.5\baselineskip \noindent \hspace*{\fill}\hlinesmooth{\linewidth}
  \hspace*{\fill}\vskip-1.4\baselineskip }\selectlanguage{british}

\global\long\def\Alg#1{\mathrm{Alg}(#1)}

\global\long\def\setbin#1#2{\perfectbinary{IncreaseHeight}\{|\}{#1\mathrel{}}{\mathrel{}#2}}

\global\long\def\moins{\mathord{-}}

\global\long\def\normalise#1{#1\mathord{\downarrow}}

\global\long\def\size#1{\perfectunary{IncreaseHeight}||{#1}}

\global\long\def\shpos{\mathord{\downarrow}\ltxifnextchar\shneg{\mkern-2.8mu  }{}}\global\long\def\shneg{\mathord{\uparrow}\ltxifnextchar\shpos{\mkern-2.8mu  }{}}

\global\long\def\Shpos{\mathord{\Downarrow}\ltxifnextchar\Shneg{\mkern-2.8mu  }{}}\global\long\def\Shneg{\mathord{\Uparrow}\ltxifnextchar\Shpos{\mkern-2.8mu  }{}}

\global\long\def\mt{\bar{\mu}}\global\long\def\colon{\mathbin{:}}

\global\long\def\pp{\mathord{+\mkern-2mu  +}}

\global\long\def\varpp#1{#1^{{\pp\mkern-2mu  }}\mkernsp{-3mu}.}

\global\long\def\termpp#1{#1_{\smash{\pp}}}

\global\long\def\polsystem#1{\system{#1}_{\mathit{p}}}

\global\long\def\ILLpdiam{\polsystem{ILL}^{\lozenge}}

\global\long\def\focvdash{\mathrel{\vdashsubscript f}}

\global\long\def\Subst#1#2{\Sigma(#1;#2)}

\global\long\def\Substbij#1#2{\Sigma^{\ast}(#1;#2)}

\global\long\def\Substexp#1#2{\Sigma^{\oc}(#1;#2)}

\global\long\def\bnf{\Coloneq}

\global\long\def\tmu{\mt}

\global\long\def\nUparrow{\Shneg}

\global\long\def\nDownarrow{\Shpos}

\global\long\def\op{\mathrm{op}}

\global\long\def\Dduploid{\mathscr{D}}

\global\long\def\Eduploid{\mathscr{E}}

\global\long\def\com{\oc}

\global\long\def\mon{\mkern-0.4mu  \mathord{\medlozenge}\mkern-1.2mu  }

\global\long\def\diam{\mkern-0.5mu  \mathord{\smalllozenge}\mkern-1mu  }\global\long\def\Emon{\mathcal{E}}

\global\long\def\ccomp{\mathbin{\mathrlap{\mathchoice{\mspace{2mu}}{\mspace{2mu}}{\mspace{2mu}}{\mspace{2mu}}\mathord{\cdot}}\mathord{\circ}}}

\global\long\def\adjoint#1#2#3#4{{#1:#2\rightleftarrows#3:#4}}

\global\long\def\graphadjoint#1#2#3#4{{#1:#2\rightleftharpoons#3:#4}}

\global\long\def\objectp#1{#1_{\mathsf{obj}}}

\global\long\def\fstcomponent#1#2{#1^{#2}_{\mathsf{fst}}}

\global\long\def\sndcomponent#1#2{#1^{#2}_{\mathsf{snd}}}

\global\long\def\Acategory{\mathscr{A}}

\global\long\def\Bcategory{\mathscr{B}}

\global\long\def\Ccategory{\mathscr{C}}

\global\long\def\Dcategory{\mathscr{D}}

\global\long\def\Ecategory{\mathscr{E}}

\global\long\def\Mcategory{\mathscr{M}}

\global\long\def\Lcategory{\mathscr{L}}

\global\long\def\Ncategory{\mathscr{N}}

\global\long\def\Pcategory{\mathscr{P}}

\global\long\def\Vcategory{\mathscr{V}}

\global\long\def\Scategory{\mathscr{S}}

\global\long\def\Ncategoryl{\mathscr{N}_{l}}

\global\long\def\Pcategoryt{\mathscr{P}_{t}}

\global\long\def\twocategory{\mathbf{2}}

\global\long\def\kleisli#1#2{\mathbf{Kl}{[#1,#2]}}

\global\long\def\cokleisli#1#2{\mathbf{coKl}{[#1,#2]}}

\global\long\def\bikleisli#1#2#3{\mathbf{biKl}{[#1,#2,#3]}}

\global\long\def\collage#1#2{\mathbf{coll}_{#1,#2}}

\global\long\def\collageofdist#1{\mathbf{coll}_{#1}}

\global\long\def\duploid#1#2{\mathbf{dupl}_{#1,#2}}

\global\long\def\Setcat{\mathcal{S}\mkern-2.5mu  \mathit{et}}

\global\long\def\adjunction#1{\dashv_{#1}}

\global\long\def\coadjunction#1{\mathrel{{_{#1}{\dashv}}}}

\global\long\def\graphto{\rightharpoonup}

\global\long\def\proofsep{\mathbin{\mathord{\vdash}\vphantom{,}}}\global\long\def\pairing#1#2{[#1,#2]}

\global\long\def\Lin{\Operatorname{\mathit{Lin}}}\global\long\def\Mult{\Operatorname{\mathit{Mult}}}

\global\long\def\lstr{\Operatorname{\mathit{lstr}}}\global\long\def\rstr{\Operatorname{\mathit{rstr}}}\global\long\def\str{\Operatorname{\mathit{str}}}

\global\long\def\Oblique{\mathscr{O}}

\global\long\def\extension#1{#1^{*}}

\global\long\def\coextension#1{{^{*}#1}}

\global\long\def\mixarrow{\multimap}

\title{Syntax and semantics of focalisation\\
with relative monads and comonads}
\date{June 12th 2026}
\author{\begin{tabular}{rl}
Éléonore Mangel & {\large\textit{IRIF, Paris}}\tabularnewline[\doublerulesep]
Paul-André Melliès & {\large\textit{IRIF, INRIA, Paris}}\tabularnewline[\doublerulesep]
Guillaume Munch-Maccagnoni & {\large\textit{INRIA, LS2N CNRS, Nantes}}\tabularnewline[\doublerulesep]
\end{tabular}}
\maketitle
\begin{abstract}
The logical principles of focalisation and polarisation can be used
to design well-behaved term syntaxes for sequent calculus, which
play a role as meta-languages for describing effectful computation.
On the semantics side, this corresponds to an axiomatic and polarised
notion of model of computation stated in terms of non-associative
categories as well as adjunctions between ``bare'' functors (reflexive
graph morphisms) over such non-associative categories.

In this paper, we study the special and delicate cases of resource
and effect modalities in a general intuitionistic and linear setting:
an exponential comonad $\com$ (refining the necessity modality $\square$)
and a strong monad (written $\mon$). The starting point of our contribution
is noticing that the completeness for a polarised syntax for $\oc$
and $\lozenge$ with respect to (co)monads in linear call-by-push-value
models can be achieved if we move to \emph{relative} (co)monads \citep{Altenkirch2015,Arkor2024}:
more precisely, comonads relative to $\shpos$ (the positive shift
functor) for $\oc$ and monads relative to $\shneg$ (the negative
shift functor) for $\lozenge$.

These specialisations of the concept of relative (co)monad to call-by-push-value
adjunctions recently appeared in \citet{Jiang2025} and \citet{Mellies2025BLC,Mellies2026TLLA}.
Yet the syntax we present arose from proof-theoretic consideration
in \citet{munchmonolateral,CFM2015}, without the link with relative
(co)monads being noticed at the time. Our first remark and explanation
is thus that (co)monads relative to a call-by-push-value adjunction
have been motivated previously from a proof-theoretic perspective
in the context of focalisation, which also provides a meta-language
for these concepts in an effectful setting.

We carry out the study of these modalities from the axiomatic, non-associative
point of view. We recall the definition of adjunction between bare
functors in this context, and establish correspondence results between
this notion of adjunction and that of relative adjunction. This correspondence
is then extended to linear-non-linear and strong versions of adjunctions
as needed to model $\oc$ and $\lozenge$.
\end{abstract}

\section{From focalisation and call-by-push-value to polarised calculi}

\begin{figure*}[t]
\begin{boiteombree}
\hspace*{\fill}\subfloat[Types]{$\addcustomstretch[1.4][-0.2em]{\begin{array}{crl}
\text{polarities} & \varepsilon & \bnf+\bmid\moins\\
\text{types} & A,B,\termeps A{} & \bnf P\bmid N\\
\text{positive} & P,Q,\termp A & \bnf\varp X\bmid\unite\bmid A\otimes B\bmid A\oplus B\bmid\falso\\
\text{negative} & N,M,\termn A & \bnf\varn X\bmid A\mixarrow B\bmid A\with B\bmid\top
\end{array}}$

}\hspace*{\fill}\subfloat[Judgements]{\begin{minipage}[t]{0.39\columnwidth}\begin{varwidth}{\linewidth}
\setlist{nosep}
\begin{itemize}
\item {\small Judgements are of the form $\Gamma\vdash A$}{\small\par}
\item {\small$\Gamma$ is an ordered list of formulæ $A_{1},\dots,A_{n}$.}{\small\par}
\item {\small$\Gamma,\Gamma'$ denotes the concatenation of $\Gamma$ and
$\Gamma'$.}{\small\par}
\end{itemize}
\end{varwidth}
\end{minipage}\vspace*{\medskipamount}

}\hspace*{\fill}\vspace*{-1ex}

\hspace*{\fill}\subfloat[Identity]{$\addcustomstretch[2.6][0.4em]{\begin{array}{cc}
\axrulesp AA{(ax)}{\vdash} & \binrule{\Gamma}A{\Gamma',A}B{\Gamma,\Gamma'}B{(cut)}\end{array}}$

}\hspace*{\fill}\subfloat[\label{fig:Structure-(exchange)-1-1}Structure]{\begin{raggedright}
\begin{varwidth}{\linewidth}
\begin{flushleft}
$\addcustomstretch[2.6]{\begin{array}{l}
\unrulec{\Gamma,A,B,\Gamma'\vdash C}{\Gamma,B,A,\Gamma'\vdash C}{(ex)}\end{array}}$
\par\end{flushleft}
\end{varwidth}
\par\end{raggedright}
}\hspace*{\fill}\vspace*{-1ex}

\hspace*{\fill}\subfloat[Logic]{$\addcustomstretch[2.6][-0.0em]{\begin{array}{ccc}
\binrule{\Gamma}A{\Gamma',B}C{\Gamma,\Gamma',A\mixarrow B}C{(\ensuremath{{\mixarrow}\vdash})} & \unrulecsep{\Gamma,A}B{\Gamma}{A\mixarrow B}{(\ensuremath{\vdash{\mixarrow}})} & \axrulesp{}{\unite}{(\ensuremath{\vdash\unite})}{\Gamma\vdash}\\
\binrule{\Gamma}A{\Gamma'}B{\Gamma,\Gamma'}{A\otimes B}{(\ensuremath{\vdash\otimes})} & \unrulecsep{\Gamma,A,B}C{\Gamma,A\otimes B}C{(\ensuremath{\otimes\vdash})} & \unrulecsep{\Gamma}C{\Gamma,\unite}C{(\ensuremath{\unite\vdash})}\\
\binrule{\Gamma}A{\Gamma}B{\Gamma}{A\with B}{(\ensuremath{\vdash\with})} & \unrulecsep{\Gamma,A_{i\in\{1,2\}}}C{\Gamma,A_{1}\with A_{2}}C{(\ensuremath{\with_{i}\focvdash})} & \axrulesp{\Gamma}{\top}{(\ensuremath{\vdash\top})}{\Gamma\vdash}\\
\binrule{\Gamma,A}C{\Gamma,B}C{\Gamma,A\oplus B}C{(\ensuremath{\oplus\vdash})} & \unrulecsep{\Gamma}{A_{i\in\{1,2\}}}{\Gamma}{A_{1}\oplus A_{2}}{(\ensuremath{\vdash\oplus_{i}})} & \axrulesp{\Gamma,\falso}A{(\ensuremath{\falso\vdash})}{\Gamma\vdash}
\end{array}}$

}\hspace*{\fill}
\end{boiteombree}
\caption{\label{fig:IMALL}$\protect\system{IMALL}$}
\end{figure*}

In this paper, we are interested in the links between the semantics
of effectful programs featuring linearity, and the proof theory of\emph{
}intuitionistic multiplicative-additive linear sequent calculus~($\system{IMALL}$,
see \prettyref{fig:IMALL}). On the syntactic side we consider a polarised
version of $\system{IMALL}$, more precisely a linear and intuitionistic
version of polarised classical logic described by \citet{girardnew}
and Danos, Joinet and Schellinx \citep{danos95new}, which originated
in a Curry-Howard interpretation of focusing proof search \citep{Andreoli92}
with roots in linear logic \citep{Gir87}. Systems arising from this
perspective are \emph{polarised} in the sense that a difference between
positive and negative formulæ appears and is handled formally. This
distinction between formulæ is similar to the one type segregation
in the call-by-push-value model of programming effects (CBPV: \citealt{Levy99CBPV,Levy2004}).

On the semantic side, we consider a linear version of CBPV, the linear
CBPV models \citep{FioreLCBPV,Mellies2012parametric,CFM2015}. Recall
that linear CBPV models are enriched (or strong) adjunctions between
a (symmetric monoidal and distributive) category of values and a category
of co-values (cartesian closed relatively to $\Vcategory$ in a certain
sense):
\begin{equation}
\xymatrix{\Vcategory\xyport{\xyarrow[rr][][][1.5pc]}{\shneg}{} & \bot & \xyport{\xyarrow[ll][][][1.5pc]}{\shpos}{}\Scategory}
\label{eq:lcbpv-adj}
\end{equation}
The interest in such situations arose after a long series of works
starting from monads used to model side-effects \citep{Moggi89computationallambda-calculus,Moggi1991},
which soon after shifted attention to adjoint situations decomposing
monads as pictured above \citep{Fiore94PhD,PowerRobinson97,Thielecke97Thesis,Levy99CBPV,Levy2004,Levy05adjunctionmodels},
followed by various attempts to ``linearise'' it, including the
aforementioned works and others such as \citet{Egger2009}.\footnote{As a matter of fact the works on polarisation by \citet{girardnew,Gir93}
can themselves be considered historically as a source of inspiration
for the adjoint point of view, via \citet{benton1996linear} and \citet{StreicherReus98ContinuationAbstractMachines}.}

We avoid developing the full details of strength for the adjunction
above. Being enriched (or strong) for the situation \eqref{eq:lcbpv-adj}
means that additional structure is needed to express multiple formulæ
on the left-hand side of sequents, obtained by asking the above adjunction
to underlie a $\widehat{\Vcategory}$-enriched adjunction where $\widehat{\Vcategory}$
is the symmetric monoidal category of presheaves on $\Vcategory$
\citep{CFM2015}, or equivalently by asking that the adjoints are
compatible with the monoidal product of $\Vcategory$ and with the
action of $\Vcategory^{\op}$ on $\Scategory$ in a manner defined
in \citet{Mellies2012parametric}.

As is by now well understood, a similar situation as \eqref{eq:lcbpv-adj}
underpins Girard's polarised logics \citep{girardnew,Gir93}, albeit
in a manner which is both more specific—by pertaining to (constructive)
classical, intuitionistic and linear logics: said differently to non-linear,
linearly-used and fully linear continuations (respectively)—and more
general: for instance, they integrate in the case of \citet{Gir93}
several adjunctions in the picture above as opposed to a single one.
Many people should be cited for establishing the links between polarised
logics and CBPV, mainly via the relationships they have in common
with negative translations, and continuation-passing style.\smartcites[for instance][]{Murthy92LC}{LRS93}{Thielecke97Thesis}{StreicherReus98ContinuationAbstractMachines}{Ogata2000}{CH00Duality}{Sel01Control}{Ogat2002}{Hof02ComplLambdaMu}{Lau02PhD}{Laurent_2003}{LQTdF05}{melliestabareau}

But a strength of the proof-theoretical point of view symbolised by
polarised logics is that they can be a source of inspiration for better
understanding CBPV models. In this paper, we would like indeed to
further understand such situations used in programming language semantics,
particularly when comes the consideration of a comonad on $\Vcategory$
and/or a monad on $\Scategory$ which are again strong in a suitable
sense—these are important situations related to the study of resources
and effects as modalities.  Thus we find it useful to start by recalling
the links between effectful semantics and polarised logics.

\subsection{The deductive system associated to an adjunction}

Proof theory is concerned with the formal study of the structure of
deductive systems. Categorical proof theory has historically approached
deductive systems as categories, given by:
\begin{enumerate}
\item a collection of objects (formulæ),
\item for each pair of objects $A,B$ a set of morphisms from $A$ to $B$
(of proofs $A\vdash B$),
\item a well-typed and associative composition operation ($A\vdash B$ and
$B\vdash C$ give $A\vdash C$), and
\item for each object $A$ an identity morphism from $A$ to $A$ ($A\vdash A$),
neutral for composition.
\end{enumerate}
A useful perspective brought by proof theory is that an adjoint situation
such as \eqref{eq:lcbpv-adj} can indeed be analysed as a deductive
system, however with a composition which is not associative a priori.
Given any (associative) adjoint situation 
\[
\adjoint L{\Acategory}{\Bcategory}R
\]
(where this notation denotes that $R:\Bcategory\rightarrow\Acategory$
is right adjoint to $L:\Acategory\rightarrow\Bcategory$), one can
define a deductive system $\duploid LR$ as follows:
\begin{align*}
\obj{\smash{\duploid LR}} & =\obj{\Acategory}\uplus\obj{\Bcategory} & \duploid LR(A,B) & =\Oblique(\varp A,\varn B)
\end{align*}
where 
\[
\Oblique:\Acategory^{\op}\times\Bcategory\rightarrow\Setcat
\]
is a chosen distributor of oblique morphisms of the adjunction, and
where by definition
\begin{align*}
\varp X & \df\begin{cases}
X & \text{if }X\in\obj{\Acategory}\\
RX & \text{if }X\in\obj{\Bcategory}
\end{cases} & \varn X & \df\begin{cases}
LX & \text{if }X\in\obj{\Acategory}\\
X & \text{if }X\in\obj{\Bcategory}
\end{cases}
\end{align*}
Objects in $\Acategory$ are called \emph{positive} in $\duploid LR$
and those in $\Bcategory$ \emph{negative}. It is an interesting exercise
to check that for all objects $A,B,C$ one can define a composition
\[
\duploid LR(C,B)\times\duploid LR(A,C)\rightarrow\duploid LR(A,B)
\]
and that this composition is associative if and only if the adjunction
$\adjoint L{\Acategory}{\Bcategory}R$ is idempotent, in the sense
of its multiplication or co-multiplication being iso \citep{munchduploids}.

\textbf{Non-associative categories}, such as $\duploid LR$, are simply
categories without the assumption that composition is associative.
We use the notation $\ccomp$ for composition operations that are
not associative a priori.

Additional structure, such as the symmetric monoidal and distributive
structure on $\Vcategory$ and the cartesian closed structure on $\Scategory$
relative to $\Vcategory$ \eqref{eq:lcbpv-adj}, allows us to interpret
logical connectives. In fact:
\begin{prop}[\citealp{CFM2015}]
\label{prop:The-deductive-system}The deductive system $\duploid{\shneg}{\shpos}$
associated to a linear CBPV adjunction \eqref{eq:lcbpv-adj} satisfies
the rules of $\system{IMALL}$ (\prettyref{fig:IMALL}).
\end{prop}

From the point of view of deductive systems as non-associative categories,
just as in categorical proof theory in general, the interpretation
of logics into the models determine a notion of equality between derivations
which we seek to understand and respect (e.g. during cut elimination).
This is the point of view of polarised logics \citep{girardnew},
which can still be described and analysed using the tools of category
theory, motivating a theory of non-associative categories developed
in \citet{munchduploids,MMMM2026} and which we partially recall and
extend in this paper.\textbf{ }

\subsection{The calculus $\protect\IMALLp$: A polarised syntax for linear CBPV}

\begin{figure}[!t]
\begin{boiteombree}
\hspace*{\fill}\subfloat[\label{fig:Expressions-and-values-1-1}Terms]{
\begin{varwidth}{\linewidth}
$\addcustomstretch[1.3][-0.33em]{\begin{array}{lrclcccccccccccccc}
 &  &  &  &  &  &  & {\scriptstyle \unite} &  & {\scriptstyle \otimes} &  & {\scriptstyle \oplus\,(i\in\{1,2\})} &  & {\scriptstyle \mixarrow} &  & {\scriptstyle \with\,(i\in\{1,2\})} &  & {\scriptstyle \top/\falso}\\
\lefteqn{\text{Values:}}\\
 & V,W & \bnf & x & \bmid & \mu\varn{\alpha}.c & \bmid & () & \bmid & V\smallotimes W & \bmid & \iith(V) & \bmid & \mu(x\push\alpha).c & \bmid & \mupair{\alpha}c{\beta}{c'} & \bmid & \term V\\
\lefteqn{\text{Stacks:}}\\
 & S & \bnf & \alpha & \bmid & \mt\varp x.c & \bmid & \mt().c & \bmid & \mt(x\smallotimes y).c & \bmid & \mui xcy{c'} & \bmid & V\push S & \bmid & \pith\push S & \bmid & \init S
\end{array}}$\\
$\addcustomstretch[1.3][-0.25em]{\begin{array}{lrclcccccccccccc}
\lefteqn{\text{Expressions:}} &  &  &  &  &  & \quad\lefteqn{\text{Contexts:}} &  &  &  &  &  & \quad\lefteqn{\text{Commands:}}\\
 & t,u & \bnf & V & \bmid & \mu\varp{\alpha}.c &  & e & \bnf & S & \bmid & \mt\varn x.c &  & c & \bnf & \varn{\cut Ve}\;\bmid\;\varp{\cut tS}
\end{array}}$
\end{varwidth}

}\hspace*{\fill}

\hspace*{\fill}\subfloat[\label{fig:Reduction-rules}Reduction rules]{{\small$\addcustomstretch[1.4][-0.3em]{\begin{array}{rccl}
(R\vareps{\mt}{}): & \vareps{\cut V{\mt\vareps x{}.c}}{} & \argred R{} & c\subst Vx\\
(R\vareps{\mu}{}): & \vareps{\cut{\mu\vareps{\alpha}{}.c}S}{} & \argred R{} & c\subst S{\alpha}\\
(R\unite): & \varp{\cut{()}{\mt().c}} & \argred R{} & c\\
(R\otimes): & \varp{\cut{V\smallotimes W}{\mt(x\smallotimes y).c}} & \argred R{} & c\subst V{x\addsubst Wy}\\
(R\oplus): & \varp{\cut{\iith(V)}{\mui{x_{1}}{c_{1}}{x_{2}}{c_{2}}}} & \argred R{} & c_{i}\subst V{x_{i}}\\
(R{\mixarrow}): & \varn{\cut{\mu(x\push\alpha).c}{V\push S}} & \argred R{} & c\subst V{x\addsubst S{\alpha}}\\
(R\with): & \varn{\cut{\mupair{\alpha_{1}}{c_{1}}{\alpha_{2}}{c_{2}}}{\pith\push S}} & \argred R{} & c_{i}\subst S{\alpha_{i}}\\
 & \text{(no rules \ensuremath{R\top}, \ensuremath{R\falso})}\\
\\\end{array}}$}{\small\par}

}\hspace*{\fill}\hspace*{\fill}\subfloat[\label{fig:Extensionality-rules}Extensionality rules ]{{\small$\addcustomstretch[1.4][-0.3em]{\begin{array}{rccl}
(E\vareps{\mt}{}): & e & \argred E{} & \mt\vareps x{}.\vareps{\cut xe}{}\\
(E\vareps{\mu}{}): & t & \argred E{} & \mu\vareps{\alpha}{}.\vareps{\cut t{\alpha}}{}\\
(E\unite): & S & \argred E{} & \mt().\varp{\cut{()}S}\\
(E\otimes): & S & \argred E{} & \mt(x\smallotimes y).\varp{\cut{x\smallotimes y}S}\\
(E\oplus): & S & \argred E{} & \mui x{\varp{\cut{\ileft(x)}S}}y{\varp{\cut{\iright(y)}S}}\\
(E{\mixarrow}): & V & \argred E{} & \mu(x\push\alpha).\varn{\cut V{x\push\alpha}}\\
(E\with): & V & \argred E{} & \mupair{\alpha}{\varn{\cut V{\pleft\push\alpha}}}{\beta}{\varn{\cut V{\pright\push\beta}}}\\
(E\top): & V & \argred E{} & \term{x_{1}\smallotimes\cdots\smallotimes x_{n}}\\
(E\falso): & S & \argred E{} & \init{x_{1}{\cdots}\mkern2mu x_{n}\push\alpha}
\end{array}}$}{\small\par}

}\hspace*{\fill}
\end{boiteombree}

\caption{\label{fig:ILLpdiam-calculus-1}$\protect\IMALLp$: calculus}
\end{figure}
\begin{figure}[!t]
\begin{boiteombree}
\hspace*{\fill}\subfloat[Judgements]{\begin{minipage}[c]{0.43\textwidth}\setlist{nosep,leftmargin=*}\vspace*{2ex}

\begin{itemize}
\item {\small Judgements are:}\\
{\small$\addcustomstretch[1][0.25em]{\begin{array}{ccc}
\Gamma\vdash t\colon A\mid & \Gamma\mid e\colon A\vdash\Delta & c\colon(\Gamma\vdash\Delta)\end{array}}$}{\small\par}
\item {\small\textbf{$\Gamma$}}{\small{} is a map from a finite set of variables
to types provided with a total order $\leq_{\Gamma}$ on its domain,
notation $\Gamma=(x_{1}\colon A_{1},\dots,x_{n}\colon A_{n})$.}{\small\par}
\item {\small\textbf{$\Delta$}}{\small{} is a pair $\alpha\colon A$ of a
co-variable and a type.}{\small\par}
\item {\small Concatenation $(\Gamma,\Gamma')$ is defined when the domains
of $\Gamma$ and $\Gamma'$ are disjoint.}{\small\par}
\end{itemize}
\end{minipage}

}\hspace*{\fill}\hspace*{\fill}\subfloat[Identity]{
\begin{varwidth}{\linewidth}
$\addcustomstretch[2.6][0.0em]{\begin{array}{cc}
\axrulesp{x\colon A}{x\colon A\mid}{(\ensuremath{\vdash\,}ax)}{\vdash} & \axrulesp{\mid\alpha\colon A}{\alpha\colon A}{(ax\ensuremath{\,\vdash})}{\vdash}\\
\unrulecsep{c\colon(\Gamma,x\colon\termeps A{}}{\Delta)}{\Gamma\mid\mt\vareps x{}.c\colon\termeps A{}}{\Delta}{(\ensuremath{\vareps{\mt}{}\vdash})} & \unrulecsep{c\colon(\Gamma}{\alpha\colon\termeps A{})}{\Gamma}{\mu\vareps{\alpha}{}.c\colon\termeps A{}\mid}{(\ensuremath{\vdash\vareps{\mu}{}})}\\
\spantwo{\binrule{\Gamma}{t\colon\termeps A{}\mid}{\Gamma'\mid e\colon\termeps A{}}{\Delta}{\vareps{\cut te}{}\colon(\Gamma,\Gamma'}{\Delta)}{(cut\ensuremath{\vareps{}{}})}}
\end{array}}$
\end{varwidth}

}\hspace*{\fill}\vspace*{4ex}

\hspace*{\fill}\subfloat[\label{fig:Structure-(exchange)-1-2}Structure]{\begin{raggedright}
\begin{minipage}[t]{0.9\columnwidth}\setlist{nosep,leftmargin=*}
\begin{itemize}
\item {\small\textbf{$\Subst{\Gamma}{\Gamma'}$}}{\small{} is the set of maps
$\sigma:\dom{\Gamma}\rightarrow\dom{\Gamma'}$ satisfying $\Gamma'(\sigma(x))=\Gamma(x)$
for all $x\in\dom{\Gamma}$.}{\small\textbf{ $\Substbij{\Gamma}{\Gamma'}$}}{\small{}
is the subset of $\Subst{\Gamma}{\Gamma'}$ of maps that are bijective.}{\small\par}
\item {\small The structural rules for $\IMALLp$  are obtained by taking
$\sigma\in\Substbij{\Gamma}{\Gamma'}$ below (renaming, exchange).
By taking $\Subst{\Gamma}{\Gamma'}$ for the range of structural maps
(adding weakening and contraction) we obtain $\polsystem{\LJ}$, in
correspondence with call-by-push-value models \citep{CFM2015}.}{\small\par}
\end{itemize}
\[
\addcustomstretch[2.6][1em]{\begin{array}{lcc}
\unrulec{\Gamma\vdash t\colon A\mid}{\Gamma'\vdash t[\sigma]\colon A\mid}{(\ensuremath{\vdash\sigma})} & \unrulec{\Gamma\mid e\colon A\vdash\Delta}{\Gamma'\mid e[\sigma]\colon A\vdash\Delta}{(\ensuremath{\sigma\vdash})} & \unrulec{c\colon(\Gamma\vdash\Delta)}{c[\sigma]\colon(\Gamma'\vdash\Delta)}{(\ensuremath{\sigma})}\end{array}}
\]
\end{minipage}
\par\end{raggedright}
}\hspace*{\fill}\vspace*{-1ex}

\hspace*{\fill}\subfloat[Logic]{\noindent\begin{minipage}[t]{1\columnwidth}\hspace*{\fill}$\addcustomstretch[2.6][1em]{\begin{array}{cc}
\binrule{\Gamma}{V\colon A\mid}{\Gamma'\mid S\colon B}{\Delta}{\Gamma,\Gamma'\mid V\push S\colon A\mixarrow B}{\Delta}{(\ensuremath{{\mixarrow}\focvdash})} & \unrulecsep{c\colon(\Gamma,x\colon A}{\alpha\colon B)}{\Gamma}{\mu(x\push\alpha).c\colon A\mixarrow B\mid}{(\ensuremath{\vdash{\mixarrow}})}\\
\binrule{\Gamma}{V\colon A\mid}{\Gamma'}{W\colon B\mid}{\Gamma,\Gamma'}{V\smallotimes W\colon A\otimes B\mid}{(\ensuremath{\focvdash\otimes})} & \unrulecsep{c\colon(\Gamma,x\colon A,y\colon B}{\Delta)}{\Gamma\mid\mt(x\smallotimes y).c\colon A\otimes B}{\Delta}{(\ensuremath{\otimes\vdash})}\\
\binrule{c\colon(\Gamma}{\alpha\colon A)}{c'\colon(\Gamma}{\beta\colon B)}{\Gamma}{\mupair{\alpha}c{\beta}{c'}\colon A\with B\mid}{(\ensuremath{\vdash\with})} & \unrulecsep{\Gamma\mid S\colon A_{i}}{\Delta}{\Gamma\mid\pith\push S\colon A_{1}\with A_{2}}{\Delta}{(\ensuremath{\with_{i}\focvdash})}\\
\binrule{c\colon(\Gamma,x\colon A}{\Delta)}{c'\colon(\Gamma,y\colon B}{\Delta)}{\Gamma\mid\mui xcy{c'}\colon A\oplus B}{\Delta}{(\ensuremath{\oplus\vdash})} & \unrulecsep{\Gamma}{V\colon A_{i}\mid}{\Gamma}{\iith(V)\colon A_{1}\oplus A_{2}\mid}{(\ensuremath{\focvdash\oplus_{i}})}
\end{array}}$\hspace*{\fill}

\hspace*{\fill}$\addcustomstretch[2.6][-0.00em]{\begin{array}{cccc}
\axrulesp{}{()\colon\unite\mid}{(\ensuremath{\vdash\unite})}{\Gamma\vdash} & \unrulecsep{c\colon(\Gamma}{\Delta)}{\Gamma\mid\mt().c\colon\unite}{\Delta}{(\ensuremath{\unite\vdash})} & \unrulecsep{\Gamma}{V\colon A\mid}{\Gamma}{\term V\colon\top\mid}{(\ensuremath{\focvdash\top})} & \unrulecsep{\Gamma\mid S\colon A}{\Delta}{\Gamma\mid\init S\colon\falso}{\Delta}{(\ensuremath{\falso\focvdash})}\end{array}}$\hspace*{\fill}\end{minipage}

}\hspace*{\fill}
\end{boiteombree}
\caption{\label{fig:ILLpdiam-typing-1}$\protect\IMALLp$: type system}
\end{figure}
\begin{figure}[!t]
\begin{boiteombree}
\vspace{1.5ex}
\hspace*{\fill}
\begin{varwidth}{\linewidth}
$\addcustomstretch[2.3][-0.0em]{\begin{array}{cclll}
\binruledouble{\Gamma}{t\colon A\mid}{\Gamma'\mid e\colon B}{\Delta}{\Gamma,\Gamma'\mid t\push e\colon A\mixarrow B}{\Delta}{(\ensuremath{{\mixarrow}\vdash})} & \quad & \termeps t1\push\termeps e2 & \df & \mt\varn x.\vareps{\cut{\mu\vareps{\alpha}2.\vareps{\cut t{\mt\vareps y1.\varn{\cut x{y\push\alpha}}}}1}e}2\\
\binruledouble{\Gamma}{t\colon A\mid}{\Gamma'}{u\colon B\mid}{\Gamma,\Gamma'}{t\smallotimes u\colon A\otimes B\mid}{(\ensuremath{\vdash\otimes})} &  & \termeps t1\smallotimes\termeps u2 & \df & \mu\varp{\alpha}.\vareps{\cut t{\mt\vareps x1.\vareps{\cut u{\mt\vareps y2.\varp{\cut{x\smallotimes y}{\alpha}}}}2}}1\\
\unrulecsepdouble{\Gamma}{t\colon A_{i}\mid}{\Gamma}{\iith(t)\colon A_{1}\oplus A_{2}\mid}{(\ensuremath{\vdash\oplus_{i}})} &  & \iith(\termeps t{}) & \df & \mu\varp{\alpha}.\vareps{\cut t{\mt\vareps x{}.\varp{\cut{\iith(x)}{\alpha}}}}{}\\
\unrulecsepdouble{\Gamma\mid e\colon A_{i}}{\Delta}{\Gamma\mid\pith\push e\colon A_{1}\with A_{2}}{\Delta}{(\ensuremath{\with_{i}\vdash})} &  & (\pith\push\termeps e{}) & \df & \mt\varn x.\cut{\mu\vareps{\alpha}{}.\varn{\cut x{\pith\push\alpha}}}e\\
\axrulecdouble{\Gamma\mid\init{\dom{(\Gamma,\Delta)}}\colon\falso\vdash\Delta}{(\ensuremath{\falso\vdash})} &  & \mathrlap{\init{\{x_{1},\dots,x_{n},\alpha\}}\quad{\df}\quad\init{x_{1}\push\dots\push x_{n}\push\alpha}}\\
\axrulecdouble{\Gamma\vdash\term{\dom{\Gamma}}\colon\top\mid}{(\ensuremath{\vdash\top})} &  & \mathrlap{\term{\{x_{1},\dots,x_{n}\}}\quad{\df}\quad\term{x_{1}\smallotimes\cdots\smallotimes x_{n}}}
\end{array}}$
\end{varwidth}
\hspace*{\fill}
\end{boiteombree}

\caption{\label{fig:ILLpdiam-derived-1}Some derivable rules in $\protect\IMALLp$.}
\end{figure}

$\IMALLp$ (\textbf{polarised $\system{IMALL}$}) is a term calculus
that can be seen as a linear and direct-style version of Levy's CBPV
\citep{Levy2004}. We recall and adapt its definition from \citealp{CFM2015}
in \prettyref{fig:ILLpdiam-calculus-1} and \ref{fig:ILLpdiam-typing-1}.\looseness=-1

The calculus $\IMALLp$ is defined à la\emph{ }Barendregt, which is
to say that we start with an untyped syntax of pseudo-terms à la Curry
(\prettyref{fig:ILLpdiam-calculus-1}). Out of these pseudo-terms,
the (proper) terms are those that possess a valid typing derivation
in the sense of \prettyref{fig:ILLpdiam-typing-1}. The equational
theory $\argeq{RE}{}$ is obtained as the well-typed restriction of
the compatible equivalence closure of reductions $\argred R{}$ (\prettyref{fig:Reduction-rules},
$\beta$-like) and expansions $\argred E{}$ (\prettyref{fig:Extensionality-rules},
$\eta$-like). The calculus $\IMALLp$ enjoys the good properties
that follow \citep[for more details see][]{CFM2015,Munch-Maccagnoni2017curry}.

From the point of view of typing, $\IMALLp$ expresses the logic $\system{IMALL}$,
which is evident from \prettyref{fig:ILLpdiam-typing-1} and \ref{fig:ILLpdiam-derived-1}.
The calculus $\IMALLp$ and its conversion rules can be interpreted
in any linear CBPV model along the lines of the proof-theoretic interpretation
of \prettyref{prop:The-deductive-system}. The interpretation into
linear CBPV models is sound: any two terms related by $\argeq{RE}{}$
have the same interpretation in all models.

On the proof-theoretic side, the compatible closure $\argRed R{}$
of reduction $\argred R{}$ is confluent. This is due to the fact
that the reduction relation $\argred R{}$ defines an orthogonal higher-order
rewriting system (left-linear, no critical pairs). In fact, $\argRed R{}$
defines a noetherian (confluent, normalising) and semantics-preserving
cut-elimination strategy for $\system{IMALL}$. Furthermore integrating
expansions into the cut-elimination results, the $\argRed R{}$-normal
and $\argRed E{}$-long terms are in correspondence with focused normal
forms, from which follows the completeness of focusing proof search
(see Appendix \ref{sec:Focused-normal-forms}). In this sense, $\IMALLp$
realises a Curry-Howard interpretation of focusing in correspondence
with linear CBPV.

\subsection{An explanation of $\protect\IMALLp$ from its proof-theoretic origins}

While these good properties make of $\IMALLp$ a syntax for linear
CBPV, they say very little about how it has been designed. Indeed,
it has not been designed initially with the goal of being a good syntax
for linear CBPV. To explain the origins of $\IMALLp$, we need to
go back to a different work on polarisation than Girard's, to Danos,
Joinet and Schellinx's work on proof-relevant $\LK$ \citep{danos95new}.
In this work, DJS investigated various ways of endowing $\LK$ with
a noetherian (terminating and confluent) cut-elimination procedure
by means of translations into linear logic. This cut-elimination procedure
is also semantics-preserving by construction, for the semantics determined
by the translation into linear logic. With the system $\polsystem{LK}^{\eta}$,
DJS have found a canonical system, in fact related to Girard's original
polarised classical sequent calculus $\LC$: canonical in the sense
that its translation into linear logic uses the least number of modalities
$\oc$ and $\wn$, and that other known classical systems can be expressed
in it.

While $\polsystem{LK}^{\eta}$ was originally formulated without a
term syntax, the same methodology was applied again using the Curien-Herbelin
syntactic technology ($\bar{\lambda}\mu\tilde{\mu}$) for sequent
calculus \citep{CH00Duality,Herb05} in \citet{munchmonolateral}.
The calculus $\IMALLp$ is therefore essentially a linear and intuitionistic
version of $\polsystem{LK}^{\eta}$ developed with the Curien-Herbelin
syntactic technology. In order to understand $\IMALLp$, it is therefore
crucial to first understand DJS's approach to the proof theory of
$\LK$.

The main idea behind DJS's noetherian cut-elimination is how they
approach critical pairs encountered during the cut elimination of
$\LK$ (such as the so-called ``Lafont'' critical pair involving
a cut between two formulae introduced by weakening). \citealt{danos95new}
resolves these critical pairs in a way which is persistent (in fact
determined by a typing annotation\footnote{The first and most general system introduced in \citealt{danos95new},
$\LK^{tq}$, has explicit and arbitrary annotations ($t$, or $-$)
and ($q$, or $+$) on formulæ. One can see this as being equivalent
to a presentation with some default direction determined for each
connective, together with ``polarity-shifting'' unary connectives
imposing a ($-$) or ($+$) direction.}), in the sense that it is fixed during cut-elimination, as a way
of avoiding looping situations. These directions are and either towards
the \emph{``head''} ($-$) or towards the \emph{``tail''} ($+$),
in a terminology inspired by, and related to, that of ``head reduction''
in the $\lambda$ calculus.

In the term syntax, annotations ($-$) or ($+$) correspond to opposite
ways of determining priority in a cut:
\begin{align*}
\varn{\cut{\mu\varn{\alpha}.c}{\mt\varn x.c'}} & \argred R{}c'\subst{\mu\varn{\alpha}.c}x\\
\varp{\cut{\mu\varp{\alpha}.c}{\mt\varp x.c'}} & \argred R{}c\subst{\mt\varp x.c'}{\alpha}
\end{align*}
Starting from this idea, it is possible to understand $\polsystem{LK}^{\eta}$
as the system whose design follows the principles of \emph{polarisation
}and \emph{focalisation}. To understand these design principles and
where they come from, assume that we take as primitive the notions
of 1) local reduction of principal cuts, and 2) $\eta$-expansion.
Then, we ask that expansions ``play well'' with reductions, and
derive consequences of this choice.

For instance in $\system{IMALL}$, one expects that a principal cut
for $\oplus$ on the left-hand side below reduces into, or at least
has the same reduction behaviour, as the cut on the right-hand side
below.
\[
\addcustomstretch[2][1em]{\begin{array}{cc}
\bussproof{\binaryinfc{\unaryinfc{\axiomc{\Gamma\vdash A}}{\Gamma\vdash A\oplus B}{}}{\binaryinfc{\axiomc{\Gamma',A\vdash\Delta}}{\axiomc{\Gamma',B\vdash\Delta}}{\Gamma',A\oplus B\vdash\Delta}{}}{\Gamma',\Gamma\vdash\Delta}{}} & \bussproof{\binaryinfc{\axiomc{\Gamma\vdash A}}{\axiomc{\Gamma',A\vdash\Delta}}{\Gamma',\Gamma\vdash\Delta}{}}\\
\cut{\ileft(t)}{\mui xcy{c'}} & \cut t{\mt x.c}
\end{array}}
\]
Similarly, and as importantly in the point of view of proof-relevant
systems such as $\polsystem{LK}^{\eta}$ and $\IMALLp$, an arbitrary
derivation for a sequent such that the one on the left-hand side below
must have the same reduction behaviour (or meaning, in some sense),
as its $\eta$-expanded version on the right-hand side below.
\[
\addcustomstretch[2][1em]{\begin{array}{cc}
\Gamma,A\oplus B\vdash\Delta & \bussproof{\binaryinfc{\binaryinfc{\unaryinfc{\unaryinfc{\axiomc{}}{A\vdash A}{}}{A\vdash A\oplus B}{}}{\axiomc{\Gamma,A\oplus B\vdash\Delta}}{\Gamma,A\vdash\Delta}{}}{\binaryinfc{\unaryinfc{\unaryinfc{\axiomc{}}{B\vdash B}{}}{B\vdash A\oplus B}{}}{\axiomc{\Gamma,A\oplus B\vdash\Delta}}{\Gamma,B\vdash\Delta}{}}{\Gamma,A\oplus B\vdash\Delta}{}}\\
e & \mui x{\cut{\ileft(x)}e}y{\cut{\iright(y)}e}
\end{array}}
\]
This assumption has two main consequences on the choices made during
cut-elimination, indeed polarisation and focalisation:
\begin{itemize}
\item \textbf{Polarisation.} We would like that an expansion leaves the
reduction order unchanged. Therefore when an expansion is available
and seems to force a particular reduction, as in the following example
of a cut of type $\oplus$:
\begin{align*}
\vareps{\cut{\mu\vareps{\alpha}{}.c}{\mt\vareps x{}.c'}}{} & \argRed E{}\vareps{\cut{\mu\vareps{\alpha}{}.c}{\mui x{\cut{\ileft(x)}{\mt\vareps x{}.c'}}y{\cut{\iright(y)}{\mt\vareps x{}.c'}}}}{}\\
 & \argred R{}c\subst{\mui x{\cut{\ileft(x)}{\mt\vareps x{}.c'}}y{\cut{\iright(y)}{\mt\vareps x{}.c'}}}{\alpha}\\
 & \argRedinvs E{}c\subst{\mt\vareps x{}.c'}{\alpha}\,,
\end{align*}
then this must be the reduction that was determined all along by the
annotation:
\begin{align*}
\vareps{\cut{\mu\vareps{\alpha}{}.c}{\mt\vareps x{}.c'}}{} & \argred R{}c\subst{\mt\vareps x{}.c'}{\alpha}\quad\text{\emph{i.e. }}\varepsilon=+
\end{align*}
Given that the connectives $\unite,\otimes,\oplus,\falso$ have their
expansions on the right-hand side of contexts as above, and the connectives
$\mixarrow,\with,\top$ have their expansion on the left-hand side
of expressions, the former are ($+$) and the latter are ($-$).
\item \textbf{Focalisation.} An expansion can reveal reductions as in the
following example, again of a cut of type $\oplus$:
\begin{align*}
\varp{\cut{\ileft(\termp t)}{\mt\varp x.c}} & \argReds{RE}{}\varp{\cut{\ileft(\termp t)}{\mui{\vphantom{\cut{}{}}y}{c\subst{\ileft(y)}x}z{c\subst{\iright(z)}x}}}\\
 & \argReds R{}\varp{\cut{\termp t}{\mt\varp y.c\subst{\ileft(y)}x}}
\end{align*}
where the reduction in last line is desired because it corresponds
to a local reduction of a principal cut for $\oplus$. This is an
evidence that $\ileft(\termp t)$ contains a hidden cut, which we
require to reveal itself without the intervention of an expansion:
\begin{align*}
\varp{\cut{\ileft(\termp t)}{\mt\varp x.c}} & \argReds R{}\varp{\cut{\termp t}{\mt\varp y.c\subst{\ileft(y)}x}}
\end{align*}
In $\IMALLp$, this is obtained by restricting the primitive form
of constructors to values and co-values (e.g. $\ileft(V)$), which
forces unrestricted forms (such as $\ileft(\termp t)$) to be derived
as in fig.~\ref{fig:ILLpdiam-derived-1}, using a cut:
\[
\ileft(\termp t)=\mu\varp{\alpha}.\varp{\cut t{\mt\varp x.\varp{\cut{\iith(x)}{\alpha}}}}\,.
\]
In terms of focused (\emph{i.e.} normal) proofs, focalisation is responsible
for the atomicity of synchronous phases.
\end{itemize}
$\IMALLp$ is obtained by starting from the Curien-Herbelin syntactic
toolkit and applying these ideas: in this sense it is a proof-relevant
$\system{IMALL}$ whose theory of noetherian proof reduction and semantics-preserving
proof equivalence derives from the principles of polarisation and
focalisation we just described.

The strength of this Curry-Howard take on Andreoli's focusing will
become manifest in \prettyref{subsec:The-calculus-}, as we will see
that the same design principles still work for the modalities $\oc$
and $\mon$, although the result is different: in particular an interpretation
of $\oc$ and $\mon$ as relative monads and comonads will appear.

\subsection{The syntactic model of linear and thunkable terms}

The syntax in fact defines a syntactic model of linear CBPV, via a
construction for which we need the following definition first.
\begin{defn}
A morphism $f$ in a non-associative category is \textbf{thunkable}
if any composition chain $h\ccomp g\ccomp f$ associates. A morphism
$h$ in a non-associative category is \textbf{linear} if any composition
chain $h\ccomp g\ccomp f$ associates.
\end{defn}

Thunkable and linear morphisms determine (associative) sub-categories
of the non-associative category. We call $\Pcategory_{t}$ the subcategory
of positive objects and thunkable maps of $\IMALLp$ seen as a deductive
system with $\argeq{RE}{}$ as equality of proofs, and $\Ncategory_{l}$
its subcategory of negative objects and linear maps. An advantage
of considering thunkable and linear maps is that we do not need to
define so-called ``complex'' values and stacks as in \citet{Levy2004}.
\begin{prop}
The $\IMALLp$ calculus defines a syntactic linear CBPV model
\[
\xymatrix{\Pcategory_{t}\xyport{\xyarrow[rr][][][1.5pc]}{\unite\mixarrow-}{} & \bot & \xyport{\xyarrow[ll][][][1.5pc]}{\unite\otimes-}{}\Ncategory_{l}}
\]
with the symmetric monoidal and distributive structure on $\Pcategory_{t}$
given by the connectives $\otimes,\unite$ and $\oplus,\falso$, and
the cartesian closure relative to $\Pcategory_{t}$ on $\Ncategory_{l}$
given by the connectives $\with,\top$ and $\mixarrow$.
\end{prop}

We therefore have two ways of organising $\IMALLp$ categorically
with cut as composition and $\argeq{RE}{}$ as equality of morphisms:
1) as deductive systems, i.e. non-associative categories, and 2) as
linear CBPV model. The two interpretations turn out to be equivalent,
in the sense that the non-associative category arising from the adjunctions
between positive thunkable and negative linear terms is equivalent
to the original non-associative category \citep{munchduploids}.\footnote{The syntactic result suggests in fact that the structure theorem from
\citet{Fuhrmann2000PhD,munchduploids} extends to linear CBPV, which
is a result that we expect.}

\subsection{\label{subsec:Adjunctions-in-non-associative}Adjunctions in non-associative
categories}

The notion that reductions and expansions have to ``play well''
together is traditionally characterised in categorical semantics with
the presence of adjunctions, in which connectives are left or right
adjoints. One recent realisation is that this characterisation already
makes sense in bare deductive systems: non-associative categories.
Let us recall the notion of adjunction between bare functors over
non-associative categories introduced in \citealt{MMMM2026}.
\begin{defn}
A \textbf{bare functor}\emph{ }$F:\Dduploid\graphto\Eduploid$ between
non-associative categories $\Dduploid$ and $\Eduploid$ is a morphism
of the reflexive graphs $\Dduploid$ and $\Eduploid$, in other words
a bare functor $F:\Dduploid\graphto\Eduploid$ is a functor that does
not necessarily preserve composition. A \textbf{proper}\emph{ }functor
is one that does (notation $\Dduploid\rightarrow\Eduploid$).
\end{defn}

\begin{defn}
\label{def:An-adjunction}An\emph{ }\textbf{adjunction} over non-associative
categories $\Dduploid$ and $\Eduploid$ \[
    \begin{tikzcd}[column sep = 1em]
      \Dduploid  \arrow[rr,"F",harpoon,  bend left = 40]  & \bot & \arrow[ll,"G",harpoon, bend left = 40] \Eduploid
    \end{tikzcd}
\] consists of two bare functors $F:\Dduploid\graphto\Eduploid$ and
$G:\Eduploid\graphto\Dduploid$ and a family of bijections: 
\[
\varphi_{A,B}:\Eduploid(FA,B)\cong\Dduploid(A,GB)
\]
natural separately in $A\in\Dduploid$ and $B\in\Eduploid$, that
is to say, for any composition chain $g\ccomp h\ccomp Ff$ in $\Eduploid$
(not necessarily associative), where $h$ is an oblique morphism,
one has $\varphi(g\ccomp h)=Gg\ccomp\varphi(h)$ and $\varphi(h\ccomp Ff)=\varphi(h)\ccomp f$
for any oblique morphism $f$.
\end{defn}

Note that an adjunction over associative categories is an adjunction
in the usual sense, since the bare functors are necessarily proper
in that case.

\begin{defn}
For any pair of non-associative categories $\Dduploid,\Eduploid$,
their \textbf{direct product} $\Dduploid\times\Eduploid$ is the non-associative
category with objects $\obj{\Dduploid}\times\obj{\Eduploid}$ and
morphisms $(\Dduploid\times\Eduploid)((A,A'),(B,B'))=\Dduploid(A,B)\times\Eduploid(A',B')$
defined in the obvious way. For $\Dduploid$ a non-associative category,
we define the \textbf{diagonal functor} $\Delta:\Dduploid\graphto\Dduploid\times\Dduploid$
with $\Delta A=(A,A)$ and $\Delta f=(f,f)$.
\end{defn}

Note that $\Delta$ preserves both identities and composition; it
is in fact a proper functor.
\begin{prop}
In $\IMALLp$ considered as a non-associative category $\Dduploid$,
the diagonal functor $\Delta$ has a left adjoint $\hat{\oplus}:\Dduploid\times\Dduploid\graphto\Dduploid$
with $A\mathbin{\hat{\oplus}}B=A\oplus B$ on objects, satisfying
a family of isomorphisms:
\[
(A\oplus B\vdash C)\cong(A\vdash C)\times(B\vdash C)
\]
natural in $(A,B)\in(\Dduploid\times\Dduploid)^{\op}$ and in $C\in\Dduploid$
separately.

The functor $\Delta$ also has a right adjoint $\hat{\with}:\Dduploid\times\Dduploid\graphto\Dduploid$
with $A\mathbin{\hat{\with}}B=A\with B$ on objects, satisfying a
family of isomorphisms:
\[
(A\vdash B\with C)\cong(A\vdash B)\times(A\vdash C)
\]
natural in $A\in\Dduploid{}^{\op}$ and in $(B,C)\in\Dduploid\times\Dduploid$
separately.
\end{prop}

More details are given in the case of sums in \prettyref{sec:IMALLp-has-sums}
(the case of the product $\with$ is symmetric). The relevant structure
of $\otimes$ and $\mixarrow$ is described in \citealt{MMMM2026};
we omit it in this paper for simplicity.

\section{$\protect\ILL^{\protect\mon}$: Resource and effect modalities}

\begin{figure}
\begin{boiteombree}
\hspace*{\fill}\subfloat[Types]{$\addcustomstretch[1.4][-0.2em]{\begin{array}{lrl}
\text{positive} & P,Q,\termp A & \bnf\dots\bmid\oc A\\
\text{negative} & N,M,\termn A & \bnf\dots\bmid\mon A
\end{array}}$

}\hspace*{\fill}\hspace*{\fill}\subfloat[Judgements]{
\begin{varwidth}{\linewidth}
\vspace*{\bigskipamount}
\setlist{nosep}
\begin{itemize}
\item {\small$\oc\Gamma$ stands for a typing context $\oc A_{1},\dots,\oc A_{n}$.}{\small\par}
\end{itemize}
\end{varwidth}
\vspace*{\medskipamount}

}\hspace*{\fill}\vspace*{-2ex}

\hspace*{\fill}\subfloat[\label{fig:Structure-(exchange)-1-1-1}Structure]{\begin{raggedright}
\begin{varwidth}{\linewidth}
\begin{flushleft}
$\addcustomstretch[2.6]{\begin{array}{cc}
\unrulec{\Gamma\vdash B\vphantom{,}}{\Gamma,\oc A\vdash B}{(w)} & \unrulec{\Gamma,\oc A,\oc A\vdash B}{\Gamma,\oc A\vdash B}{(c)}\end{array}}$
\par\end{flushleft}
\end{varwidth}
\par\end{raggedright}
}\hspace*{\fill}\hspace*{\fill}\subfloat[Logic]{$\addcustomstretch[2.2][-0.0em]{\begin{array}{cc}
\unrulec{\oc\Gamma\vdash A\vphantom{,}}{\oc\Gamma\vdash\oc A}{(\ensuremath{\vdash\oc})} & \unrulecsep{\Gamma,A}B{\Gamma,\oc A}B{(\ensuremath{\oc\vdash})}\\
\unrulecsep{\Gamma}{A\vphantom{,}}{\Gamma}{\mon A}{(\ensuremath{\vdash\mon})} & \unrulec{\oc\Gamma,A\vdash\mon B}{\oc\Gamma,\mon A\vdash\mon B}{(\ensuremath{\mon\vdash})}
\end{array}}$

}\hspace*{\fill}
\end{boiteombree}

\caption{\label{fig:ILLdiam}$\protect\ILL^{\protect\mon}=\protect\system{IMALL}$
+ the rules above}
\end{figure}

In this section, we extend the previous development, starting from
a sequent calculus for multiplicative additive intuitionistic linear
logic $\system{IMALL}$ now with a resource or exponential modality
noted $A\mapsto\oc A$; in other words the full propositional intuitionistic
linear logic $\system{ILL}$, what we write 
\[
\addcustomstretch[0][-0.12em]{\begin{array}{ccccccc}
\system{ILL} &  & = &  & \system{IMALL} & + & \oc\end{array}}
\]
We also consider the combination of $\system{ILL}$ with an effect
modality noted $A\mapsto\mon A$ required to be strong with respect
to the resource modality, in the sense that it comes with the axiom
schema 
\[
\oc A\otimes\mon B\quad{\vdash}\quad\mon\,(\oc A\otimes B)
\]
We obtain in that way the intuitionistic modal logic 
\[
\addcustomstretch[0][-0.12em]{\begin{array}{ccccccccccccccc}
\system{ILL}^{\mon} &  & = &  & \system{ILL} & + & \mon &  & = &  & \system{IMALL} & + & \oc & + & \mon\end{array}}
\]
whose sequent calculus is presented in \prettyref{fig:ILLdiam}. Note
that $\system{ILL}^{\mon}$ may be seen as a linear variant of the
constructive modal logic $\system{CS4}$ \citep{Alechina2001} where
we find convenient (and in the tradition of linear logic) to write
$A\mapsto\,\,!\,A$ for the necessity modality $A\mapsto\square A$.
We replay the developments from the previous section.

\subsection{Linear CBPV models with resource and effect modalities\label{subsec:Linear-CBPV-models}}

As in \prettyref{prop:The-deductive-system} for $\system{IMALL}$,
there is an interpretation of linear CBPV models as deductive systems
that validate the rules of \prettyref{fig:ILLdiam}. Let us first
recall that a linear CBPV model with a resource modality \citep{CFM2015}
is given by a linear CBPV adjunction \eqref{eq:lcbpv-adj} together
with a linear-non-linear adjunction\begin{equation}\label{eq:lcbpv-resource}    \begin{tikzcd}
      \Mcategory \arrow[rr, bend left, "\Lin"] & \bot & \Vcategory \arrow[ll, bend left, "\Mult"] \arrow[rr, bend left, "\shneg"] & \bot & \Scategory \arrow[ll, bend left, "\shpos"]
    \end{tikzcd}\end{equation}that is, a symmetric monoidal adjunction as above where
the monoidal structure on $\Mcategory$ is cartesian, leading to:
\begin{prop}[\citealp{CFM2015}]
\label{prop:The-deductive-system-1}The deductive system $\duploid{\shneg}{\shpos}$
associated to a linear CBPV adjunction with resource modality~\eqref{eq:lcbpv-resource}
satisfies the rules of $\system{ILL}$, that is to say the rules of
\prettyref{fig:ILLdiam} without the modality $\mon$.
\end{prop}

Interpreting all the rules from fig.\ \ref{fig:ILLdiam} in an analogous
of \prettyref{prop:The-deductive-system-1} for of $\system{ILL}^{\mon}$
is possible and requires an additional structure in the form of an
effect modality: something like a monad $T=G\circ F$ on $\Scategory$
which is strong with respect to an action of $\Mcategory^{\op}$ on
$\Scategory$ \citep{Mellies2012parametric} leading to a series of
``strong'' adjunctions\begin{equation}\label{eq:lcbpv-resource-effect}  \begin{tikzcd}
    \Mcategory \arrow[rr, bend left, "\Lin"] & \bot & \Vcategory \arrow[rr, "\shneg", bend left]\arrow[ll, "\Mult", bend left] & \bot & \Scategory \arrow[rr, "F", bend left]\arrow[ll, "\shpos", bend left]  & \bot & \Ccategory \arrow[ll, "G", bend left]
  \end{tikzcd}\end{equation}In order to make it no more complicated than necessary
for our purposes, here with $F\dashv G$ being strong we mean that
$T$ is provided with a natural family of morphisms in $\Vcategory$,
\[
\lstr_{X,P}:\Lin X\otimes\shpos T\shneg P\rightarrow\shpos T\shneg(\Lin X\otimes P)\,,
\]
making four usual diagrams commute—corresponding to the notion of
\emph{strength with respect to $\Lin:\Mcategory\rightarrow\Vcategory$}
\citep{Blute1996a,Hasegawa2002} for the monad $\shpos\circ T\circ\shneg$
on $\Vcategory$ obtained by composing adjunctions. This notion might
seem somewhat restrictive, and it is going to be all the more interesting
that this is enough for our purposes.
\begin{example}
~
\begin{itemize}
\item Any (non-polarised) model of intuitionistic linear logic, that is
to say a situation \eqref{eq:lcbpv-resource} in which the symmetric
monoidal category $\Vcategory$ is closed, $\Scategory=\Vcategory$,
and $\shneg=\shpos=\Id_{\Vcategory}$, has a situation \eqref{eq:lcbpv-resource-effect}
in which $T$ is the exception monad $\Emon=\mathord{-}\oplus E$
(for any $E\in\Vcategory$), which is strong with respect to $\Lin$,
in the sense that the linear tautologies 
\[
\oc\Gamma\otimes(A\oplus E)\vdash(\oc\Gamma\otimes A)\oplus E
\]
 underlie a natural family of morphisms
\[
\Lin X\otimes\Emon A\rightarrow\Emon(\Lin X\otimes A)
\]
satisfying the four coherence laws. Diagrammatically:\[
  \begin{tikzcd}
    \Mcategory \arrow[rr, bend left, "\Lin"] & \bot & \Vcategory \arrow[ll, "\Mult", bend left] \arrow[rr, "F^{\Emon}", bend left] & \bot & \Vcategory^{\Emon} \arrow[ll, "G^{\Emon}", bend left]
  \end{tikzcd}
\]
\item Any model of linear logic, that is to say \eqref{eq:lcbpv-resource}
in which $\Vcategory$ is $*$-autonomous, $\Scategory=\Vcategory$,
and $\shneg=\shpos=\Id$, has a situation \eqref{eq:lcbpv-resource-effect}
in which $T=\wn$ is the dual of $\oc$, $\Ccategory=\Mcategory^{\op}$,
$F=\Mult\circ\mathord{\neg}$ and $G={\neg}\circ\Lin$, and in which
the linear tautologies 
\[
\oc\Gamma\otimes\wn A\vdash\wn(\oc\Gamma\otimes A)
\]
 underlie again a strength of $\wn$ in the previous sense.\[
  \begin{tikzcd}
    \Mcategory \arrow[rr, bend left, "\Lin"] & \bot & \Vcategory \arrow[ll, "\Mult", bend left] \arrow[rr, "\Mult\circ {\neg}", bend left]  & \bot & \Mcategory^\op \arrow[ll, "{\neg}\circ \Lin", bend left]
  \end{tikzcd}
\]
\item By not assuming $\Vcategory$ to be $*$-autonomous but merely a dialogue
category (a symmetric monoidal category having a negation functor
$\neg:\Vcategory\rightarrow\Vcategory^{\op}$, in the sense of having
an exponentiable object), we arrive at a linear CBPV model with a
resource modality \eqref{eq:lcbpv-resource} with $\Scategory=\Vcategory^{\op}$
and with the adjunction of negation with itself. Such models give
again rise to $\ILLpdiam$ models \eqref{eq:lcbpv-resource-effect}
with $\Ccategory=\Mcategory^{\op}$, $F=\Mult$ and $G=\Lin$.\begin{equation}\label{eq:adj-dialogue-resource}  \begin{tikzcd}
    \Mcategory \arrow[rr, bend left, "\Lin"] & \bot & \Vcategory \arrow[rr, "\neg", bend left]\arrow[ll, "\Mult", bend left] & \bot & \Vcategory^\op \arrow[rr, "\Mult", bend left]\arrow[ll, "\neg", bend left]  & \bot & \Mcategory^\op \arrow[ll, "\Lin", bend left]
  \end{tikzcd}\end{equation}This situation has been used to explain translations
of (polarised) classical logic into linear logic and to describe game
models of linear logic in \citet{melliestabareau}.
\end{itemize}
\end{example}

\begin{prop}
The deductive system $\duploid{\shneg}{\shpos}$ associated to a linear
CBPV adjunction with resource and effect modalities \eqref{eq:lcbpv-resource-effect}
satisfies the rules of $\system{ILL}^{\mon}$ of \prettyref{fig:ILLdiam}.
\end{prop}

Up until now, we did not detail the interpretation of formulæ of $\system{IMALL}$
as objects of linear CBPV models since it was rather straightforward.
However, the interpretation is no longer straightforward for $\oc$
and $\mon$. Indeed, the valid interpretation of the rules for $\oc$
and $\mon$ require an extra polarity shift:
\begin{equation}
\begin{aligned}\oc A & =\mathord{\Lin}\circ\mathord{\Mult}\circ\shpos(\varn A)\\
\mon A & =T\circ\shneg(\varp A)
\end{aligned}
\label{eq:modal-interp}
\end{equation}
This phenomenon corresponds to the fact that starting from an oblique
morphism
\[
\mathscr{O}(\Lin X_{1}\otimes\cdots\otimes\Lin X_{n},\shneg P)
\]
one cannot in general obtain an oblique morphism 
\[
\mathscr{O}(\Lin X_{1}\otimes\cdots\otimes\Lin X_{n},\shneg\Lin\Mult P)\,,
\]
merely one in 
\begin{equation}
\mathscr{O}(\Lin X_{1}\otimes\cdots\otimes\Lin X_{n},\shneg\Lin\Mult\shpos\shneg P)\,.\label{eq:oblique-comonad-shift}
\end{equation}
This phenomenon was described semantically in \citet{CFM2015} but
was observed syntactically as early as \citet{Andreoli92} with the
fact that the introduction of $\oc$ terminates a synchronous phase
and that the introduction of $\wn$ (or $\oc$ on the left) terminates
an asynchronous phase in focusing proof search. The monad $T$ undergoes
a dual phenomenon. But these additional shifts in \eqref{eq:modal-interp}
have consequences on the completeness of the interpretation, as we
will see.

\subsection{\label{subsec:The-calculus-}The calculus $\protect\ILLpdiam$: A
polarised syntax for $\protect\oc$ and $\protect\mon$ modalities}

\begin{figure*}[!t]
\begin{boiteombree}
\hspace*{\fill}\subfloat[\label{fig:Expressions-and-values-1}Terms]{$\addcustomstretch[1.4][-0.0em]{\begin{array}{lrcccccc}
 &  &  &  &  & {\scriptstyle \oc} &  & {\scriptstyle \mon}\\
\text{Values:} & V,W & \bnf & \dots & \bmid & \mu\prom{\alpha}.c & \bmid & \diam V\\
\text{Stacks:} & S & \bnf & \dots & \bmid & \oc S & \bmid & \mt\diam x.c
\end{array}}$

}\hspace*{\fill}

\hspace*{\fill}\subfloat[Reduction rules]{$\addcustomstretch[1.4][-0.0em]{\begin{array}{rccl}
(R\oc): & \varp{\cut{\mu\prom{\alpha}.c}{\prom S}} & \argred R{} & c\subst S{\alpha}\\
(R\mon): & \varn{\cut{\diam V}{\mt\diam x.c}} & \argred R{} & c\subst Vx
\end{array}}$

}\hspace*{\fill}\subfloat[Extensionality rules]{$\addcustomstretch[1.4][-0.0em]{\begin{array}{rcll}
(E\oc): & V & \argred E{} & \mu\prom{\alpha}.\varp{\cut V{\prom{\alpha}}}\\
(E\mon): & S & \argred E{} & \mt\diam x.\varp{\cut{\diam x}S}
\end{array}}$

}\hspace*{\fill}\vspace*{3ex}

\hspace*{\fill}\subfloat[Judgements]{\begin{minipage}[t]{0.3\columnwidth}\setlist{nosep,leftmargin=*}
\begin{itemize}
\item {\small$\oc\Gamma$ stands for a typing context $x_{1}\colon\oc A_{1},\dots,x_{n}\colon\oc A_{n}$.}{\small\par}
\item {\small$\mon\Delta$ stands for a typing context $\alpha\colon\mon A$.}{\small\par}
\end{itemize}
\end{minipage}

}\hspace*{\fill}\subfloat[Structure]{\begin{minipage}[t]{0.6\textwidth}\setlist{nosep,leftmargin=*}
\begin{itemize}
\item {\small\textbf{$\Substexp{\Gamma}{\Gamma'}$}}{\small{} is the subset
of $\Subst{\Gamma}{\Gamma'}$ of maps that are bijective on variables
not of the form $\oc A$.}{\small\par}
\item {\small The structural rules for $\ILLpdiam$ are given by taking $\Substexp{\Gamma}{\Gamma'}$
for the range of structural maps (renaming and exchange, as well as
weakening and exchange for formulae of the form $\oc A$).}
\end{itemize}
\end{minipage}

}\hspace*{\fill}\vspace*{-1ex}

\hspace*{\fill}\subfloat[Logic]{$\addcustomstretch[2.6][0.5em]{\begin{array}{cc}
\unrulec{c\colon(\oc\Gamma\vdash\alpha\colon A)}{\oc\Gamma\vdash\mu\prom{\alpha}.c\colon\oc A\mid}{(\ensuremath{\vdash\oc})} & \unrulecsep{\Gamma\mid S\colon A}{\Delta}{\Gamma\mid\prom S\colon\oc A}{\Delta}{(\ensuremath{\oc\focvdash})}\\
\unrulecsep{\Gamma}{V\colon A}{\Gamma}{\diam V\colon\mon A}{(\ensuremath{\focvdash\mon})} & \unrulec{c\colon(\oc\Gamma\mid x\colon A\vdash\mon\Delta)}{\oc\Gamma\mid\mt\diam x.c\colon\mon A\vdash\mon\Delta}{(\ensuremath{\mon\vdash})}
\end{array}}$

}\hspace*{\fill}

\hspace*{\fill}\subfloat[\label{fig:ILLpdiam-derived}Additional derived rules in $\protect\ILLpdiam$.]{
\begin{varwidth}{\linewidth}
$\addcustomstretch[2.3][-0.0em]{\begin{array}{ccrll}
\unrulecsepdouble{\Gamma\mid e\colon A}{\Delta}{\Gamma\mid\prom e\colon\oc A}{\Delta}{(\ensuremath{\oc\vdash})} &  & \oc\termeps e{} & \df & \mt\varp x.\cut{\mu\vareps{\alpha}{}.\cut x{\oc\alpha}}e\\
\unrulecsepdouble{\Gamma}{t\colon A}{\Gamma}{\diam t\colon\mon A}{(\ensuremath{\vdash\mon})} &  & \diam\termeps t{} & \df & \mu\varn{\alpha}.\cut t{\mt\vareps x{}.\cut{\diam\alpha}x}
\end{array}}$
\end{varwidth}

}\hspace*{\fill}
\end{boiteombree}

\caption{\label{fig:ILLpdiam-calculus}$\protect\ILLpdiam=\protect\IMALLp$
+ above}
\end{figure*}

The calculus $\ILLpdiam$ extending $\IMALLp$ is defined in \prettyref{fig:ILLpdiam-calculus}.
It is $\IMALLp$ with modalities $\oc$ and $\mon$ where we use a
syntax introduced for linear logic exponentials $\oc,\wn$ in \citet{munchmonolateral}.
One of the goals of this paper is to explain this somewhat unintuitive,
but mathematically justified syntax.

The calculus $\ILLpdiam$ enjoys the same good properties we mentioned
for $\IMALLp$: its type system expresses the logic $\system{ILL}^{\mon}$,
and its derivations can be soundly interpreted in any linear CBPV
model with a resource and an effect modalities in the non-relative
sense of \eqref{eq:lcbpv-resource-effect}. The reduction relation
$\argRed R{}$ defines again a noetherian and semantics-preserving
cut-elimination strategy for $\system{ILL}^{\mon}$. The $\argRed R{}$-normal
and $\argRed E{}$-long terms are again in correspondence with standard
focused normal forms, from which the completeness of focusing proof
search follows (see Appendix \ref{sec:Focused-normal-forms}). 

\subsection{An explanation of $\protect\ILLpdiam$ from a proof-theoretic perspective}

Now we turn to the design constraints related to polarisation and
focalisation that gave rise to this syntax. As those are concerned
with playing well with $\eta$ rules, the first observation is that
the $\eta$-rules for $\oc$ and $\mon$ are very different to other
connectives. For instance, in terms of mere provability, none of their
left- and right-introduction rules are invertible. There is a weaker
property of ``semi-invertibility'' observed in \citet{Lau02PhD}:
provided that the context $\Gamma$ is of the right form, the modality
$\oc$ is invertible on the right in terms of mere provability.

For a notion of $\eta$ rule which is justified semantically (not
merely in terms of provability), we have to adopt a notion of invertibility
(or $\eta$-expansion) for $\oc$ and $\mon$ which is weaker still.
It relies on the following principle: a well-typed value of type $\oc A$
must come from either a promotion or a variable (possibly with some
structural rules applied); intuitively it corresponds to a ``box''
in terms of proof nets for linear logic. In particular the context
of this value is of the form $\oc\Gamma$. The weak invertibility
principle for $\oc$ is then reflected in the syntax with a promotion
constructor in the shape of pattern matching together with a dereliction
constructor in the shape of a pattern:
\[
\addcustomstretch[0][1em]{\begin{array}{cc}
\oc\Gamma\vdash\mu\oc\alpha.c & \Gamma\mid\oc e:\oc A\vdash\Delta\end{array}}
\]
in such a way that the corresponding expansion correctly reflects
a (weak) invertibility of $\oc$ on the right:
\begin{equation}
V\argred E{}\mu\prom{\alpha}.\varp{\cut V{\prom{\alpha}}}\label{eq:oc-expansion}
\end{equation}
and that it is well-typed. Dually, a well-typed co-value of type $\mon A$
must come from either an extension or a covariable, such that its
context is of the form $\oc\Gamma,\mon\Delta$, in such a way that
the expansion
\[
S\argred E{}\mt\diam x.\varp{\cut{\diam x}S}
\]
is well-typed.\footnote{and whose interpretations into linear CBPV models are sound due to
a stronger observation: the interpretation of expressions of type
$\oc A$ that are values are $\oc$-coalgebra morphisms whereas the
interpretation of coexpressions of type $\mon A$ that are covalues
are $\mon$-algebra morphisms, although the initial motivations were
syntactic (in terms of proof net boxes for $\oc$ and $\wn$ in linear
logic, as explained).} Starting from these desired $\eta$-expansions, polarisation and
focalisation again help us design the rest of the calculus.
\begin{itemize}
\item \textbf{Polarisation.} We first notice that the weak invertibility
property is not sufficient to make $\oc$ a negative connective (and
dually to make $\mon$ to be a positive connective) since they apply
conditionally, and thus cannot be used to force a reduction. Other
constraints require $\oc$ to be positive (and dually $\mon$ negative):
namely to ensure that values of type $\oc A$ and covalues of type
$\mon A$ are ``boxes'' – (co)algebra morphisms – i.e. promotions
or variables for $\oc$, extension or covariable for $\mon$. To illustrate
this requirement, consider in a generic cut
\[
\vareps{\cut t{\mt\vareps x{}.c}}{}
\]
where $t$ is an expression of type $\oc A$, and where $x$ might
occur several or zero times in $c'$. We do not know that $t$ is
allowed to be duplicated or erased – i.e. is a box, or a !-coalgebra
morphism. In general this might even be ill-typed since the context
of $t$ is not necessarily of the form $\oc\Gamma$. The solution
is indeed to assign the polarity $\varepsilon=+$ to the modality
$\oc$, enforcing a call-by-value reduction:
\begin{alignat*}{2}
\text{case }t=\mu\varp{\alpha}.c': & \quad & \varp{\cut{\mu\varp{\alpha}.c'}{\mt\varp x.c}} & \argred R{}c'\subst{\mt\varp x.c}{\alpha}\\
\text{case }t=V\neq\mu\varp{\alpha}.c': &  & \varp{\cut V{\mt\varp x.c}} & \argred R{}c\subst V{\alpha}
\end{alignat*}
This has the crucial consequence that duplication and erasure (and
expansion \eqref{eq:oc-expansion}) can only happen for $V$ of the
form $V=x$ or $V=\mu\oc\alpha.c$. Dually, assigning the polarity
$\varepsilon=-$ to the modality $\mon$ enforces a call-by-name reduction
meaning that a covalue $S$ of type $\mon A$ is necessarily some
$S=\mt\diam x.c$ or some $S=\alpha$.
\item \textbf{Focalisation.} The introduction rule for $\oc$ on the right
(respectively $\mon$ on the left) does not admit a restriction to
values (resp. stacks). Indeed, it is not possible to derive this rule
from the same rule restricted to a value (resp. a stack) using cuts,
unlike for the positive connectives $\oplus$ and $\otimes$ on the
right, like $\ileft t$ derives from $\ileft V$ (resp. for the negative
connectives $\mixarrow$ and $\with$ on the left). These rules have
a suspending effect. This one way we can interpret our observation
on oblique morphisms \eqref{eq:oblique-comonad-shift} from the introduction,
forcing extra shifts in the interpretation into the models. However,
weak focalisation properties are available for $\oc$ on the left
(dereliction) and $\mon$ on the right (return), forced by the weak
invertibility on the opposite side, as follows:
\begin{align*}
\varp{\cut V{\oc\termn e}} & \argRed E{}\varp{\cut{\mu\prom{\alpha}.\varp{\cut V{\prom{\alpha}}}}{\oc\termn e}}\argReds R{}\varn{\cut{\mu\varn{\alpha}.\varp{\cut V{\oc\alpha}}}{\termn e}}\\
\varn{\cut{\diam\termp t}S} & \argRed E{}\varn{\cut{\diam\termp t}{\mt\diam x.\varn{\cut{\diam x}S}}}\argReds R{}\varp{\cut{\termp t}{\mt\varp x.\varn{\cut{\diam x}S}}}
\end{align*}
The weak focalisations implementing the derived equivalences above
for dereliction and return are given in fig.\ \ref{fig:ILLpdiam-derived}:
\begin{align*}
\varp{\cut V{\oc\termn e}} & \argRed R{}\varn{\cut{\mu\varn{\alpha}.\varp{\cut V{\oc\alpha}}}{\termn e}}\\
\varn{\cut{\diam\termp t}S} & \argRed R{}\varp{\cut{\termp t}{\mt\varp x.\varn{\cut{\diam x}S}}}
\end{align*}
These focalisation properties are weak in the sense that the following
unrestricted focalisation-like rules are unsound in general:
\begin{align*}
\varp{\cut{\termp t}{\oc\termn e}} & \argRed{(unsound)}{}\varn{\cut{\mu\varn{\alpha}.\varp{\cut{\termp t}{\oc\alpha}}}{\termn e}}\\
\varn{\cut{\diam\termp t}{\termn e}} & \argRed{(unsound)}{}\varp{\cut{\termp t}{\mt\varp x.\varn{\cut{\diam x}{\termn e}}}}\;.
\end{align*}
\end{itemize}
The calculus $\ILLpdiam$ amounts to $\IMALLp$ augmented with a resource
modality $\oc$ and an effect modality $\mon$ following the design
constraints of polarisation and focalisation described above, guided
again from the idea of valid reductions and expansions playing well
together.

The calculus $\ILLpdiam$ and its specific syntax for $\oc$ and $\mon$
are justified by a sound interpretation into linear call-by-push-value
models with resource and effect modalities, and indeed serve to prove
(a proof-relevant version of) the completeness of Andreoli's focusing
proof search adapted to $\ILL^{\mon}$ (Appendix \ref{sec:Focused-normal-forms}).

We will now seek to recover the completeness property missing from
\citealt{CFM2015} by moving to relative (co)monads.

\subsection{Linear call-by-push-value models with relative resource and effect
modalities}

By moving from $\IMALLp$ to $\ILLpdiam$ (and even merely the fragment
$\ILLp$ without $\mon$), we lose completeness: it does not seem
possible to recover a resource modality on the category of thunkables
nor an effect modality on the category of linears. From the syntax,
one can build in fact a syntactic linear CBPV model with a resource
comonad relative to \emph{$\shpos$} and a monad relative to $\shneg$,
strong with respect to the comonad, as defined next. Monads (respectively
comonads) relative to $\shneg$ (resp. $\shpos$) were previously
used for modelling effect (resp. resource) modalities in \citet{Jiang2025}
(resp. \citet{Mellies2025BLC}).
\begin{defn}[\citealp{Altenkirch2015}]
A relative monad $T$ on a functor $L:\Acategory\to\Bcategory$ is
given by
\begin{itemize}
\item a mapping on objects $T:\obj{\Acategory}\rightarrow\obj{\Bcategory}$,
\item for any $A\in\Acategory$, a map $\eta_{A}\in\Bcategory(LA,TA)$,
\item for any $A,B\in\Acategory$, and any $f\in\Bcategory(LA,TB)$, a map
$\extension f\in\Bcategory(TA,TB)$,
\end{itemize}
such that two unital laws and one associativity law are satisfied:
\begin{align*}
f & =\extension f\circ\eta & \extension{\eta_{A}} & =\id_{TA} & \extension{(\extension f\circ g)} & =\extension f\circ\extension g
\end{align*}

\end{defn}

By duality:
\begin{defn}
A relative comonad $S$ on a functor $R:\Bcategory\rightarrow\Acategory$
is given by:
\begin{itemize}
\item a mapping on objects $S:\obj{\Bcategory}\rightarrow\obj{\Acategory}$,
\item for any $B\in\Bcategory$, a map $\varepsilon_{A}\in\Acategory(SA,RA)$,
\item for any $A,B\in\Bcategory$, and any $f\in\Acategory(SA,RB)$, a map
$\coextension f\in\Acategory(SA,SB)$,
\end{itemize}
such that two unital laws and one associativity law are satisfied:
\begin{align*}
f & =\varepsilon\circ\coextension f & \coextension{\varepsilon_{B}} & =\id_{SB} & \coextension{(f\circ\coextension g)} & =\coextension f\circ\coextension g
\end{align*}

\end{defn}

We will in fact use the characterisation of relative monads and comonads
as relative adjunctions and coadjunctions from \citet{Altenkirch2015,Arkor2024}.
\begin{defn}[Relative adjunctions and coadjunctions on categories]
Let $\Acategory$ and $\Bcategory$ be two categories and $L:\Acategory\to\Bcategory$
be a functor. An \textbf{$L$-relative adjunction} $F\adjunction LG$
consists of a category $\Ccategory$, two functors $F:\Acategory\to\Ccategory$
and $G:\Ccategory\to\Bcategory$, and a family of bijections: 
\[
\Phi_{A,C}:\Ccategory(FA,C)\cong\Bcategory(LA,GC).
\]
natural in~$A$ and~$C$. Diagrammatically:\vspace*{-2ex}
\begin{equation}\label{equation/adjunction}    \begin{tikzcd}[row sep = 0.5em, column sep=0.5em]
      & \Ccategory \arrow[rdd, bend left = 20, "G"]& \\
      & \dashv &  \\
     \Acategory \arrow[rr, "L"] \arrow[ruu, bend left = 20, "F"] & & \Bcategory 
    \end{tikzcd} \end{equation}We note $\eta:L\to G\circ F$ the unit of the adjunction.

Dually, for a functor $R:\Bcategory\rightarrow\Acategory$, an \textbf{$R$-relative
coadjunction} $F\coadjunction RG$ consists of a category $\Ccategory$,
two functors $F:\Ccategory\to\Acategory$ and $G:\Bcategory\to\Ccategory$,
and a family of bijections: 
\[
\Psi_{C,B}:\Acategory(FC,RB)\cong\Ccategory(C,GB)
\]
natural in~$B$ and~$C$. Diagrammatically:\vspace*{-2ex}
\begin{equation}\label{equation/coadjunction}     \begin{tikzcd}[row sep = 0.5em, column sep=0.5em]
       & \Ccategory \arrow[ldd, bend right = 20, "F"{swap}]& \\
       & \dashv &  \\
      \Acategory & & \Bcategory  \arrow[ll, "R"{swap}] \arrow[luu, bend right = 20, "G"{swap}] 
    \end{tikzcd}\end{equation}We note $\varepsilon:F\circ G\to R$ the counit of
the coadjunction.
\end{defn}

We can now revisit the interpretation from \prettyref{subsec:Linear-CBPV-models}
in light of the notions of relative adjunction and coadjunction starting
from the following simple fact about relative (co)adjunctions:
\begin{prop}
\label{prop:In--one}In \eqref{eq:modal-interp} one has a $\shneg$-relative
coadjunction $\Lin\coadjunction{\shneg}\Mult\circ\shneg$, and dually,
adjoint resolutions $F\dashv G$ of a comonad $T$ on $\Scategory$
give rise to $\shpos$-relative adjunctions $F\circ\shpos\adjunction{\shpos}G$.
\end{prop}

We observe that our interpretation of $\system{\ILL}^{\mon}$ into
linear CBPV models \eqref{eq:modal-interp} makes only use of this
relative structure. Indeed, assume we are given on top of a linear
CBPV adjunction \eqref{eq:lcbpv-adj}
\begin{itemize}
\item a $\shneg$-relative adjunction 
\[
\Lin\coadjunction{\shneg}\Mult:\Scategory\rightarrow\Mcategory
\]
where $\Lin:\Mcategory\rightarrow\Vcategory$ is strong monoidal,
defining a relative symmetric monoidal comonad $\oc=\Lin\circ\Mult$,
and
\item a $\shpos$-relative adjunction 
\[
F\adjunction{\shpos}G:\Vcategory\rightarrow\Ccategory
\]
defining a relative monad $\mon=G\circ F$, together with a strength
with respect to $\Lin$, in the form of a natural family of morphisms
\begin{equation}
\str_{X,P}:\Lin X\otimes\shpos\mon P\rightarrow\shpos\mon(\Lin X\otimes P)\label{eq:rel-rel-strength}
\end{equation}
making the four usual diagrams commute. 
\end{itemize}
Diagrammatically:\vspace*{-2ex}
\begin{equation}\label{eq:lcbpv-resource-effect-relative}  \begin{tikzcd}
    \Mcategory \arrow[rr, bend left, "\Lin"] & \bot & \Vcategory \arrow[rr, "\shneg", bend left] \arrow[rrrr, "F", bend left=40] & \bot & \Scategory \arrow[ll, "\shpos", bend left] \arrow[llll, "\Mult", bend left=40]  & \bot & \Ccategory \arrow[ll, "G", bend left]
  \end{tikzcd}\end{equation}
\begin{defn}
We call the data \eqref{eq:lcbpv-resource-effect-relative} given
by an $\IMALLp$ (linear CBPV) model from \citealp{CFM2015} together
with a $\shneg$-relative symmetric monoidal adjunction and a $\shpos$-relative
adjunction strong in the sense of \prettyref{eq:rel-rel-strength}
an \textbf{$\ILLpdiam$ model}.
\end{defn}

Then one has:
\begin{prop}
\label{prop:The-syntactic-linear}The syntactic linear CBPV model
$\Pcategory_{t}\rightleftarrows\Ncategory_{l}$ of thunkable and linear
terms of $\ILLpdiam$ has a relative comonad $\oc:\Ncategory_{l}\rightarrow\Pcategory_{t}$
on $\Shpos$, and a relative monad $\mon:\Pcategory_{t}\rightarrow\Ncategory_{l}$
on $\Shneg$, given by the corresponding rules of $\oc$ and $\mon$.
\end{prop}

\noindent The idea is that thunkable terms $\oc N\vdash\Shpos M$
in the syntactic model are in correspondence to generic terms of type
$\oc N\vdash M$, in such a way that the co-unit matches the dereliction
restricted to covalues,
\[
\mid\oc\alpha\colon\oc N\vdash\alpha:N
\]
and the coextension matches promotion, which indeed builds a thunkable
morphism:
\[
\unrulec{c\colon(\oc N\vdash\alpha\colon M)}{\oc N\vdash\mu\prom{\alpha}.c\colon\oc M\mid}{}
\]
For more details consult Appendix~\ref{sec:Relative-monads-and}.
It is another sign and testimony of the fundamental nature of linear-non-linear
relative coadjunctions introduced in the polarised setting of tensorial
logic and functorial game semantics that they also provide a way
to resolve a completeness issue appearing in the syntax and semantics
of linear CBPV with resource modalities \citep{CFM2015}.

We now establish that this new definition of model did not invalidate
soundness.
\begin{prop}
\label{prop:The-deductive-system-2}The deductive system $\duploid{\shneg}{\shpos}$
associated to a linear CBPV adjunction with resource and effect modalities
in the relative sense of \eqref{eq:lcbpv-resource-effect-relative}
satisfies the rules of $\system{ILL}^{\mon}$.
\end{prop}

Indeed, rephrasing in terms of the oblique morphism distributor, we
have a natural family of bijections:
\begin{equation}
\Oblique(\oc\varn A,\varn B)\cong\Mcategory(\Mult\varn A,\Mult\varn B)\label{eq:comon-oblique}
\end{equation}
mapping into $\Oblique(\oc\varn A,\varn{(\oc B)})$, thus corresponding
to the rule $(\vdash\oc)$ (with a single antecedent, omitting the
lax structure at the moment). Symmetrically for $\mon$, we have a
natural family of bijections:
\begin{equation}
\Oblique(\varp A,\mon\varp B)\cong\Ccategory(F\varp A,F\varp B)\label{eq:mon-oblique}
\end{equation}
mapping into $\Oblique(\varp{(\mon A)},\mon\varp B)$, thus corresponding
to the rule $(\mon\vdash)$ (again with a single antecedent, omitting
the strength).

In order to be more precise, we will now make use of the notions of
(co)algebras for relative (co)monads. Just as for their non-relative
counterparts, categories of algebras are terminal in a category of
adjoint resolutions for a given relative (co)monad, in such a way
that for any relative (co)adjunction there exists a comparison functor
factoring the right (resp. left) adjoint \citep[Theorem 2.12]{Altenkirch2015}.
The sets \eqref{eq:comon-oblique} thus map via the comparison functor
into the set of $\oc$-coalgebra morphisms $\Vcategory^{\oc}(F^{\oc}\varn A,F^{\oc}\varn B)$,
in which values of type $\oc A$ are interpreted, where objects $F^{\oc}\varn X$
are identified with the types $\oc X$ of $\ILLpdiam$. Symmetrically,
the sets \eqref{eq:mon-oblique} map via the comparison functor into
the set of $\mon$-algebra morphisms $\Scategory^{\mon}(G^{\mon}\varp A,G^{\mon}\varp B)$
in which covalues of type $\mon A$ are interpreted, where objects
$G^{\mon}\varp X$ are identified with the types $\mon X$ of $\ILLpdiam$.
The rest of the interpretation of $\ILLpdiam$ into linear CBPV models
with resource and effect modalities remains unchanged. Thus:
\begin{prop}
The interpretation of $\ILL^{\mon}$ from \prettyref{prop:The-deductive-system-2}
extends into a sound interpretation of $\ILLpdiam$.
\end{prop}

\subsection{Examples of $\protect\ILLpdiam$ models}

By \prettyref{prop:In--one}, any linear CBPV adjunction with (non-relative)
resource and effect modalities~\eqref{eq:lcbpv-resource-effect}
provides an $\ILLpdiam$ model. Two series of works, by \citet{Jiang2025}
and by \citet{Mellies2025BLC,Mellies2026TLLA}, motivate a relaxation
of the notions of effect and resource modalities to monads relative
to $\shneg$ and to comonads relative to $\shpos$.
\begin{example}
\citet{Jiang2025} investigate an interface for user-defined programming
effects in CBPV using Church encodings and coinductive types. Since
universal quantification and coinductive types are typically negative,
this leads to several examples of monads whose type of monadic computations
is negative, including various examples of implementations of exception,
continuation, state and free monads. The situations they study amount
to particular $\ILLpdiam$ models \eqref{eq:lcbpv-resource-effect-relative}
 with $\Vcategory=\Mcategory$ and $\Ccategory=\Alg{\mon}$ the category
of $\mon$-algebras, i.e. non-linear CBPV models with an effect modality
as follows:\begin{equation}\label{eq:lcbpv-effect-relative}  \begin{tikzcd}[column sep = 1.5em]
    \Mcategory \arrow[rr, "\shneg", bend left] \arrow[rrrr, "F^{\mon}", bend left=40] & \bot & \Scategory \arrow[ll, "\shpos", bend left]  & \bot & \Alg{\mon} \arrow[ll, "G^{\mon}", bend left]
  \end{tikzcd}\end{equation}

This is a very natural situation owing to the following observation
extending the classic results by \citet{Beck} on distributive laws
between monads:
\end{example}

\begin{defn}
Let $(S,\eta,-^{\dagger})$ be a monad on a category $\Acategory$
and $L:\Acategory\rightarrow\Bcategory$ a functor. A structure of
$L$-relative monad $(L\circ S,\tilde{\eta},-^{\ddagger})$ on the
functor $L\circ S:\Acategory\rightarrow\Bcategory$ is an \emph{extension
of }$(S,\eta,-^{\dagger})$ if and only if for all $A\in\Acategory$
one has:
\begin{align*}
\tilde{\eta}_{A} & =L\eta_{A}\\
(Lf)^{\ddagger} & =L(f^{\dagger})\,.
\end{align*}
\end{defn}

\begin{thm}
\label{thm:beck-relative}Let $T,S$ be two monads on a category $\Acategory$
and $\adjoint L{\Acategory}{\Bcategory}R$ a resolution of the monad
$T=R\circ L$. There is an one-to-one correspondence between:
\begin{itemize}
\item distributive laws of monads $ST\rightarrow TS$,
\item extensions of the monad $S$ to a $L$-relative monad structure on
$L\circ S$, diagrammatically:   
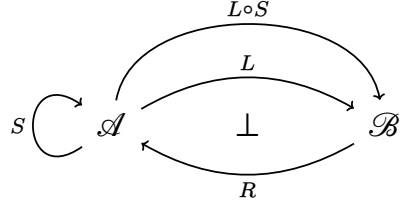
\[\begin{tikzcd}[column sep = 2.5em]
    \Acategory \arrow[loop, "S", in=145, out=215, distance=2.5em] \arrow[rr, "L", bend left] \arrow[rr, "L\circ S", bend left=80] & \bot & \Bcategory \arrow[ll, "R", bend left]
  \end{tikzcd}\]
\end{itemize}
\end{thm}

This theorem generalises Beck's characterisation of distributive laws
$ST\rightarrow TS$ as extensions of $S$ to the Kleisli category
of $T$. In particular, given $T,S$ two monads on a category $\Acategory$
and a distributive law of monads $ST\rightarrow TS$, this result
can be applied to the category $\Acategory^{T}$ of $T$-algebras
to provide us with a relative adjoint situation as follows   
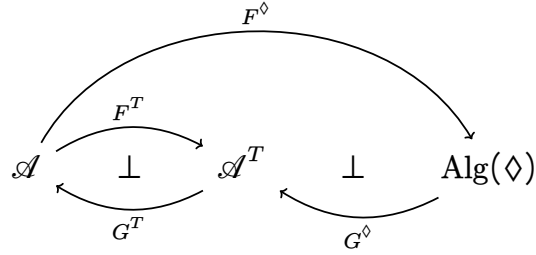
\[\begin{tikzcd}[column sep = 1.5em]
    \Acategory \arrow[rr, "F^T", bend left] \arrow[rrrr, "F^{\mon}", bend left=60] & \bot & \Acategory^T \arrow[ll, "G^T", bend left]  & \bot & \Alg{\mon} \arrow[ll, "G^{\mon}", bend left]
  \end{tikzcd}\]where $\mon=F^{T}\circ S$ is the $F^{T}$-relative monad obtained
by \prettyref{thm:beck-relative}.

\begin{example}
\citet{Mellies2025BLC,Mellies2026TLLA} proposed to extend tensorial
logic and dialogue categories to linear logic exponentials, using
relative rather than non-relative symmetric monoidal comonads as a
way of circumventing obstacles encountered within the programme of
\emph{functorial game semantics} \citep{Mellies12} that are similar
to the difficulties we have described previously regarding notions
of polarities and focusing for exponentials. This setting corresponds
to a variant of the situation \eqref{eq:adj-dialogue-resource} involving
a dialogue category, now with relative comonad and monad:\begin{equation}  \nonumber
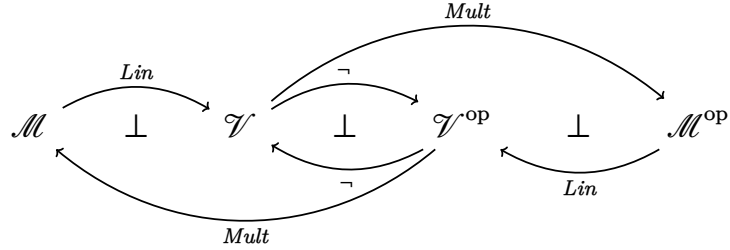
\begin{tikzcd}[column sep = 1.5em]
    \Mcategory \arrow[rr, bend left, "\Lin"] & \bot & \Vcategory \arrow[rr, "\neg", bend left] \arrow[rrrr, "\Mult", bend left=40] & \bot & \Vcategory^\op \arrow[ll, "\neg", bend left] \arrow[llll, "\Mult", bend left=40]  & \bot & \Mcategory^\op \arrow[ll, "\Lin", bend left]
  \end{tikzcd}\end{equation}
\end{example}

\begin{example}
Another approach to a notion of strength for a relative monad is to
assume a given monoidal structure on the target category $\Scategory$,
which then allows us to define the extension of the relative monad
under a context $\Gamma\in\Scategory$,
\begin{equation}
\Scategory(\Gamma\otimes\shneg P,\mon Q)\rightarrow\Scategory(\Gamma\otimes\mon P,\mon Q)\label{eq:strong-rel-liell}
\end{equation}
This is an assumption made in \citet{LiellCock2026}. By way of comparison,
we assume that $\shneg$ has a right adjoint $\shpos$, which is essential
in making our definition of strength \eqref{eq:rel-rel-strength}
work; one advantage in return is that it does not depend on having
a monoidal structure on $\Scategory$. Such a (strong) monoidal structure
on negative types is indeed unavailable in certain situations that
interest us, such as categories of algebras in general.

\citet{LiellCock2026} introduce the notion of \emph{RMM-LNL model},
given by a linear-non-linear (LNL) model $\adjoint F{\Mcategory}{\Scategory}G$
together with a strong $F$-relative monad $\mon$ in the above sense
\eqref{eq:strong-rel-liell}, and some strong monoidal functor from
some monoidal category into $\Scategory$. Further assuming $\Mcategory$
is distributive, this is again an instance of a non-linear CBPV model
with an effect modality \eqref{eq:lcbpv-effect-relative}. Indeed,
recall that an LNL model \citep{BentonLinearNonLinear} consists in
a symmetric monoidal adjunction $\adjoint F{\Mcategory}{\Scategory}G$
between a cartesian closed category $\Mcategory$ and an symmetric
monoidal closed category $\Scategory$. It is in particular a (non-linear)
CBPV model without sums where $\shneg=F$ and $\shpos=G$; a very
particular class in which the category of stacks $\Scategory$ is
monoidal. In an RMM-LNL model, the natural family of morphisms \eqref{eq:rel-rel-strength}
constituting the strength of the $\shneg$-relative monad,  instantiated
as $\lstr_{\Gamma,P}\in\Mcategory(\Gamma\times G\mon P,G\mon(\Gamma\times P))$,
derives from
\begin{gather*}
\eta_{\Gamma\times P}\in\Scategory(F(\Gamma\times P),\mon(\Gamma\times P))\cong\Scategory(F\Gamma\otimes FP,\mon(\Gamma\times P))\rightarrow\Scategory(F\Gamma\otimes\mon P,\mon(\Gamma\times P))\rightarrow\\
\Scategory(F\Gamma\otimes FG\mon P,\mon(\Gamma\times P))\cong\Scategory(F(\Gamma\times G\mon P),\mon(\Gamma\times P))\cong\Mcategory(\Gamma\times G\mon P,G\mon(\Gamma\times P))\,.
\end{gather*}
\end{example}

\begin{example}
By moving from (co)adjunctions to relative (co)adjunctions for modelling
resource and effect modalities, we indeed recovered the missing completeness
property of the interpretation, giving rise to the syntactic $\ILLpdiam$
model. This previous syntactic example extends to a more general phenomenon
giving rise to a final example of $\ILLpdiam$ models, obtained by
completing another $\ILLpdiam$ model with its thunkable and linear
morphisms. Remember that, in line with the relationships between direct
and indirect models of effectful computation \citep{fuhrmanndirectmodels,Sel01Control},
any linear CBPV model $\adjoint{\shneg}{\Vcategory}{\Scategory}{\shpos}$
\eqref{eq:lcbpv-adj} gives rise to a linear CBPV model
\begin{equation}
\xymatrix{\Pcategory_{t}\xyport{\xyarrow[rr][][][1.5pc]}{\Shneg}{} & \bot & \xyport{\xyarrow[ll][][][1.5pc]}{\Shpos}{}\Ncategory_{l}}
\label{eq:lcbpv-adj-1}
\end{equation}
obtained starting from the non-associative category $\duploid{\shneg}{\shpos}$
and restricting to positive thunkable and negative linear morphisms.
Remember also that this construction is a reflection for the category
of adjunctions $\shpos\dashv\shneg$, and more precisely a completion
(of morphisms of $\Vcategory$ into thunkable morphisms, and of morphisms
of $\Scategory$ into linear morphisms of the adjunction) provided
the starting adjunction has its unit mono and its counit epi \citep{munchduploids}.
Then, resource and effect modalities on $\Vcategory\rightleftarrows\Scategory$,
whether non-relative \eqref{eq:lcbpv-resource-effect} or relative
\eqref{eq:lcbpv-resource-effect-relative}, give rise to relative
resource and effect modalities on $\Pcategory_{t}\rightleftarrows\Ncategory_{l}$.
This is a corollary to the results established in the next section
(\prettyref{cor:Any-ILLp-model}). In the non-relative case, however,
the same obstacle \eqref{eq:oblique-comonad-shift} previously mentioned
prevents us from obtaining (directly and in general) a non-relative
comonad on $\Pcategory_{t}$ and a non-relative monad on $\Ncategory_{l}$.
\end{example}

These results motivate in the rest of the paper an investigation from
the non-associative perspective into broader notions of model based
on resource and effect modalities relative to the linear CBPV adjunction.

\section{Adjunctions over duploids and relative adjunctions}

\subsection{Prelude: distributive closed duploids and distributive dialogue duploids}

We have recalled previously that every adjunction $\shneg\dashv\shpos$
induces a non-associative category $\duploid{\shneg}{\shpos}$; it
is in fact a duploid as defined next.
\begin{defn}
An object $A$ of a non-associative category is \textbf{positive}
if every morphism $f:A\rightarrow B$ is linear. It is \textbf{negative}
if every morphism $f:B\rightarrow A$ is thunkable. A \textbf{duploid}
is a non-associative category in which every object is positive or
negative (or both), and in which the identity functor has both a positive
left adjoint (which we write $\Shpos$) and a negative right adjoint
(written $\Shneg$). Diagrammatically:  
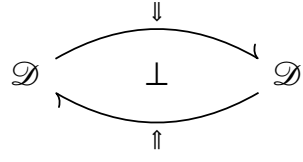
\[
    \begin{tikzcd}
 \Dduploid \arrow[rr,harpoon, bend left, "\Shpos"] 
 & \bot & \Dduploid \arrow[ll,harpoon, "\Shneg", bend left]
    \end{tikzcd}
  \]
\end{defn}

\begin{defn}
An adjunction between two non-associative categories $\graphadjoint F{\Dduploid}{\Eduploid}G$
(\prettyref{def:An-adjunction}) is said to be \textbf{positive} when
the image of $F$ on objects is positive. Dually, it is \textbf{negative}
when the image of $G$ on objects is negative.
\end{defn}

The notions of \emph{symmetric monoidal closed duploid} and of \emph{dialogue
duploid}, providing non-associative categorical and axiomatic semantics
to $\system{IMLL}$ (intuitionistic multiplicative linear logic) and
$\system{MLL}$ (classical multiplicative linear logic) respectively,
are defined in \citealt{MMMM2026}. These non-associative categorical
and axiomatic semantics are designed to ensure that their equational
theory and normal forms coincide with the syntax of the multiplicative
fragment of the corresponding $L$-calculi.

We now define here the distributive structure in the continuation
of \prettyref{subsec:Adjunctions-in-non-associative}.
\begin{defn}
Let $\Dduploid$ be a non-associative category. \textbf{Finite sums}
on $\Dduploid$ is the structure given by: 1)~a positive left adjoint
to the diagonal functor $\Delta:\Dduploid\rightarrow\Dduploid\times\Dduploid$,
which we write $\oplus$, 2)~a positive left adjoint to the functor
$\Dduploid\rightarrow\ast$ into the trivial category, in other words
a positive initial object, which we write $\falso$. \textbf{Finite
products} (notation $\with,\top$) is the structure given by finite
sums on $\Dduploid^{\op}$.
\end{defn}

\begin{defn}
A \textbf{distributive duploid} is a symmetric monoidal duploid $\Dduploid$
that has finite sums and such that for all $A,B,C\in\Dduploid$ the
canonical thunkable map 
\[
(A\otimes B)\oplus(A\otimes C)\rightarrow A\otimes(B\oplus C)
\]
has a thunkable inverse.
\end{defn}

\subsection{A correspondence between relative adjunctions and adjunctions in
non-associative direct semantics}

Our goal is to extend this non-associative categorical semantics by
introducing a notion of linear-non-linear adjunction and strong monad
for duploids. One main observation of the present paper is that the
appropriate generalisation of linear-non-linear adjunctions to a polarised
setting is in correspondence with the notion of relative coadjunction
with respect to the polarity shift~$\nDownarrow$ appearing as right
adjoint in the adjunction associated to the duploid   
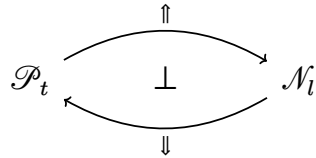
\[
    \begin{tikzcd}
 \Pcategoryt \arrow[rr, bend left, "\nUparrow"] 
 & \bot & \Ncategoryl \arrow[ll, "\nDownarrow", bend left]
    \end{tikzcd}
  \]between the category of positive thunkable maps and the category of
negative linear maps, whereas the strong monads on duploids are in
correspondence with the notion of relative adjunction with respect
to the adjoint polarity shift~$\Shneg$. 

We first explain how, given an adjunction $L\dashv R$, an $R$-relative
coadjunction $F\coadjunction RG$ induces a positive adjunction on
the duploid~$\duploid LR$ (\prettyref{prop:constr-pos-adj}). By
duality, this implies that every $L$-relative adjunction $F\adjunction LG$
induces a negative adjunction on the duploid~$\duploid LR$. Conversely,
we show that every positive adjunction $\graphadjoint F{\Ccategory}{\Dduploid}G$
between an associative category~$\Ccategory$ and a duploid $\Dduploid$
induces a $\nDownarrow$-relative coadjunction (\prettyref{prop:constr-rel-coadj}).
(And thus dually for the construction of a relative adjunction from
negative adjunction.)

Let us recall first a few fundamental and useful properties of adjunctions
over non-associative categories, also developed in more detail in
\citealt{MMMM2026}.
\begin{prop}[R.A.P.L.]
Right adjoints preserve linear maps and their composition.
\end{prop}

\begin{prop}
Left adjoints preserve thunkable maps and their composition.
\end{prop}

\begin{prop}
The left adjoint $F:\Ccategory\graphto\Dduploid$ of a positive adjunction
where $\Ccategory$ is a category is a functor, and moreover restricts
into a functor 
\[
F\;:\;\Ccategory\longrightarrow\Pcategoryt
\]
 from the category $\Ccategory$ to the category $\Pcategoryt$ of
positive thunkable maps in $\Dduploid$. Dually, the right adjoint
$G:\Ccategory\graphto\Dduploid$ of a negative adjunction where $\Ccategory$
is a category is a functor, and moreover restricts into a functor
\[
G\;:\;\Ccategory\longrightarrow\Ncategoryl
\]
 from the category $\Ccategory$ to the category $\Ncategoryl$ of
negative linear maps in $\Dduploid$. 
\end{prop}

\paragraph{Constructing a positive adjunction from an $R$-relative coadjunction}

To demonstrate the construction of an adjunction on duploids from
an $R$-relative coadjunction, we start by showing that we can construct
a bare functor from the right adjoint of an $R$-relative coadjunction.
\begin{prop}
Let $\adjoint L{\Acategory}{\Bcategory}R$ be an adjunction over associative
categories and $F\coadjunction RG:\Bcategory\to\Ccategory$ be a $R$-relative
coadjunction. Then, there exists a bare functor $G^{-}:\duploid LR\to\Ccategory$.
\end{prop}

\begin{proof}
On objects, we define $G^{-}$ as $G$ on objects of $\Bcategory$
and as $G\circ L$ on objects of $\Acategory$. Let $f\in\duploid LR(X,Y)$,
that is, a morphism in $\Acategory(X^{+},RY^{-})$. The definition
of $G^{-}f$ depends on whether $X$ is in $\Acategory$ or $\Bcategory$.
If $X$ is in $\Acategory$, we define $G^{-}f$ as: 
\[
G^{-}f\df G\Phi^{-1}_{X,Y^{-}}(f)\in\Ccategory(GLX,GY^{-})
\]
where $\Phi_{A,B}:\Bcategory(LA,B)\stackrel{\cong}{\longrightarrow}\Acategory(A,RB)$.
If $X$ is in $\Bcategory$, we define $G^{-}$ as: 
\[
G^{-}f\df\Psi_{GX,Y^{-}}(f\circ\epsilon_{X})\in\Ccategory(GX,GY^{-})
\]
where $\Psi_{C,B}:\Acategory(FC,RB)\stackrel{\cong}{\longrightarrow}\Ccategory(C,GB)$.
\end{proof}

\begin{prop}
\label{prop:constr-pos-adj}Every $R$-relative coadjunction $F\coadjunction RG$
of the form~\eqref{equation/coadjunction} induces a positive adjunction
on $\duploid LR$:\[\begin{tikzcd}[column sep = 1.5em]
  {\Ccategory}  \arrow[rr,"F", bend left = 40]  & \bot & \arrow[ll,"G^-",harpoon, bend left = 40] \makebox[1ex][l]{$\duploid{L}{R}$}
\end{tikzcd}\]
\end{prop}

\begin{proof}
We construct the bijection of the adjunction $\duploid LR(FC,X)\cong\Ccategory(C,G^{-}X)$
by case analysis on $X$. If $X$ is in $\Bcategory$, then: 
\[
\duploid LR(FC,X)=\Acategory(FC,RX)\cong\Ccategory(C,GX)\cong\Ccategory(C,G^{-}X)
\]
If $X$ is in $\Acategory$, then: 
\[
\duploid LR(FC,X)=\Acategory(FC,RLX)\cong\Ccategory(C,GLX)\cong\Ccategory(C,G^{-}X)\qedhere
\]
\end{proof}

\paragraph{Constructing a $\protect\nDownarrow$-relative coadjunction from
a positive adjunction}
\begin{prop}
\label{prop:constr-rel-coadj}Every positive adjunction $\graphadjoint F{\Ccategory}{\Dduploid}G$
where $\Ccategory$ is a category induces a $\nDownarrow$-relative
coadjunction of the form: \[\begin{tikzcd}[row sep = 0.5em, column sep = 0.5em]
       & \Ccategory \arrow[ldd, bend right = 20, "F"{swap}]& \\
       & \dashv &  \\
      \Pcategoryt & & \Ncategoryl  \arrow[ll, "\nDownarrow"] \arrow[luu,bend right = 20, "G"{swap}] 
    \end{tikzcd}\]where we write $G$ for the restriction of the original right adjoint
$G:\Dduploid\graphto\Ccategory$.
\end{prop}

\begin{proof}
Let $C$ be an object of $\Ccategory$ and $N$ an object of $\Ncategoryl$.
We construct the bijection of the $\nDownarrow$-relative adjunction
by using the bijection of the adjunction on duploids:
\[
\Pcategoryt(FC,\nDownarrow N)\cong\Dduploid(FC,N)\cong\Ccategory(C,GN)\qedhere
\]
\end{proof}

\subsection{Relative LNL coadjunctions and direct LNL adjunctions}

We now consider the case where the modality is the exponential $A\mapsto\,\com\,A$.
The linear-non-linear adjunction becomes, when polarized, either a
$R$-relative linear-non-linear coadjunction in an indirect setting
or a positive adjunction with a cartesian category in a direct setting.
We also show how we can construct one from the other.
\begin{defn}
Let $\Acategory$ and $\Bcategory$ be two categories such that $\Acategory$
is symmetric monoidal category and $R:\Bcategory\to\Acategory$ be
a functor. A \textbf{$R$-relative linear-non-linear coadjunction}
is a $R$-relative coadjunction:\begin{equation}\label{equation/relative-LNL}
    \begin{tikzcd}
      \Mcategory \arrow[rr, bend left, "\Lin"]& \bot & \Acategory  & & \Bcategory \arrow[ll, bend left, "R"] \arrow[llll, bend left = 40, "\Mult"]
    \end{tikzcd}\end{equation}such that $\Mcategory$ is cartesian and $\Lin$ is
a strong monoidal functor. 
\end{defn}

This relative coadjunction means that we have a family of bijections
: 
\[
\varphi_{\vec{X_{i}},A}:\Acategory(\Lin X_{1}\otimes\dots\otimes\Lin X_{n},RA)\cong\Mcategory(X_{1}\times\dots\times X_{n},\Mult A)
\]
natural in $X_{1}$, …, $X_{n}$ and $A$.

\begin{defn}
A \textbf{direct linear-non-linear adjunction} is a positive adjunction
  \[
    \begin{tikzcd}
      \Mcategory \arrow[rr, bend left, "\Lin"] & \bot & \Dduploid \arrow[ll, bend left, harpoon, "\Mult"]
    \end{tikzcd}
  \]between a symmetric monoidal duploid $\Dduploid$ and a cartesian
category $\Mcategory$ such that the functor from $\Mcategory$ to
$\Pcategoryt$ induced by $\Lin$ is strong monoidal. 
\end{defn}

\paragraph{Correspondence theorem}
\begin{prop}
Every $R$-relative linear-non-linear coadjunction $\Lin\coadjunction R\Mult$
of the form \eqref{equation/relative-LNL} induces a direct linear-non-linear
adjunction:\begin{equation}\label{equation/indirect-direct-lnl}
\begin{tikzcd}
  \makebox[1ex][r]{$\Mcategory$}  \arrow[rr,"\Lin",  bend left = 40]  & \bot & \arrow[ll,"\Mult^-",harpoon, bend left = 40] \makebox[1ex][l]{$\duploid{L}{R}$}
\end{tikzcd}\end{equation}
\end{prop}

\begin{prop}
Every direct linear-non-linear adjunction $\graphadjoint{\Lin}{\Mcategory}{\Dduploid}{\Mult}$
induces a $\nDownarrow$-relative linear-non-linear coadjunction of
the form:   \[
\begin{tikzcd}[row sep = 0.5em, column sep = 0.5em]
       & \Mcategory \arrow[ldd, bend right = 20, "\Lin"{swap}]& \\
       & \dashv &  \\
      \Pcategoryt & & \Ncategoryl  \arrow[ll, "\nDownarrow"{swap}] \arrow[luu,bend right = 20, "\Mult"{swap}] 
    \end{tikzcd}
  \]
\end{prop}

\subsection{Strength for a relative monad with respect to a comonad}

Finally, we consider the case where we have the resource modality
$A\mapsto\com A$ and an effect modality $A\mapsto\mon A$ strong
with respect to $\com$.
\begin{defn}
Suppose given an adjunction $\adjoint L{\Acategory}{\Bcategory}R$
where the category~$\Acategory$ is equipped with a symmetric monoidal
structure $(\otimes,\unite)$ and where the monad $R\circ L$ is strong.
Suppose also given an $L$-relative adjunction $F\dashv G$ as well
as an $R$-relative linear-non-linear coadjunction $\Lin\dashv\Mult$.
Diagramatically:\begin{equation}\label{equation/relative-strength-situation}
  \begin{tikzcd}
    \Mcategory \arrow[rr, bend left, "\Lin"] & \bot & \Acategory \arrow[rr, "L", bend left] \arrow[rrrr, "F", bend left=37] & \bot & \Bcategory \arrow[ll, "R", bend left] \arrow[llll, "\Mult", bend left=37]  & \bot & \Ccategory \arrow[ll, "G", bend left]
  \end{tikzcd}\end{equation}The relative monad $\mon=G\circ F$ is strong with
respect to the functor $\Lin$ if it is equipped with a natural family
of morphisms in $\Acategory$ 
\[
\str_{X,A}:\Lin X\otimes R\mon A\to R\mon(\Lin X\otimes A)
\]
making the four usual coherence diagrams commute.
\end{defn}

\begin{defn}
Let a direct linear-non-linear adjunction $\graphadjoint{\Lin}{\Mcategory}{\Dduploid}{\Mult}$
and a negative adjunction $\graphadjoint F{\Dduploid}{\Ccategory}G$.
Diagrammatically:\begin{equation}\label{equation/direct-strength-situation}    \begin{tikzcd}[column sep = 3em]
      \Mcategory \arrow[rr, bend left = 30, "\Lin"] & \bot & \Dduploid \arrow[ll,harpoon, bend left = 30, "\Mult"] \arrow[rr,harpoon, bend left = 30, "F"] & \bot & \Ccategory \arrow[ll, bend left = 30, "G"]
    \end{tikzcd}\end{equation}The bare endofunctor $\mon=G\circ F$ is strong with
respect to the bare functor $\Lin$ if it is equipped with a natural
family of morphisms in $\Dduploid$ 
\[
\str_{X,A}:\Lin X\otimes\mon A\to\mon(\Lin X\otimes A)
\]
linear wrt. $\mon A$ \citep[Definition F.2, p.42]{MMMM2026}, making
the four usual coherence diagrams commute.
\end{defn}

\paragraph{Correspondence theorem}
\begin{prop}
Every situation of the form \eqref{equation/relative-strength-situation}
such that the $L$-relative monad $G\circ F$ is strong with respect
to $\Lin$ induces a negative adjunction: \[
\begin{tikzcd}
  \makebox[1ex][r]{$\duploid{L}{R}$} \arrow[rr,"F^+",harpoon,  bend left = 40]  & \bot &  {\Ccategory} \arrow[ll,"G", bend left = 40] 
\end{tikzcd}
\]such that $G\circ F^{+}$ is strong with respect to $\Lin$, the left
adjoint of the induced direct linear-non-linear adjunction from \eqref{equation/indirect-direct-lnl}.
\end{prop}

\begin{prop}
Every situation of the form \eqref{equation/direct-strength-situation}
such that the bare endofunctor $G\circ F$ is strong with respect
to $\Lin$ induces a $\nUparrow$-relative adjunction:\[\begin{tikzcd}[row sep = 0.5em, column sep = 0.5em]
       & \Ccategory \arrow[rdd, bend left = 20, "G"]& \\
       & \dashv &  \\
      \Pcategoryt \arrow[ruu,bend left = 20, "F"] \arrow[rr, "\nUparrow"{swap}] & & \Ncategoryl  
    \end{tikzcd}\]such that the relative monad $G\circ F$ is strong with respect to
$\Lin$.
\end{prop}

\begin{cor}
\label{cor:Any-ILLp-model}Any $\ILLpdiam$ model over an adjunction
$\adjoint{\shneg}{\Vcategory}{\Scategory}{\shpos}$ induces an $\ILLpdiam$
model over its thunkable-linear completion $\adjoint{\Shneg}{\Pcategory_{t}}{\Ncategory_{l}}{\Shpos}$.
\end{cor}

\todo[inline]{Thm de composition?}

\section*{Acknowledgements}

This work has received funding from the European Research Council
under the European Union’s Horizon 2020 research and innovation programme
(Synergy Project Malinca, ERC Grant Agreement No 670624).

\printshorthands{}

\printbibliography[heading=bibintoc]

\newpage{}

\appendix
\section{Appendix}

\subsection{Componentwise bare and proper functors}
A componentwise bare functor
\begin{center}
\begin{tikzcd}
F \quad : \quad \Acategory,\Bcategory \arrow[rr] && \Ccategory
\end{tikzcd}
\end{center}
between non associative categories~$\Acategory$, $\Bcategory$ and~$\Ccategory$,
is a pair consisting of a function 
\begin{center}
\begin{tikzcd}
\objectp{F} \quad : \quad 
\objectp{\Acategory}\times\objectp{\Bcategory} \arrow[rr] && \objectp{\Ccategory}
\end{tikzcd}
\end{center}
which associates an object $F(A,B)$ of $\Ccategory$
to every pair~$(A,B)$ consisting of an object~$A$ of $\Acategory$
and of an object~$B$ of $\Bcategory$ ; together with a family of functions
\begin{center}
\begin{tikzcd}
\fstcomponent{F}{B}(A_1,A_2) \quad : \quad  \Acategory(A_1,A_2) 
\arrow[rr] && \Ccategory((A_1,B),(A_2,B))
\end{tikzcd}
\end{center}
\begin{center}
\begin{tikzcd}
\sndcomponent{F}{A}(B_1,B_2) \quad : \quad \Bcategory(B_1,B_2)
\arrow[rr] && \Ccategory((A,B_1),(A,B_2))
\end{tikzcd}
\end{center}
such that the series of equalities holds:
\begin{center}
$\fstcomponent{F}{B}(\id{A}) = \id_{\objectp{F}(A,B)}$
\quad\quad 
$\sndcomponent{F}{A}(\id{B}) = \id_{\objectp{F}(A,B)}$
\end{center}
A componentwise proper functor is a componentwise bare functor
such that the equality holds:
\begin{center}
\begin{tabular}{c}
$\fstcomponent{F}{B}(g\circ f) = \fstcomponent{F}{B}(g)\circ \fstcomponent{F}{B}(f)$ 
\\
\\
$\sndcomponent{F}{A}(g\circ f) = \sndcomponent{F}{A}(g)\circ \sndcomponent{F}{A}(f)$
\end{tabular}
\end{center}

\subsection{Tensor product}
The tensor product $\Acategory\boxtimes\Bcategory$
of two non associative categories~$\Acategory$ and~$\Bcategory$
is the non associative category
whose objects are pairs $(A,B)$ 

$$
\begin{tikzcd}
(A,g) \quad : \quad (A,B_1) \arrow[rr] && (A,B_2)
\end{tikzcd}
$$
$$
\begin{tikzcd}
(f,B) \quad : \quad (A_1,B) \arrow[rr] && (A_2,B)
\end{tikzcd}
$$
such that
$$
(A,g_2\circ g_1) = (A,g_2)\circ (A,g_1)
\quad\quad
(A,\id{B}) = \id{(A,B)}
$$
$$
(f_2\circ f_1,B) = (f_2,B)\circ (f_1,B)
\quad\quad
(\id{A},B) = \id{(A,B)}
$$
There is a componentwise proper functor
$$
\begin{tikzcd}
\Acategory,\Bcategory \arrow[rr] && \Acategory\boxtimes\Bcategory
\end{tikzcd}
$$
which induces a bijection between componentwise proper functors
$$
\begin{tikzcd}
\Acategory,\Bcategory \arrow[rr] && \Ccategory
\end{tikzcd}
$$
and proper functors
$$
\begin{tikzcd}
\Acategory\boxtimes\Bcategory \arrow[rr] && \Ccategory
\end{tikzcd}
$$
Moreover, it induces a bijection between componentwise bare functors 
$$
\begin{tikzcd}
\Acategory,\Bcategory \arrow[rr] && \Ccategory
\end{tikzcd}
$$
and bare functors
$$
\begin{tikzcd}
\Acategory\boxtimes\Bcategory \arrow[rr] && \Ccategory
\end{tikzcd}
$$

\section{\label{sec:IMALLp-has-sums}$\protect\IMALLp$ has sums}
\begin{defn}
A non-associative category $\Dduploid$ has (binary) sums if $\Delta$
has a positive left adjoint, in other words a bare functor $\oplus:\Dduploid\times\Dduploid\graphto\Dduploid$
such that $A\oplus B$ is positive and such that there exists a family
of isomorphisms:
\[
\Dduploid(A\oplus B,C)\cong\Dduploid(A,C)\times\Dduploid(B,C)
\]
which is natural in $(\Dduploid\times\Dduploid)^{\op}\boxtimes\Dduploid\graphto\Setcat$.
\end{defn}

Let us spell out the naturality condition more explicitly: for any
$f:A\rightarrow C$ and $g:B\rightarrow C$ there exists a pairing
$\pairing fg:A\oplus B\rightarrow C$ such that:
\begin{align}
\forall h:C\rightarrow D, &  & h\ccomp\pairing fg & =\pairing{h\ccomp f}{h\ccomp g}\label{eq:sum-nat-1}\\
\forall h_{1}:A'\rightarrow A,\forall h_{2}:B'\rightarrow B, &  & \;\pairing fg\ccomp(h_{1}\oplus h_{2}) & =\pairing{f\ccomp h_{1}}{g\ccomp h_{2}}\label{eq:sum-nat-2}
\end{align}

\begin{prop}
The syntactic duploid has sums given by the type $\oplus$ and its
constructors:
\begin{align*}
\oplus_{A_{x},B_{y},C,D,z} & :\Dduploid(A_{x},C)\times\Dduploid(B_{y},D)\rightarrow\Dduploid((A\oplus B)_{z},C\oplus D)\\
t\oplus_{x,y,z}u & \df\mu\alpha.\cut z{\mui x{\cut{\ileft(t)}{\alpha}}y{\cut{\iright(u)}{\alpha}}}\\
 & \text{where }\iith(s)\df\mu\alpha.\cut t{\mt x.\cut{\iith(x)}{\alpha}}\\
\varphi_{A,B,C} & :\begin{cases}
\Dduploid(A\oplus B,C_{\alpha})\rightarrow\Dduploid(A,C_{\alpha})\times\Dduploid(B,C_{\alpha})\\
e\mapsto(\mt x.\cut{\ileft(x)}e,\mt y.\cut{\iright(y)}e)
\end{cases}\\
\psi_{A,B,C} & :\begin{cases}
\Dduploid(A,C_{\alpha})\times\Dduploid(B,C_{\alpha})\rightarrow\Dduploid(A\oplus B,C_{\alpha})\\
(e_{1},e_{2})\mapsto\pairing{e_{1}}{e_{2}}\df\mui x{\cut x{e_{1}}}y{\cut y{e_{2}}}
\end{cases}
\end{align*}
\end{prop}

\begin{proof}
~
\begin{itemize}
\item Positivity: (by definition, and in fact) from the observation that
expansion at type $A\oplus B$ implies the linearity of all $e\in\Dduploid(A\oplus B,C)$.
\item $(\psi\circ\varphi)_{A,B,C}=\id_{{\Dduploid(A\oplus B,C)}}$ amounts
to the equality between terms:
\begin{align*}
 & \mui x{\cut x{\mt x.\cut{\ileft(x)}e}}y{\cut y{\mt y.\cut{\iright(y)}e}}\\
 & \argeq R{}\mui x{\cut{\ileft(x)}e}y{\cut{\iright(y)}e}\\
 & \argeq E{}e
\end{align*}
\item $(\varphi\circ\psi)_{A,B,C}=\id_{{\Dduploid(A,C)\times\Dduploid(B,C)}}$
amounts to the equality between terms:
\begin{align*}
 & \mt y.\cut{\iith(y)}{\mui{x_{1}}{\cut{x_{1}}{e_{1}}}{x_{2}}{\cut{x_{2}}{e_{2}}}}\\
 & \argeq R{}\mt y.\cut y{e_{i}}\\
 & \argeq e{}e_{i}
\end{align*}
\item The naturality condition \eqref{eq:sum-nat-1} amounts, for all $e_{1}\in\Dduploid(A,C_{\alpha})$,
$e_{2}\in\Dduploid(B,C_{\alpha})$ and $e_{3}\in\Dduploid(C,D)$,
to the equality between terms:
\begin{align*}
c_{1} & \df\cut{\mu\alpha.\cut x{\mui y{\cut y{e_{1}}}z{\cut z{e_{2}}}}}{e_{3}}\\
c_{2} & \df\cut x{\mui y{\cut{\mu\alpha.\cut y{e_{1}}}{e_{3}}}z{\cut{\mu\alpha.\cut z{e_{2}}}{e_{3}}}}\\
c_{1} & \argeq{RE}{}c_{2}
\end{align*}
when $e_{3}$ is not necessarily linear. It holds thanks to $\argred E{}$:
\begin{align*}
\mt x.c_{1} & \argred E{}\mui y{\cut{\ileft(y)}{\mt x.c_{1}}}z{\cut{\iright(z)}{\mt x.c_{1}}}\\
 & \argReds R{}\mui y{\cut{\mu\alpha.\cut y{e_{1}}}{e_{3}}}z{\cut{\mu\alpha.\cut z{e_{2}}}{e_{3}}}\\
 & \argredinv r{}\mt x.\cut x{\mui y{\cut{\mu\alpha.\cut y{e_{1}}}{e_{3}}}z{\cut{\mu\alpha.\cut z{e_{2}}}{e_{3}}}}\\
 & =\mt x.c_{2}
\end{align*}
\item The naturality condition \eqref{eq:sum-nat-2} amounts, for all $e_{1}\in\Dduploid(A,C_{\alpha})$,
$e_{2}\in\Dduploid(B,C_{\alpha})$, $t\in\Dduploid(A'_{x},A)$, $u\in\Dduploid(B'_{y},B)$,
to the equality between terms:
\begin{align*}
c_{1} & \df\cut{t\oplus_{x,y,z}u}{\pairing{e_{1}}{e_{2}}}\\
c_{2} & \df\cut z{\mui x{\cut t{e_{1}}}y{\cut u{e_{2}}}}\\
c_{1} & \argeq{RE}{}c_{2}
\end{align*}
It holds indeed:
\begin{align*}
\cut{t\oplus_{x,y,z}u}{\pairing{e_{1}}{e_{2}}} & =\cut{\mu\alpha.\cut z{\mui x{\cut{\ileft(t)}{\alpha}}y{\cut{\iright(u)}{\alpha}}}}{\pairing{e_{1}}{e_{2}}}\\
 & \argred r{}\cut z{\mui x{\cut{\ileft(t)}{\pairing{e_{1}}{e_{2}}}}y{\cut{\iright(u)}{\pairing{e_{1}}{e_{2}}}}}\\
 & \argReds{}{}\cut z{\mui x{\cut t{e_{1}}}y{\cut u{e_{2}}}}\\
 & =c_{2}\qedmultiline
\end{align*}
\end{itemize}
\end{proof}

\section{\label{sec:Relative-monads-and}Relative monads and comonads in $\protect\ILLpdiam$}

We prove that the syntactic linear CBPV model $\Vcategory\rightleftarrows\Scategory$
of thunkable and linear terms of $\ILLpdiam$ has a relative comonad
$\oc$ on $\shpos:\Scategory\rightarrow\Vcategory$. A thunkable map
$\oc N\vdash\shpos M$ of the syntactic model is the same thing as
a term with judgement $\oc N\vdash M$ in $\ILLpdiam$. The co-unit
$\varepsilon_{N}$ is thus given by the following covalue: 
\[
\mid\oc\alpha\colon\oc N\vdash\alpha:N
\]
Given any thunkable map $\oc N\vdash\shpos M$, in the form of a command
\[
c\colon(x\colon\oc N\vdash\alpha\colon M)\,,
\]
its coextension $\coextension c$ is the value:
\[
x\colon\oc N\vdash\underbrace{\mu\oc\alpha.c}_{\coextension c}\colon\oc M
\]
which is indeed thunkable by definition. The two unital laws
\begin{align*}
\varepsilon\circ\coextension f & =f & \id_{SB} & =\coextension{\varepsilon_{B}}
\end{align*}
coincide respectively with the following instances of reduction and
expansion:
\begin{align*}
\varp{\cut{\mu\oc\alpha.c}{\oc\alpha}} & \argred R{}c & x & \argred E{}\mu\oc\alpha.\cut x{\oc\alpha}
\end{align*}
The associativity law:
\[
\coextension{(f\circ\coextension g)}=\coextension f\circ\coextension g
\]
instantiated for $c\colon(x\colon\oc N_{1}\vdash\alpha\colon N_{2})$
and $c'\colon(y\colon\oc N_{2}\vdash\beta\colon N_{3})$ corresponds
to 
\begin{align*}
\mu\oc\beta.\varp{\cut{\mu\oc\alpha.c}{\mt\varp y.c'}} & \argeq{RE}{}\mu\varp{\gamma}.\varp{\cut{\mu\oc\alpha.c}{\mt\varp y.\varp{\cut{\mu\oc\beta.c'}{\gamma}}}}
\end{align*}
which indeed holds given that both sides are convertible to $\mu\oc\beta.c'\subst{\mu\oc\alpha.c}y$.

The case of the relative monad $\mon:\Vcategory\rightarrow\Scategory$
on $\shneg:\Vcategory\rightarrow\Scategory$ is symmetric. This establishes
\prettyref{prop:The-syntactic-linear}.
\begin{prop}
A term $x\colon\oc A\vdash t\colon\oc B$ is a $\oc$-coalgebra morphism
in the syntactic model if and only if the following commutation holds
for any $c$:
\[
\varp{\cut t{\mt\varp y.\varp{\cut{\mu\oc\beta.c}{\alpha}}}}\argeq{RE}{}\varp{\cut{\mu\oc\beta.\varp{\cut t{\mt\varp y.c}}}{\alpha}}
\]
Dually, a term $e\colon\mon A\vdash\alpha\colon\mon B$ is a $\mon$-algebra
morphism in the syntactic model if and only if the following commutation
holds for any $c$:
\[
\varn{\cut{\mu\varn{\beta}.\varn{\cut x{\mt\diam y.c}}}e}\argeq{RE}{}\varn{\cut x{\mt\diam y.\varn{\cut{\mu\varn{\beta}.c}e}}}\,.
\]
\end{prop}

These characterisations extend beyond free (co)algebras if we add
other types of (co)algebras to the type system making these commutations
well typed.

\section{\label{sec:Focused-normal-forms}Focused normal forms for $\protect\ILLpdiam$}

\subsection{\label{sec:Strong-normalisation}Strong normalisation}

\global\long\def\TT{\mathbb{T}}\global\long\def\EE{\mathbb{E}}\global\long\def\Bot{\mathord{\Vbar}}\global\long\def\polar{\mathrel{\Bot}}\global\long\def\polarp{\mathrel{\Bot^{+}}}\global\long\def\polarn{\mathrel{\Bot^{\moins}}}\global\long\def\orth#1{{{#1}^{\bot}}}\global\long\def\biorth#1{{{#1}^{\bot\bot}}}\global\long\def\porth#1{{{#1}^{\termp{\bot}}}}\global\long\def\pbiorth#1{{{#1}^{\termp{\bot}\termp{\bot}}}}\global\long\def\north#1{{{#1}^{\termn{\bot}}}}\global\long\def\nbiorth#1{{{#1}^{\termn{\bot}\termn{\bot}}}}\global\long\def\interp#1{\big|#1\big|}\global\long\def\interpv#1{\mathopen{\Vert}#1\mathclose{\Vert}}\global\long\def\powerset#1{\mathcal{P}(#1)}\global\long\def\valset#1{#1_{\mathbb{V}}}\global\long\def\stackset#1{#1_{\mathbb{S}}}

This section is taken from \citep{Munch-Maccagnoni2017curry} with
minor adaptations, which was itself inspired by Krivine's adaptation
\citep{Kri93,Kri04Realizability} of Girard's reducibility candidates
\citep{Gir72,Gir87}.

Judgements refer to typability in $\ILLpdiam$ (or any extension with
additional structural rules, since the restriction on structural rules
is not accounted for by the model).
\begin{prop}[Weak standardization]
\label{prop:weak-std}If $c\mathrel{(\argRed R{}\setminus\argred R{})}^{*}c'\argred R{}c''$
then there exists $c'''$ such that $c\argred R{}c'''\argReds R{}c''$.
\end{prop}

\begin{proof}
Suppose one has $c\mathrel{(\argRed R{}\setminus\argred R{})}^{*}c'\argred R{}c''$.
We consider the rule involved in the reduction $c'\argred R{}c''$:
if it is $(R\otimes)$ then $c'$ is of the form $\cut{V'\smallotimes W'}{\mt(x\smallotimes y).c_{0}'}$
and $c''=c_{0}'\subst{V'}{x\addsubst{W'}y}$. Now $\argRed R{}\setminus\argred R{}$
can only create strict subcommands, therefore $c$ is of the form
$\cut{V\smallotimes W}{\mt(x\smallotimes y).c_{0}}$ where $V\argReds R{}V'$,
$W\argReds R{}W'$, and $c_{0}\argReds R{}c_{0}'$. Therefore by taking
$c'''=c_{0}\subst V{x\addsubst Wy}$ one has $c\argred R{}c'''\argReds R{}c'_{0}$.
Similarly for all other reduction rules: no redex for $\argred R{}$
can be created by $\argRed R{}\setminus\argred R{}$ and one can apply
the same reasoning.
\end{proof}

\begin{defn}
A command $c$ is strongly normalizing ($c\in\Bot$) if any $\argRed R{}$-sequence
starting from $c$ is finite. An expression $t$ (respectively a context
$e$) is strongly normalizing ($t\in\TT$, resp. $e\in\EE$) if any
$\vareps{\cut t{\alpha}}{}$ (resp. any $\vareps{\cut xe}{}$) is
in $\Bot$.
\end{defn}

\begin{lem}
\label{lem:comm-sm}Let $c$ be a command. If all strict sub-commands
of $c$ are in $\Bot$, and if either $c\argnred R{}$ or $c\argred R{}c'\in\Bot$
for some command $c'$, then $c\in\Bot$.
\end{lem}

\begin{proof}
Assume that $c$ admits an infinite $\argRed R{}$-sequence. If this
sequence is an infinite ($\argRed R{}\setminus\argred R{}$)-sequence,
then at least one immediate sub-command of $c$ (there are finitely
many so) admits an infinite $\argRed R{}$-sequence. Otherwise, let
us write $c'\argred R{}c''$ the first occurrence of $\argred R{}$
in the sequence. By \prettyref{prop:weak-std} there exists $c'''$
with $c\argred R{}c'''\argReds R{}c''$ and therefore $c'''$ admits
an infinite $\argRed R{}$-sequence. In this case, since $\argred R{}$
is deterministic, this means that all $\argred R{}$-reducts admit
an infinite $\argRed R{}$-sequence.
\end{proof}

\begin{lem}
\label{lem:term-sn}One has $t\in\TT$ (respectively $e\in\EE$) if
and only if all sub-commands of $t$ (resp. $e$) are in $\Bot$.
\end{lem}

\begin{proof}
($\Rightarrow$) Immediate. ($\Leftarrow$) Assume that all sub-commands
of $t$ are in $\Bot$. We apply \prettyref{lem:comm-sm} on the command
$\vareps{\cut t{\alpha}}{}$. If there exists $c$ with $\vareps{\cut t{\alpha}}{}\argred R{}c$,
then one has $t=\mu\vareps{\beta}{}.c'$ with $c=c'\subst{\alpha}{\beta}$.
Now $c'$, being a sub-command of $t$, is in $\Bot$, and since the
rewriting rules are left-linear, variable substitution cannot unlock
new reductions. Therefore $c\in\Bot$ as required. The proof is similar
for $e$.
\end{proof}

Notice that the direction ($\Leftarrow$) relies on eliminations (e.g.
function application $tu$) being expressed with a $\mu$ binder.
In calculi where $tu$ is given as a primitive, it is not possible
to state that $\cut t{u\push\alpha}$ is a sub-command of $tu$!
\begin{prop}[Saturation]
\label{prop:saturation}One has $\cut te\in\Bot$ if and only if
$t\in\TT$, $e\in\EE$, and either $\cut te\argnred R{}$ or $\cut te\argred R{}c\in\Bot$.
\end{prop}

\begin{proof}
By \prettyref{lem:comm-sm}~and~\prettyref{lem:term-sn}.
\end{proof}

As usual, we let $\Bot$ define an antitone Galois correspondence
$\orth{\cdot}\dashv\orth{\cdot}:\powerset{\TT}\rightarrow{\powerset{\EE}^{\op}}$.
In fact, we define two such Galois correspondences:
\begin{align*}
\parens{\porth{\cdot}} & \dashv\parens{\porth{\cdot}}:\powerset{\TT}\rightarrow{\powerset{\mathbb{S}}^{\op}}\\
\parens{\north{\cdot}} & \dashv\parens{\north{\cdot}}:\powerset{\mathbb{V}}\rightarrow{\powerset{\EE}^{\op}}
\end{align*}
by the following two poles:
\begin{align*}
t\polarp S & \iff\varp{\cut tS}\in\Bot\\
V\polarn e & \iff\varn{\cut Ve}\in\Bot
\end{align*}

\global\long\def\VV{\mathbb{V}}\global\long\def\SS{\mathbb{S}}

\begin{defn}
~
\begin{itemize}
\item $\VV\subseteq\TT$ is the set of strongly normalising values.
\item $\SS\subseteq\mathbb{E}$ is the set of strongly normalising stacks.
\item We write $\valset{\mathcal{X}}\df\mathcal{X}\cap\VV$ and $\stackset{\mathcal{X}}\df\mathcal{X}\cap\SS$.
\item We consider the following maps:
\begin{align*}
\mathcal{V}\otimes\mathcal{W} & \df\setbin{V\smallotimes W}{V\in\mathcal{V}\text{ and }W\in\mathcal{W}}\\
\mathcal{V}_{1}\oplus\mathcal{V}_{2} & \df\setbin{\iota_{i}(V)}{V\in\mathcal{V}_{i}}\\
\oc\mathcal{S} & \df\setbin{\oc S}{S\in\mathcal{S}}\\
\mathcal{V}\bullet\mathcal{S} & \df\setbin{V\push S}{V\in\mathcal{V}\text{ and }S\in\mathcal{S}}\\
\mathcal{S}_{1}\with\mathcal{S}_{2} & \df\setbin{\pi_{i}\push S}{S\in\mathcal{S}_{i}}\\
\mon\mathcal{V} & \df\setbin{\diam V}{V\in\mathcal{V}}
\end{align*}
\end{itemize}
\end{defn}

\begin{lem}
These maps preserve strong normalisation.
\end{lem}

\begin{proof}
This follows immediately from \prettyref{lem:term-sn}.
\end{proof}

\begin{defn}
We define by mutual induction interpretations of types:
\begin{align*}
\interpv P & \subseteq\VV & \interpv N & \subseteq\SS\\
\TT(A) & \subseteq\TT & \EE(A) & \subseteq\EE
\end{align*}
with:
\begin{align*}
\TT(P) & \df\pbiorth{\interpv P} & \TT(N) & \df\north{\interpv N}\\
\EE(P) & \df\porth{\interpv P} & \EE(N) & \df\nbiorth{\interpv N}\\
\interpv{\varp X} & \df\VV & \interpv{\varn X} & \df\SS\\
\interpv{\unite} & \df\{()\} & \interpv{A\mixarrow B} & \df\valset{\TT(A)}\bullet\stackset{\EE(B)}\\
\interpv{A\otimes B} & \df\valset{\TT(A)}\otimes\valset{\TT(B)} & \interpv{A\with B} & \df\stackset{\EE(A)}\with\stackset{\EE(B)}\\
\interpv{\oc A} & \df\valset{\porth{(\prom{\stackset{\EE(A)}})}} & \interpv{\mon A} & \df\stackset{\north{(\mon\valset{\TT(A)})}}\\
\interpv{A\oplus B} & \df\valset{\TT(A)}\oplus\valset{\TT(B)}\\
\interpv{\falso} & \df\emptyset & \interpv{\top} & \df\emptyset
\end{align*}
\end{defn}

From now on it will be clear from the polarity of the formula whether
$\orth{\cdot}$ means $\porth{\cdot}$ or $\north{\cdot}$.

\begin{lem}
\label{lem:gen}For all $A$ one has:
\begin{align*}
\TT(A) & =\orth{\stackset{\EE(A)}} & \EE(A) & =\orth{\valset{\TT(A)}}\,.
\end{align*}
\end{lem}

\begin{proof}
If $A$ is positive, then $\stackset{\EE(\termp A)}=\EE(\termp A)$
and $\TT(A)=\orth{\EE(A)}$ by definition; if it is negative then
we prove $\TT(\termn A)=\orth{\stackset{\EE(\termn A)}}$ by inclusion.
($\supseteq$) One has $\interpv{\termn A}\subseteq\stackset{\EE(\termn A)}$,
hence $\TT(\termn A)=\orth{\interpv{\termn A}}\supseteq\orth{\stackset{\EE(\termn A)}}$.
($\subseteq$) One has $\stackset{\EE(\termn A)}\subseteq\EE(\termn A)=\biorth{\interpv{\termn A}}$,
hence $\TT(\termn A)=\orth{\interpv{\termn A}}\subseteq\orth{\stackset{\EE(\termn A)}}$.
The reasoning for $\EE(A)=\orth{\valset{\TT(A)}}$ is symmetric.
\end{proof}

\begin{lem}
\label{lem:id}Let $A$ be a formula. For all $x$ one has $x\in\valset{\TT(A)}$
and for all $\alpha$ one has $\alpha\in\stackset{\EE(A)}$.
\end{lem}

\begin{proof}
This is immediate from the definitions of $\TT$ and $\EE$.
\end{proof}

\begin{defn}
For $\Gamma=(x_{1}\colon A_{1},\dots,x_{n}\colon A_{n})$ and $\Delta=(\alpha\colon B)$
(possibly empty), and any substitution $\sigma:\{x_{1},\dots,x_{n},\alpha\}\rightarrow\VV\cup\SS$,
we write $\sigma\Vdash\Gamma,\Delta$ whenever $\forall i,\sigma(x_{i})\in\valset{\TT(A_{i})}$
and $\sigma(\alpha)\in\stackset{\EE(A)}$.
\end{defn}

\begin{thm}
\label{lem:adequacy}Consider $\sigma\Vdash\Gamma,\Delta$. One has:
\begin{itemize}
\item if $c\colon(\Gamma\vdash\Delta)$ then $c[\sigma]\in\Bot$,
\item if $\Gamma\vdash t\colon A$ then $t[\sigma]\in\TT(A)$,
\item if $\Gamma\mid e\colon A\vdash\Delta$ then $e[\sigma]\in\EE(A)$.
\end{itemize}
\end{thm}

\begin{proof}
This is proved by induction on the derivations.

\emph{Rules ($\vdash\mathit{ax}$) and ($\mathit{ax}\vdash$):} immediate.\emph{
Rule ($\vdash\vareps{\mu}{}$):} One has to show $\mu\vareps{\alpha}{}.c[\sigma]\in\TT(\termeps A{})$
for some $\sigma\Vdash\Gamma$ with $\alpha\not\in\dom{\Gamma}$.
By \prettyref{lem:gen}, it suffices to show $\vareps{\cut{\mu\vareps{\alpha}{}.c[\sigma]}S}{}\in\Bot$
for all $S\in\stackset{\EE(\termeps A{})}$. By \prettyref{prop:saturation},
this follows from $c[\sigma\addsubst S{\alpha}]\in\Bot$ which follows
from the induction hypothesis. \emph{Rule ($\vareps{\mt}{}\vdash$):}
same reasoning.

\emph{Rule (cut):} One has to show $\vareps{\cut te}{}[\sigma]\in\Bot$
for any $\sigma\Vdash\Gamma,\Gamma',\Delta$. We consider the restrictions
$\sigma'\Vdash\Gamma$ and $\sigma''\Vdash\Gamma',\Delta$ of $\sigma$.
By induction hypothesis, one has $t[\sigma']\in\TT(\termeps A{})$
and $e[\sigma'']\in\EE(\termeps A{})$, hence the result from $\TT(\termeps A{})\mathrel{\vareps{\Bot}{}}\EE(\termeps A{})$.

\emph{Rules ($\sigma$), ($\vdash\sigma$) and ($\sigma\vdash$) for
$\sigma\in\Substexp{\Gamma}{\Gamma'}$.} For any $\sigma'\Vdash\Gamma'$
one has $(\sigma'\circ\sigma)\Vdash\Gamma$, from which the result
follows by induction.

\emph{Rule ($\mixarrow\focvdash$):} One has to show $V\push S[\sigma]\in\valset{\TT(A)}\bullet\stackset{\EE(B)}$
for any $\sigma\Vdash\Gamma,\Gamma',\Delta$. We consider the restrictions
$\sigma'\Vdash\Gamma$ and $\sigma''\Vdash\Gamma',\Delta$ of $\sigma$.
By induction, one has $V[\sigma']\in\valset{\TT(A)}$ and $S[\sigma'']\in\stackset{\EE(B)}$,
hence the result. \emph{Same reasoning for the rules ($\focvdash\otimes$),
($\with_{i}\focvdash$), and ($\ensuremath{\focvdash\oplus_{i}}$).}

\emph{Rule ($\vdash\with$):} the goal is to prove $\mupair{\alpha}c{\beta}{c'}\in\TT(A\with B)$
for $\sigma\Vdash\Gamma$ with $\alpha,\beta\notin\dom{\Gamma}$.
By definition it is enough to prove $\varn{\cut{\mupair{\alpha}c{\beta}{c'}[\sigma]}{\pith\push S_{i}}}\in\Bot$
for both $S_{1}\in\stackset{\EE(A)}$ and $S_{2}\in\stackset{\EE(B)}$.
By \prettyref{prop:saturation}, this follows from $\mupair{\alpha}c{\beta}{c'}[\sigma]\in\TT$,
$c[\sigma\addsubst{S_{1}}{\alpha}]\in\Bot$ and $c'[\sigma\addsubst{S_{2}}{\beta}]\in\Bot$.
The latter two follow from induction hypothesis. The former follows
from \prettyref{lem:term-sn} and $c[\sigma],c'[\sigma]\in\Bot$ obtained
by the induction hypothesis by substituting $\alpha$ and $\beta$
by themselves by \prettyref{lem:id}. \emph{Same reasoning for the
rules ($\vdash\mathord{\mixarrow}$), ($\otimes\vdash$), ($\oplus\vdash$).}

\emph{Rule ($\ensuremath{\focvdash\top}$):} one has to show $\term V[\sigma]\in\north{\emptyset}=\mathbb{V}$.
By \prettyref{lem:term-sn}, this follows from $V[\sigma]\in\mathbb{V}$
which follows by the induction hypothesis. \emph{Same reasoning for
the rule ($\falso\ensuremath{\focvdash}$).}

\emph{Rule ($\ensuremath{\vdash\oc}$):} we show $\mu\prom{\alpha}.c[\sigma]\in\interpv{\oc A}=\valset{\porth{(\prom{\stackset{\EE(A)}})}}$.
We show for any $S\in\stackset{\EE(A)}$, $\varp{\cut{\mu\prom{\alpha}.c[\sigma]}{\oc S}}\in\Bot$.
By \prettyref{prop:saturation} this follows from $c[\sigma\addsubst S{\alpha}]\in\Bot$
which follows from induction hypothesis. \emph{Same reasoning for
the rule ($\mon\vdash$).}

\emph{Rule ($\ensuremath{\oc\focvdash}$):} the goal is to show $\oc S[\sigma]\in\porth{\valset{\porth{(\prom{\stackset{\EE(A)}})}}}$.
One has $\oc S[\sigma]\in\prom{\stackset{\EE(A)}}$ by induction hypothesis,
and $\prom{\stackset{\EE(A)}}\subseteq\porth{\valset{\porth{(\prom{\stackset{\EE(A)}})}}}$
follows from $\valset{\porth{(\prom{\stackset{\EE(A)}})}}\subseteq\porth{\prom{\stackset{\EE(A)}}}$.
\emph{Same reasoning for the rule ($\ensuremath{\focvdash\mon}$).}
\end{proof}

\begin{thm}[Strong normalization]
\label{thm:SN}Any typable term is strongly normalizing.
\end{thm}

\begin{proof}
For $c\colon(\Gamma\vdash\Delta)$, $\Gamma\vdash t\colon A$ and
$\Gamma\mid e\colon A\vdash\Delta$, we show $c\in\Bot$, $t\in\TT$
and $e\in\EE$. This follows from \prettyref{lem:adequacy} applied
with the identity substitution, which indeed satisfies $\id\Vdash\Gamma,\Delta$
by \prettyref{lem:id}.
\end{proof}

\subsection{\label{subsec:compl-Focusing}Completeness of focusing}

We describe focused normal forms \citep{Andreoli92,Lau04FocLL,Liang2009a}
and prove their completeness, strengthened with respect to the interpretation:
for any proof there is an \emph{equivalent} focused proof. This section
is taken from \citet{Munch-Maccagnoni2017curry} with slight adaptations.

In this section,
\begin{itemize}
\item sequents are measured by the number of connectives and units,
\item terms are measured by the number of subterms of the form $\varp{\cut V{\alpha}}$,
$\varn{\cut xS}$, $\oc S$ or $\diam V$ in their $\argRed R{}$-normal
form if it exists (in which case it is unique by confluence), otherwise
it is infinite by convention.
\end{itemize}
\begin{defn}
We define an \emph{inversion} relation $\succ$ between sequents and
multisets of sequents:
\begin{align*}
(\Gamma\vdash A\mixarrow B) & \succ\{(\Gamma,A\vdash B)\}\\
(\Gamma,A\otimes B,\Gamma'\vdash\Delta) & \succ\{(\Gamma,A,B,\Gamma'\vdash\Delta)\}\\
(\Gamma,\unite,\Gamma'\vdash\Delta) & \succ\{(\Gamma,\Gamma'\vdash\Delta)\}\\
(\Gamma\vdash A\with B) & \succ\{(\Gamma\vdash A),(\Gamma\vdash B)\}\\
(\Gamma,A\oplus B,\Gamma'\vdash\Delta) & \succ\{(\Gamma,A,\Gamma'\vdash\Delta),(\Gamma,B,\Gamma'\vdash\Delta)\}\\
(\Gamma\vdash\top) & \succ\emptyset\\
(\Gamma,\falso,\Gamma'\vdash\Delta) & \succ\emptyset
\end{align*}
A sequent that is normal for $\succ$ is called \emph{inverted}.
\end{defn}

In other words, a sequent $\Gamma\vdash\Delta$ is inverted\emph{
}if:
\begin{itemize}
\item $\Gamma$ only contains formulae that are either negative or of the
form $\oc A$ or $\varp X$; and
\item $\Delta$ contains either a positive formula, or a formula of the
form $\varn X$ or $\mon A$.
\end{itemize}
\begin{lem}
\label{lem:The-extension-of}The extension of $\succ$ to a relation
between multisets of sequents is terminating and confluent.
\end{lem}

\begin{proof}
The relation is strictly decreasing for the induced multiset order,
therefore it is terminating. It is confluent by Newman's lemma because
it is locally confluent.
\end{proof}

\begin{figure}
\hspace*{\fill}
\begin{varwidth}{\linewidth}
\[
\addcustomstretch[2.8]{\begin{array}{cc}
\unrulecsepdouble{c\colon(\Gamma,x\colon A}{\alpha\colon B)}{\cut{\mu(x\push\alpha).c}{\beta}\colon(\Gamma}{\beta\colon A\mixarrow B)}{} & \binruledouble{c\colon(\Gamma}{\alpha\colon A)}{c'\colon(\Gamma}{\beta\colon B)}{\cut{\mupair{\alpha}c{\beta}{c'}}{\gamma}\colon(\Gamma}{\gamma\colon A\with B)}{}\\
\unrulecsepdouble{c\colon(\Gamma,x\colon A,y\colon B,\Gamma'}{\Delta)}{\cut z{\mt(x\smallotimes y).c}\colon(\Gamma,z\colon A\otimes B,\Gamma'}{\Delta)}{} & \binruledouble{c\colon(\Gamma,x\colon A,\Gamma'}{\Delta)}{c'\colon(\Gamma,y\colon B,\Gamma'}{\Delta)}{\cut z{\mui xcy{c'}}\colon(\Gamma,z\colon A\oplus B,\Gamma'}{\Delta)}{}\\
\unrulecsepdouble{c\colon(\Gamma,\Gamma'}{\Delta)}{\cut x{\mt().c}\colon(\Gamma,x\colon\unite,\Gamma'}{\Delta)}{} & \axrulecdouble{\cut x{\init{}_{\Gamma,\Gamma',\Delta}}\colon(\Gamma,x\colon\falso,\Gamma'\vdash\Delta)}{}\\
 & \axrulecdouble{\cut{\term{}_{\Gamma}}{\alpha}\colon(\Gamma\vdash\alpha\colon\top)}{}
\end{array}}
\]
\end{varwidth}
\hspace*{\fill}

\hlinefill{}

\caption{\label{fig:Inversion}Inversion}
\hspace*{\fill}
\begin{varwidth}{\linewidth}
\[
\addcustomstretch[2.6]{\begin{array}{c}
\unrulecsepdouble{c\colon(\Gamma}{\alpha\colon N)}{\Gamma}{\mu\varn{\alpha}.c\colon N\mid}{}\qquad\unrulecsepdouble{c\colon(\Gamma,x\colon P}{\Delta)}{\Gamma\mid\mt\varp x.c\colon P}{\Delta}{}\qquad\unrulecsepdouble{c\colon(\oc\Gamma}{\alpha\colon A)}{\oc\Gamma}{\mu\prom{\alpha}.c\colon\oc A\mid}{}\qquad\unrulecsepdouble{c\colon(\oc\Gamma,x\colon A}{\mon\Delta)}{\oc\Gamma\mid\mt\diam x.c\colon\mon A}{\mon\Delta}{}\\
\axruledouble{!\Gamma,x\colon\varp X}{x\colon\varp X\mid}{}\qquad\axruledouble{!\Gamma,x\colon\oc A,!\Gamma'}{x\colon\oc A\mid}{}\qquad\axruledouble{!\Gamma\mid\alpha\colon\varn X}{\alpha\colon\varn X}{}\\
\binruledouble{\oc\Gamma',\restr{\Gamma}{\fv V}}{V\colon A\mid}{\oc\Gamma',\restr{\Gamma}{\fv W}}{W\colon B\mid}{\oc\Gamma',\Gamma}{V\smallotimes W\colon A\otimes B\mid}{}\qquad\binruledouble{\oc\Gamma',\restr{\Gamma}{\fv V}}{V\colon A\mid}{\oc\Gamma',\restr{\Gamma}{\fv S}\mid S\colon A}{\Delta}{\oc\Gamma',\Gamma\mid V\push S\colon A\mixarrow B}{\Delta}{}\\
\axruledouble{\oc\Gamma}{()\colon\unite\mid}{}\qquad\unrulecsepdouble{\Gamma}{V\colon A_{i}\mid}{\Gamma}{\iith(V)\colon A_{1}\oplus A_{2}\mid}{}\qquad\unrulecsepdouble{\Gamma\mid S\colon A_{i}}{\Delta}{\Gamma\mid\pith\push S\colon A_{1}\with A_{2}}{\Delta}{}
\end{array}}
\]
\end{varwidth}
\hspace*{\fill}

\hlinefill{}

\caption{\label{fig:Focusing}Focusing in $\protect\ILLpdiam$}
\end{figure}

\begin{lem}
\label{lem:For-any-term}Let $f$ be any term. Then $f\subst{x\smallotimes y}z$,
$f\subst{\iith(x)}y$, $f\subst{()}x$, $f\subst{x\push\alpha}{\beta}$,
and $f\subst{\pith\push\alpha}{\beta}$ are smaller than $f$.
\end{lem}

\begin{proof}
Note that if $f$ is not $\argRed R{}$-normalisable then the result
is trivial by convention. For any term $f$ let us write $\size f$
the number of subterms of the form $\varp{\cut V{\alpha}}$, $\varn{\cut xS}$,
$\oc S$ or $\diam V$, and $\normalise f$ the $\argRed R{}$-normal
form of $f$ if it exists. In particular, if $f$ is normalisable
then its size is given by $\size{\normalise f}$. We now consider
$f\subst p{\kappa}$ as in the statement. With these notations the
result is stated $\size{\normalise{f\subst p{\kappa}}}\leq\size{\normalise f}$.

We first prove the result for $f$ $\argRed R{}$-normal, that is
$\size{\normalise{f\subst p{\kappa}}}\leq\size f$. Since $f$ is
$\argRed R{}$-normal, the only redexes in $f\subst p{\kappa}$ are
the ones created by the substitution. Necessarily, those are of the
form $\varn{\cut{\mu p.C}p}$ or $\varp{\cut p{\mt p.C}}$ (up to
renaming of bound variables), and the reduction just selects a field
of $C$, thus decreasing $\size{\cdot}$ without creating new redexes.
Iterating the argument by induction on $f\subst p{\kappa}$ to eliminate
all the redexes in this way, we obtain its $\argRed R{}$-normal form\footnote{It is in fact the parallel reduct of the redexes created by the substitution,
and also the complete development of the substitution, but it is not
necessary to use these notions here.} and it verifies $\size{\normalise{f\subst p{\kappa}}}\leq\size{f\subst p{\kappa}}$.
We conclude by noting $\size{f\subst p{\kappa}}\leq\size f$ since
substitution can only decrease the number of commands of the form
$\varp{\cut V{\alpha}}$ or $\varn{\cut xS}$.

Now, in the case of $f$ $\argRed R{}$-normalisable, we have just
proved that $\normalise{(\normalise f\subst p{\kappa})}$ exists and
$\size{\normalise{(\normalise f\subst p{\kappa})}}\leq\size{\normalise f}$.
But $\normalise{(\normalise f\subst p{\kappa})}=\normalise{f\subst p{\kappa}}$
via the reduction $f\subst p{\kappa}\argReds R{}\normalise f\subst p{\kappa}\argReds R{}(\normalise{\normalise f\subst p{\kappa})}$
and uniqueness of the normal form. 
\end{proof}

\begin{lem}
\label{lem:For-any-typable}For any $c\colon\Psi$, and for any $\{\Psi_{1},\dots,\Psi_{n}\}\prec\Psi$,
there exists an equivalent command $c'\colon\Psi$ derived from smaller
commands of type $\Psi_{1},\dots,\Psi_{n}$ by application of a rule
in Fig.~\ref{fig:Inversion}.
\end{lem}

\begin{proof}
The result follows by mapping each inversion as follows:
\begin{align*}
c\colon(\Gamma,z\colon A\otimes B,\Gamma'\vdash\Delta) & \succ\{\cut{x\smallotimes y}{\mt z.c}\colon(\Gamma,x\colon A,y\colon B,\Gamma'\vdash\Delta)\}\\
c\colon(\Gamma,z\colon A\oplus B,\Gamma'\vdash\Delta) & \succ\{\cut{\ileft(x)}{\mt z.c}\colon(\Gamma,x\colon A,\Gamma'\vdash\Delta),\cut{\iright(y)}{\mt z.c}\colon(\Gamma,y\colon B,\Gamma'\vdash\Delta)\}\\
c\colon(\Gamma,x\colon\unite,\Gamma'\vdash\Delta) & \succ\{\cut{()}{\mt x.c}\colon(\Gamma,\Gamma'\vdash\Delta)\}\\
c\colon(\Gamma\vdash\beta\colon A\mixarrow B) & \succ\{\cut{\mu\beta.c}{x\push\alpha}\colon(\Gamma,x\colon A\vdash\alpha\colon B)\}\\
c\colon(\Gamma\vdash\gamma\colon A\with B) & \succ\{\cut{\mu\gamma.c}{\pleft\push\alpha}\colon(\Gamma\vdash\alpha\colon A),\cut{\mu\gamma.c}{\pright\push\beta}\colon(\Gamma\vdash\beta\colon B)\}\\
c\colon(\Gamma,x\colon\falso,\Gamma'\vdash\Delta) & \succ\emptyset\\
c\colon(\Gamma\vdash\alpha\colon\top) & \succ\emptyset
\end{align*}
The commands on the right-hand side are well-typed by straightforward
derivation. They are smaller by application of \prettyref{lem:For-any-term}.
When composed with the respective constructions of Fig.~\ref{fig:Inversion},
one obtains a command $c'\argeq{RE}{}c$ and thus equivalent.
\end{proof}

\begin{prop}[Inversion]
\label{prop:For-any-typable-1}For any sequent $\Psi$, there exist
inverted sequents $\Psi_{1},\dots,\Psi_{n}$ such that any command
$c\colon\Psi$ can up to equivalence be derived from smaller commands
$c_{1}\colon\Psi_{1},\dots,c_{n}\colon\Psi_{n}$ by applying the rules
in Fig.~\ref{fig:Inversion} in any order.
\end{prop}

\begin{proof}
By Lemmas~\ref{lem:The-extension-of} and~\ref{lem:For-any-typable}.
\end{proof}

\begin{prop}[Focusing]
\label{prop:For-any-typable}For any typable command $c\colon(\Gamma\vdash\Delta)$
where $\Gamma\vdash\Delta$ is inverted, there exist equivalent command
and derivation of either of the following forms:
\begin{align*}
 & \unrulecsepdouble{\oc\Gamma',\Gamma''}{V\colon\termp{\Delta(\alpha)}\mid}{\varp{\cut V{\alpha}}\colon(\Gamma}{\Delta)}{} &  & \unrulecsepdouble{\oc\Gamma',\Gamma''\mid S\colon\termn{\Gamma(x)}}{\Delta}{\varn{\cut xS}\colon(\Gamma}{\Delta)}{} &  & \unrulecsepdouble{\oc\Gamma',\Gamma''\mid S\colon\Gamma'(x)}{\Delta}{\varp{\cut x{\oc S}}\colon(\Gamma}{\Delta)}{} &  & \unrulecsepdouble{\oc\Gamma',\Gamma''}{V\colon\Delta(\alpha)}{\varn{\cut{\diam V}{\alpha}}\colon(\Gamma}{\Delta)}{}
\end{align*}
where $\oc\Gamma'$ is the restriction of $\Gamma$ to formulae of
the form $\oc A$, and $\Gamma''$ is obtained from $\Gamma$ by
removing $\oc\Gamma'$ and $x$ (if applicable). Moreover, $S$ and
$V$ derive from zero or more commands strictly smaller than $c$
by applications of rules in Fig.~\ref{fig:Focusing}.
\end{prop}

\begin{proof}
Up to equivalence, one can assume $c$ $\argRed R{}$-normal: indeed,
the normal form exists by \prettyref{thm:SN}, its derivation is equivalent
by soundness, and it has the same size by definition. Then one has
either $c=\cut x{S^{\Gamma(x)}}$ or $c=\cut{V^{\Delta(\alpha)}}{\alpha}$
by generation lemma, with $S$ and $V$ $\argRed R{}$-normal. We
sort them into one of three cases: either $\varp{\cut{V'}{\alpha}}$,
$\varn{\cut x{S'}}$, or $\varp{\cut x{\oc S'}}$, derived as in the
above statement. If $\Gamma(x)$ is negative or $\Delta(\alpha)$
is positive, we are in one of the first two cases. Since $\Gamma\vdash\Delta$
is inverted, only three cases remain: $\Gamma(x)=\varp X$, $\Gamma(x)=\oc A$,
$\Delta(\alpha)=\varn X$, and $\Delta(\alpha)=\mon A$. By generation
lemma, one respectively has $S=\alpha$ (first case), $S=\oc S'$
(third case), $V=x$ (second case), and $V=\diam V'$ (fourth case).
Notice that $V'$ and $S'$ are strictly smaller than $c$.  Then
the result follows by an induction on $S'$ and $V'$ typable in an
inverted sequent, by generation lemma, and applying weakening on the
hypotheses. The base cases are
\begin{align*}
 & \Gamma\mid\alpha\colon N\vdash\Delta &  & \Gamma\vdash x\colon P\mid &  & \oc\Gamma\vdash\mu\prom{\alpha}.c\colon\oc A\mid &  & \oc\Gamma\mid\mt\diam x.c\colon\mon A\vdash\mon\Delta &  & \Gamma\mid S_{0}\colon P\vdash\Delta &  & \Gamma\vdash V_{0}\colon N\mid
\end{align*}
In the first two cases, one has indeed $P=\Gamma(x)$ of the form
$\varp X$ or $\oc A$ and $N=\Delta(\alpha)$ of the form $\varn X$
or $\mon A$ by inversion. In the last three cases, $\mu\prom{\alpha}.c$,
$\mt\diam x.c$, $V_{0}$, and $S_{0}$ are strictly smaller than
$c$ by being subterms of $V'$ or $S'$. Then in the last two cases
we apply a $\varp{\mt}$- or a $\varn{\mu}$-expansion to obtain equivalent
terms $\mt\varp x.c'$ and $\mu\varn{\alpha}.c'$. The expansion preserves
the size, so $c'$ is strictly smaller than $c$.
\end{proof}

Thus, for any proof, there is an equivalent proof obtained by alternating
one inversion phase with one or more\footnote{in the case where a rule $(\vdash\oc)$ or $(\mon\vdash)$comes from
a hypothesis that is already inverted.} focusing phases.

\end{document}